\RequirePackage{fix-cm}
\documentclass[smallextended]{svjour3}       
\smartqed  
\usepackage{graphicx}
\usepackage{amsmath}
\usepackage{amsfonts,dsfont}
\usepackage[]{algorithm2e}
\usepackage{multirow}
\usepackage{mathtools}
\usepackage{amssymb,stmaryrd,mathtools}
\usepackage[dvipsnames]{xcolor}
\usepackage{float}
\usepackage[utf8]{inputenc}
\usepackage[normalem]{ulem}

\newtheorem{remarque}{Remark}
\newtheorem{hypothesis}{Hypothesis}

\newcommand{\myspecialendproof}{\begin{flushright}$\square$\end{flushright}}
\newcommand\taulinear{\tau_L}
\newcommand\taup{\tau}
\newcommand\cextin{C_{\tau}}

\usepackage{epstopdf}
\AppendGraphicsExtensions{.tif}

\usepackage{xargs}       
\usepackage[colorinlistoftodos,prependcaption,textsize=tiny]{todonotes}
\begin{document}

\title{Stochastic nonlinear model for somatic cell population dynamics during ovarian follicle activation
}
\titlerunning{Stochastic nonlinear cell population model for ovarian follicle activation}        

\author{Fr\'{e}d\'{e}rique Cl\'{e}ment    \and
        Fr\'{e}d\'{e}rique Robin \and Romain Yvinec  
}


\institute{Fr\'{e}d\'{e}rique Cl\'{e}ment \at
            Inria, Centre de recherche Inria Saclay-Île-de-France\\
              \email{frederique.clement@inria.fr}           
           \and
           Fr\'{e}d\'{e}rique Robin \at
           Inria, Centre de recherche Inria Saclay-Île-de-France\\
           \email{frederique.robin@inria.fr}    \and 
            Romain Yvinec \at
           PRC, INRAE, CNRS, Universit\'{e} de Tours, 37380 Nouzilly, France \\
           \email{romain.yvinec@inrae.fr}    
}

\date{Received: date / Accepted: date}

\maketitle

\begin{abstract}

In mammals, female germ cells are sheltered within somatic structures called ovarian follicles, which remain in a quiescent state until they get activated, all along reproductive life. We investigate the sequence of somatic cell events occurring just after follicle activation, {starting by the awakening of precursor somatic cells, and their transformation into proliferative cells}. We introduce a nonlinear stochastic model accounting for the joint dynamics of {the} two cell types, {and allowing us to investigate the potential impact of a feedback from proliferative cells onto precursor cells}. 
{To tackle the key issue of whether cell proliferation is concomitant or posterior to cell awakening, we} assess both the time needed for all precursor cells to {awake}, and the corresponding increase in the {total} cell number with respect to the initial cell number. Using the probabilistic theory of first passage times, we design a numerical scheme based on a rigorous Finite State Projection and coupling techniques to compute the mean extinction time and the cell number at extinction time. {We find that the feedback term clearly lowers the number of proliferative cells at the extinction time.} We calibrate the model parameters using an exact likelihood approach. We carry out a comprehensive comparison between the initial model and a series of submodels, which helps to select the critical cell events taking place during activation, {and suggests that awakening is prominent over proliferation}.



\end{abstract}

\section{Introduction}\label{intro}
    In mammals, the number of oocytes (egg cells) available for a female throughout her reproductive life is fixed once for all, during the fetal or perinatal period \cite{monniaux_18}.  Dormant oocytes are sheltered within somatic structures called ovarian follicles, which remain in a quiescent state until they get activated and undergo a longstanding process of growth and maturation ending by ovulation (release of a fertilizable oocyte). Growth initiation is asynchronous among follicles, so that all developmental stages can be observed in the ovaries at a given time, and follicles can remain quiescent for as long as tens of years \cite{reddy_10}.
	
	In the earliest stages of development, ovarian follicles are made up of the oocyte and a single layer of surrounding somatic cells.
	The initial cell number is on the order of ten or several of tens according to the species and is quite variable between follicles. Such a variability is inherited from the mechanism underlying the formation of primordial follicles \cite{monniaux_18b,sawyer_02}, which assemble from the fragmentation of syncytium structures (the germ cell cysts) and retrieve more or less somatic cells. 
	
	The activation of primordial (quiescent) follicles is characterized by three main processes \cite{picton_01}: (i) an irreversible transition of the somatic cell phenotype, characterized by a change in their shape, from flattened (precursor cells) to cuboidal (proliferative cells); (ii) an increase in the number of somatic cells by cell division and (iii) the awakening and associated enlargement of the oocyte. The activation phase is ended when all somatic cells have transitioned, at which time the mono-layer developmental stage is completed, and somatic cells will go on proliferating and build up several concentric layers \cite{fortune_03,clement_coupled_2013,CRY2019}.
	
	In this work, we focus on the sequence of events occurring just after the initiation of follicle growth. A key issue is to determine whether cell proliferation is concomitant or posterior to cell shape change, and to assess both the time needed for all precursor cells to complete transition and the corresponding increase in the cell number with respect to the initial cell number. 
	
    We introduce {a continuous-time Markov chain model for} cell population dynamics accounting for both cell transition and division. Within such a  formalism, linear models have been built up on the branching property, disregarding cellular interactions \cite{kimmel_theory_1963,harris_theory_1963}, while nonlinear models have accounted for interactions among different cell populations (e.g., typically, a feedback from differentiated cells onto precursor cells) either to ensure homeostasis, as in dynamical models for blood cells \cite{getto_mathematical_2015,stiehl_stem_2017,pujo_blood_2016}, or to achieve a proper developmental sequence, as in dynamical models for neural cells \cite{freret-hodara_16}.
	On our side, we are interested in assessing the duration of the  activation process, {\textit{i.e.} the extinction time of the population of precursor cells}, and in ordering the events taking place
	during activation. A natural concept in probability theory to investigate these issues is the first passage time theory \cite{darling_first_1953,van_kampen_stochastic_1992}, which aims to characterize the statistics of random events related to some particular outcomes. The analysis of first passage times are becoming more and more popular in mathematical biology \cite{chou_first_2014,castro_mathematical_2015}, to quantify random times needed to reach a given final state, such as population extinction for instance.
	
	Typically, the parameters of cell dynamics models are calibrated using time series of cell counts sorted into different cell types \cite{marr_multiscale_2012,glauche_lineage_2007}. In contrast, in the case of early folliculogenesis, precursor and proliferative cell numbers are not available directly as a function of time, but only in relation with other morphological variables such as the oocyte and follicle diameters \cite{braw-tal_studies_1997,gougeon_morphometric_1987,lundy_populations_1999,meredith_classification_2000}, so that we lack kinetic information. To overcome this difficulty, {we use } the embedded discrete-time Markov chain {to} apply classical statistical tools like the maximum likelihood \cite{wilkinson_markov_2013}, and parameter identifiability concepts \cite{raue_structural_2009}. 

	The manuscript is organized as follows. {In Section~2, we introduce a stochastic model of cell population dynamics, with two state variables and four cell events (reactions).} 
	{In section~3 we analyze both the linear and nonlinear versions of the model in the Markov chain framework}. In the linear case, we obtain analytical formulas for the mean extinction time. In the nonlinear case, we design a numerical scheme based on a rigorous Finite State Projection (see \cite{munsky_finite_2006,kuntz_deterministic_2017}) and coupling techniques to assess the mean extinction time. In both cases, we study the sensitivity of the extinction time, as well as of the proliferative cell number at extinction time, with respect to the parameter values. In section 4, using the embedded Markov chain, we calibrate the parameters of  {the model} from experimental, time-free datasets, and analyze the practical identifiability. {Using model selection criteria, we identify the key parameters that shed light on the most likely events that occur during the  activation process.} From {this} data-fitting {approach}, we manage to retrieve hidden kinetic information and provide some biological interpretations of our results. We conclude in section~5.	
	
	\section{Model design and formulation}
    Our model allows us to study the joint dynamics of the precursor cells $F$ and  proliferative cells $C$ within a single follicle, whose populations are ruled by four types of possible cell events. {In the absence of specific information, we used the simplest formulation as possible for all event rates, according to Occam’s razor principle. }
    Two cell events occur at the expense of the precursor cells, which are consumed during their transition~: (i) $\mathcal{R}_1$ is the spontaneous transition of precursor cells into proliferative cells, whose rate $\alpha_1 F$ is linearly proportional to the number of precursor cells; (ii) $\mathcal{R}_2$ is the auto-amplified transition of precursor cells into proliferative cells, which occurs at rate $\beta \frac{FC}{F + C}$. This event represents the feedback of proliferative cells onto the transition of the precursor cells.
	Two other cell events increase the proliferative cell population without affecting the precursor cell population: (i) $\mathcal{R}_3$ is an asymmetric division of precursor cells $F$ (giving rise to one precursor cell and one proliferative cell), which occurs at rate $\alpha_2 F$; (ii) $\mathcal{R}_4$ is a symmetric division of the proliferative cells $C$ (giving rise to two proliferative cells), which occurs at rate $\gamma C$. \\
	These four cell events are the building blocks of the main model $\mathcal{M}_{FC}$, which is summarized below :
	\begin{equation}\tag{$\mathcal{M}_{FC}$}\label{Model_FC}
	\begin{array}{lcl}
	& \text{ Cell events } & \text{ Rate } \\
	\mathcal{R}_1 : & (F,C)  \rightarrow (F- 1, C + 1),  & \,\,\, \alpha_1 F, \\
	\mathcal{R}_2 : & (F,C)  \rightarrow (F- 1, C + 1),  & \,\,\,  \beta \frac{FC}{F + C} , \\
	\mathcal{R}_3 : & (F,C)  \rightarrow (F, C + 1), & \,\,\, \alpha_2 F, \\
	\mathcal{R}_4 : &(F,C)  \rightarrow (F, C + 1), & \,\,\, \gamma C \, .
	\end{array}
	\end{equation}
	
	Cell events $\mathcal{R}_1$ and $\mathcal{R}_4$ constitute the fundamental ingredients involved in the activation process. We also consider two additional cell events, $\mathcal{R}_3$ and $\mathcal{R}_{{2}}$, which are not only intended to enrich the model behavior, but are also substantiated by biological observations.
	
	{Cell event $\mathcal{R}_1$ corresponds to the spontaneous transition undergone by a precursor cell, including the very first event. Firing only $\mathcal{R}_1$ events is sufficient to complete activation, yet in this case the final cell number is unchanged with respect to the initial number, which is not what is systematically observed in the experimental data {\cite{lundy_populations_1999,gougeon_morphometric_1987,lintern_79}}. On the scale of a whole follicle, the awakening of precursor cells triggers the exit from the primordial follicle pool and initiate the process of follicle growth and development \cite{zhang_somatic_2014}. Awakening is induced by activation of the protein complex mTORC1 in somatic cells (and not oocyte), by oxygen and stress-or energy-induced metabolites reaching the follicle environment in the ovarian cortex {\cite{zhang_somatic_2014}}.}
	{A natural choice for this spontaneous reaction is to consider that it occurs} {independently in each precursor cell,} {so that the transition rate of precursor cells is proportional (with coefficient $\alpha_1$) to their number $F$.}
			
	Cell event $\mathcal{R}_2$ {corresponds to auto-amplified precursor cell transitions. It has the same cell output as cell event $\mathcal{R}_1$ (loss of one precursor cell), yet it can speed up the transition rate after the first triggering event. The amplification is mediated by a positive feedback exerted by already transitioned, proliferative cells; the transition rate is $0$ when $C=0$, and increases with $C$.} {Such an auto-amplification is expected} to result from the molecular mechanisms underlying follicle activation and establishing a dialog between the oocyte and somatic cells \cite{monniaux_16}. {The activated somatic cell(s) start} stimulating the oocyte through specific signaling pathways (KIT-Ligand cytokine).
	In turn, once activated, the oocyte signals to the somatic cells through factors of the TGF$\beta$ family \cite{knight_06} (mainly GDF9 and BMP15). This molecular dialog settles a positive feedback loop, which can be represented by an auto-amplified transition rate.
	In sheep, there exist natural mutations affecting this molecular dialog (disruption of either the GDF9 or BMP15 ligand, or the receptor to BMP15). Introducing cell event $\mathcal{R}_2$ can help to investigate possible differences in the activation process in wild-type compared to mutant strains. More specifically, we have access to experimental cell numbers (courtesy of Ken McNatty) obtained either from a wild-type strain (Ile-de-France) or a mutant strain for BMP15R (Booroola), whose follicle development is known to be clearly different in the multi-layer stages \cite{lundy_populations_1999}, especially as far as cell dynamics. Whether cell dynamics is also affected during the mono-layer stage remains unclear \cite{reader_12}, which is an additional motivation for this work.
	{In the analysis performed in Section~\ref{model_analysis}, the specific formulation of the reaction rate of event $\mathcal{R}_2$ does not matter much. We just need to assume that it is linearly bounded by $F$, which is sensible with respect to cell cycle constraints. 
	To fit the model to available data in Section~\ref{sec-parameter_calibration}, we needed to specify further the shape of the nonlinearity, and chose a parameterization including as few parameters as possible. We refer to Appendix \ref{sec:appendix-R2} for a basic justification of this choice.} 
	
	{Cell event $\mathcal{R}_3$ corresponds to a self-renewal transition event that does not consume a precursor cell and may be fired as a first triggering event. It is directly inspired from asymmetric division events commonly observed in developmental cell lineages.} {The speculation that} precursor (flattened) cells {might divide while transitioning} is {compatible} with experimental studies where KI67 staining (a marker of cell cycle progression) was detected in some flattened cells \cite{dasilva_08}. Since the number of flattened cells is non increasing, one can envisage the existence of self-renewing asymmetric divisions in flattened cells, giving birth to one proliferative cell (and keeping the precursor cell number unchanged). {The rate of event $\mathcal{R}_3$ is chosen in a similar way as that of event  $\mathcal{R}_1$.}
	
	{Cell event $\mathcal{R}_4$ corresponds to the symmetric division of proliferative cells, giving birth to two identical daughter cells. It has the same cell output as cell event $\mathcal{R}_4$ (gain of one proliferative cell). All $C$ cells are supposed to progress through the cell cycle (growth fraction of one), and there is no cell loss at mitosis, so that the proliferative population growths exponentially (with rate $\gamma$).} 
	
	All the reactions rates ($\alpha_1$, $\beta$, $\alpha_2$ and $\gamma $) are non-negative. At initial time, there are only precursor cells, and the initial condition is chosen as a random positive integer variable, in consistency with the observed biological variability. 

    In the following, we will use different submodels derived from the full model $\mathcal{M}_{FC}$, by removing either one or several cell events (hence setting to zero the corresponding parameter values  $\beta $, $\alpha_2$ and/or $\gamma$). We will name these submodels by explicitly mentioning the remaining events. For instance, model ($\mathcal{R}_1,\mathcal{R}_3$) consists only of the spontaneous cell transition event and asymmetric cell division ($\beta=\gamma=0$), while model ($\mathcal{R}_1,\mathcal{R}_4$) is composed of the spontaneous cell transition event and asymmetric cell division ($\beta=\alpha_2=0$).\\
	
	Model $\mathcal{M}_{FC}$ {is} mathematically formulated {as a} Continuous time Markov chain (CTMC). 
	The stochastic description is especially appropriate when dealing with a small number of cells. {Before introducing a precise mathematical formulation, w}e can illustrate the dynamics of both the precursor and proliferative cells (Figure \ref{fig:trajecttime}). {Initially, the whole population is made of precursor cells ($C=0$), so that the first event has to be a triggering event} {generating the first proliferative cell} ($\mathcal{R}_1$ or $\mathcal{R}_3$). 
	The $C$ population grows as the $F$ population decreases until extinction (top-left panel). {The last event} {before extinction of precursor cells} {has to be a consuming event ($\mathcal{R}_1$ or $\mathcal{R}_2$).} The proportion of proliferative cells $p_C := \frac{C}{F+C}$ increases monotonously from $0$ to $1$ (bottom-left panel). In the $(C,F)$ phase plane (top-right panel), we can observe that the number of precursor cells remains constant (aligned red or black points on the horizontal line $(k, F)$, $k\in \mathbb{N}$) whenever there is a division event ($\mathcal{R}_3$ or $\mathcal{R}_4$). In contrast, whenever there is a transition event ($\mathcal{R}_1$ or $\mathcal{R}_2$), the number of precursor cells decreases by one, as illustrated by the jump from  the current line ($(k, F), k\in \mathbb{N}$) to the lower one ($(k, F-1)$, $k\in \mathbb{N}$). Hence, in this simulation, we observe a sequence of transition and division events (which appear to be here mainly spontaneous transitions $\mathcal{R}_1$ and asymmetric divisions $\mathcal{R}_3$ due to the specific parameter choice). If we are only given the sequence of events in this plane, we cannot discriminate $\mathcal{R}_1$ from $\mathcal{R}_2$, neither $\mathcal{R}_3$ from $\mathcal{R}_4$. Note that, depending on the initial condition, some parts of the phase plane cannot be reached.
	The trajectories can also be observed in the $(C,p_C)$ phase plane  (bottom-right panel). In this case, the trajectories remain on the curves parameterized by ($(k, \frac{k}{F + k})$, $k\in \mathbb{N}$) if a division event  ($\mathcal{R}_3$ or $\mathcal{R}_4$) occurs, whereas they move to the upper curves parameterized by ($(k, \frac{k}{F-1 + k})$, $k\in \mathbb{N}$) whenever a transition event ($\mathcal{R}_1$ or $\mathcal{R}_2$) occurs. 
	
	\begin{figure}[h]
		\centering
		\includegraphics[width=\linewidth]{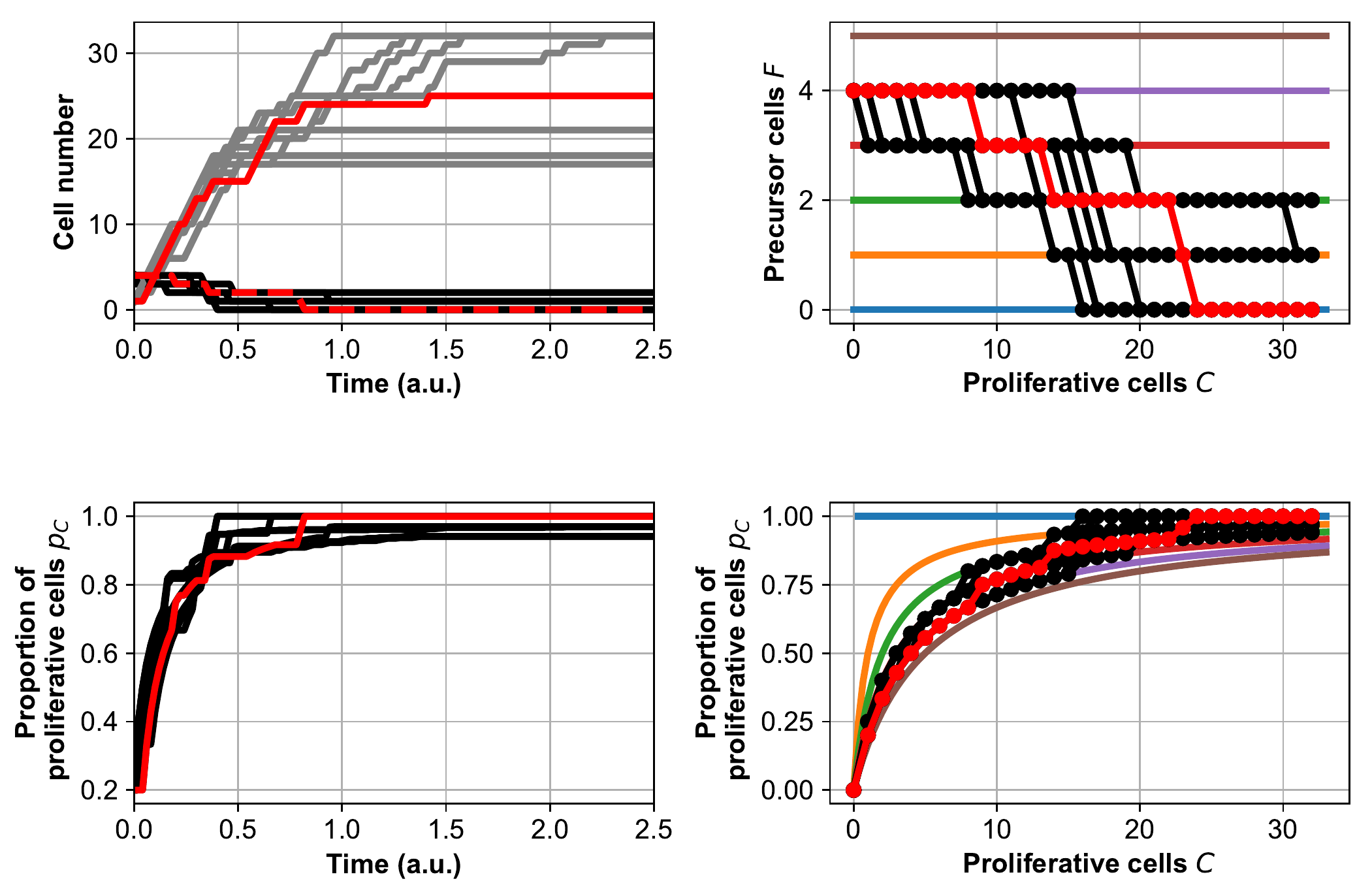}
		\caption{\textbf{Illustration of the dynamics generated by model \ref{Model_FC}.} The dynamics of the precursor and proliferative cells are computed using a Gillespie SSA algorithm \cite{gillespie_approximate_2001} with the parameter values: $\alpha_1 = 1$, $\beta = 0.01$, $\alpha_2 = 10$, $\gamma = 0.001$ and a deterministic initial condition $F(0) = 4$. In each panel, the black or gray lines represent 9 different trajectories of the process and the red line corresponds to one specific trajectory. Top-left panel: Number of precursor $F$ (black lines) and proliferative $C$ (gray lines) cells as a function of time (arbitrary units, (a.u)). Bottom-left panel: Proportion of proliferative cells $p_C$ as a function of time. Top-right panel: Number of precursor cells $F$ as a function of the number proliferative cells $C$. Bottom-right panel: Proportion of proliferative cells $p_C$ as a function of the number of proliferative cells $C$.}
		\label{fig:trajecttime}
	\end{figure}

	\paragraph{Model formulation and hypotheses}
		On a probability space $(\Omega, \mathcal{F}, \mathbb{P}) $,  let the initial number of flattened cells $F_0$ be a positive integer random variable. The population of precursor cells $F$ and proliferative cells $C$ follows the Stochastic Differential Equation (SDE) below:
	\begin{multline}\label{eq-SDE}
		F_t = F_0 - \mathcal{Y}_1 \left( \int_{0}^{t}\alpha_1 F_s ds\right) - \mathcal{Y}_2 \left( \int_{0}^{t}\beta \frac{F_s C_s}{F_s + C_s}ds\right), \\
		C_t = \mathcal{Y}_1 \left( \int_{0}^{t}\alpha_1 F_s ds\right) +  \mathcal{Y}_2 \left( \int_{0}^{t}\beta \frac{F_s C_s}{F_s + C_s}ds\right) \\
		 +  \mathcal{Y}_3 \left( \int_{0}^{t}\alpha_2 F_s ds\right)  + \mathcal{Y}_4 \left( \int_{0}^{t}\gamma C_s ds\right).
	\end{multline}
	where $\mathcal{Y}_i$, for all $i = 1, 2, 3, 4$, are mutually independent standard Poisson processes. $X = (X_t)_{t \geq 0} $, with $X_t := (F_t,C_t)$ for all $t\geq 0$, denotes the solution of \eqref{eq-SDE}. $(\mathcal{F}_t)_{t \geq 0} $ denotes the canonical filtration generated by the process X.
	
	{Classically, X can also be seen} as a continuous-time Markov chain with countable state space $\mathcal{S}:= \mathbb{N}^2 \backslash \{(0,0) \}$ {whose} infinitesimal generator $\mathcal{L}$ is given by
	\begin{multline*}
		\mathcal{L}g(f,c) = (\alpha_1 f + \beta \frac{fc}{f + c}) \left[ g(f- 1, c+ 1 ) - g(f,c)\right]  \\
		+  (\alpha_2 f + \gamma c) \left[ g(f, c + 1) - g(f, c)\right],
	\end{multline*}
	for all $g$ bounded functions and for all $(f,c) \in \mathcal{S}$.
	
	In the whole study, we will need the following hypotheses:
	\begin{hypothesis}\label{hyp-rate1}
		The spontaneous activation rate $ \alpha_1$ is positive.
	\end{hypothesis}
	\begin{hypothesis}\label{hyp-L0}
		The initial condition $F_0 $ is $L_2$-integrable.
	\end{hypothesis}

	With Hypothesis \ref{hyp-L0}, we apply Theorem 1.22 of \cite{anderson_stochastic_2015} (p.12-13) and deduce that the process $M^g_t $ defined as 
	\begin{equation}\label{MartingaleProblem_eq}
	M^g_t := g(X_t) - g(X_0) - \int_{0}^{t}\mathcal{L}g(X_s)ds
	\end{equation}
	is a $\mathcal{F}_t$-martingale, for all $t \geq 0$ and any bounded function $g$. \\
	
	Note that process $F$ is a non-negative decreasing process. To study the hitting time of the state $F = 0$, we introduce the following definition
	
	\begin{definition}\label{def-hitting-time}
		Let $\tau$ be the extinction time of the precursor cell population~$F$ 
		\begin{equation*}
			\tau := \inf \{  t \geq 0; \quad F_t = 0 | F_0 \} \, .
		\end{equation*}
		The number of proliferative cells $C$ at $t =  \tau$ is $C_{\tau} $.
	\end{definition}
	{To control the first moment of $C_{\tau} $, the number of proliferative cells at the extinction time}, we will also need an additional hypothesis: 
	\begin{hypothesis}\label{hyp-rate2}
		The {maximal} activation rate $\alpha_1{+\beta}$ is strictly greater than the proliferation rate $\gamma$: $\alpha_1 {+\beta} > \gamma $.
	\end{hypothesis}

 \section{Model analysis}\label{model_analysis}
		In this section we analyze  the mean extinction time of the precursor cell population and the number of proliferative cells at extinction. {We start in subsection \ref{part-analyticalexp} by recalling some analytical formulas for model ($\mathcal{R}_1,\mathcal{R}_3,\mathcal{R}_4$) (when $\beta=0$, linear rate functions). Then, in subsection \ref{subsec:upperbound}, we deduce a necessary and sufficient condition to ensure that the mean of $\cextin$ is finite for the complete model \eqref{Model_FC}, by finding a lower and an upper-bound thanks to a coupling argument. In subsection \ref{part-general-case}, we finally use the upper-bound in a finite-state projection algorithm to obtain an efficient way to simulate the means of $\taup$ and $\cextin$, and numerically investigate the role of the feedback rate $\beta$ on their values.}
		
		To simplify the proofs, we will consider in the following that the initial condition is a deterministic value $f_0 \in \mathbb{N}^*$. All the proofs can be generalized to the random $ F_0 $ case by conditioning by the law of $F_0$. 
		
			\subsection{Analytical expressions in the linear case, model ($\mathcal{R}_1,\mathcal{R}_3,\mathcal{R}_4$).}\label{part-analyticalexp}
			 	When $\beta=0$, process $X$ is linear, and we can compute the law of the extinction time. In the case when, in addition, $\alpha_2=0$ and/or $\gamma=0 $, the mean of $\cextin$ can also be computed. \\
			 	In this subsection we will write $X^L_t = (F^L_t,C^L_t)  $ the solution of SDE \eqref{eq-SDE} when $ \beta = 0 $ and $\taulinear $ the associated extinction time of the population $F^L_t $:
			 	\begin{equation}\label{eq-def-tau-lin}
			 			\taulinear := \inf \{  t; \quad F^L_t = 0 | f_0 \} \, .
			 	\end{equation} 
			 	
			Note that process $F^L$ is independent of process $C^L $. The jumping times $T_k $ of $F^L$, for all $k \in \llbracket 0, f_0 -1 \rrbracket $, are given by
			\begin{equation}\label{eq-jump-time-FL}
				T_{k+1} := T_k + \mathcal{E}\left(\alpha_1(f_0 - k) \right),
			\end{equation}
			with $T_0 = 0$ by convention, {and $\mathcal{E}(\lambda)$ denotes an exponential random variable of mean $1/\lambda$}. Note that $T_{f_0} = \taulinear$.
			
			\begin{proposition}[$F^L_t $ and $\taulinear $ laws]\label{prop-loi-fLT}
				Under Hypothesis \ref{hyp-rate1} and for all $t \geq 0$,
				$ F^L_t$ follows a binomial law with parameters $(f_0,e^{-\alpha_1 t})$, and the extinction time $\taulinear $, defined by formula \eqref{eq-def-tau-lin}, follows a generalized Erlang law (or hypo-exponential law) of density:
				\begin{equation*}
					f_{\taulinear}(t) = \alpha_1 f_0 e^{- \alpha_1t}(1- e^{- \alpha_1 t})^{f_0 - 1} \mathds{1}_{[0, +\infty[}(t),
				\end{equation*}
				such that $\mathbb{E}\left[ \taulinear \right] =  \displaystyle \frac{1}{\alpha_1} \sum_{k = 1} \frac{1}{k} $.
			\end{proposition}
			{The proof of Proposition \ref{prop-loi-fLT} is standard and given in Appendix \ref{ssec:proof_linear} for the {reader convenience}.}
		
		We now study the mean number of proliferative cells at the extinction time. {We first decompose process $C_t^L$ as a sum of elementary processes.}
 {We introduce the following binary branching process}
			\begin{equation}\label{eq:Yule}
				C^{0,0}_t = 1 + \mathcal{Y}\left( \gamma  \int_{0}^{t} C^{0,0}_s ds\right),
			\end{equation}
			where $\mathcal{Y}$ is a Poisson process. {This process is often referred as the Yule process \cite[Chap. 8]{Bailey}}. 
						We {then} define the stochastic processes $C^{k,j}$, for $(k,j) \in \mathbb{N} \times \mathbb{N} $, as independent and identically distributed Yule processes. {Each process $C^{k,j}$ represents the cell population that arises from the successive symmetric divisions starting from a single newly transitioned proliferative cell. We thus refer the $C^{k,j}$ processes as "cell lineages".} {Process $ C^ L_t $ is a branching process with immigration driven by cell events $\mathcal{R}_1$ and $\mathcal{R}_3$, } 
			it can {indeed} be written as the sum of the cell lineages $C^{k,j} $ (illustrated in Figure \ref{fig:branchingillustration}, {such lineage decomposition goes back at least to clonal cell population size studies like \cite{Luria_43}}): for all $t \geq 0$,
			\begin{equation}\label{eq-decomp-Yule}
				 C^L_t = \underbracket{\sum_{k = 1}^{F_0}  C^{k,0}_{t - T_k^0} \mathds{1}_{t \geq T_k^0}}_{\text{ cell lineages generated by cell event } \mathcal{R}_1 } + \underbracket{\sum_{k = 0}^{F_0 - 1} \sum_{j = 1}^{N_k(t)} C^{k,j}_{t - T_k^j} \mathds{1}_{t \geq T_k^j}}_{ \text{cell lineages generated by cell event } \mathcal{R}_3 },
			\end{equation}
			where we define, for all $k \in \llbracket 1, f_0 \rrbracket$,
			\begin{itemize}
				\item $T^0_k:= T_k $ (with $T_k$ given by equation \eqref{eq-jump-time-FL}), the $k$-th jumping time of cell event $\mathcal{R}_1$. 
				\item $N_k(t)$, the number of occurrences of cell event $ \mathcal{R}_3 $ between $T_k$ and $T_{k+1}$, for $t \geq T_k$. Note that 
				\begin{equation}\label{eq-N_k-formula}
					\displaystyle N_k(t) = \mathcal{Y}_3 \left( \alpha_2 \int_{0}^{t \wedge T_{k+1}}  F^L_s ds \right) - \mathcal{Y}_3 \left( \alpha_2 \int_{0}^{T_{k}}  F^L_s ds \right).
				\end{equation}
				\item for all $j \in \llbracket 1, N_k(t) \rrbracket $,  
				\begin{equation}\label{eq-T-jk}
					T^j_k := T_k^{j-1} + \mathcal{E}\left(\alpha_2 (f_0 -k)\right),
				\end{equation} 
				the $j$-th jumping time of the cell event $\mathcal{R}_3$ occurring between random times $T_{k}$ and $T_{k + 1} $. {Hence
				\begin{equation*}
				    N_k(t)=\sum_j 1_{t<\min(T_k^j,T_{k+1})}\,.
				\end{equation*}
				}
			\end{itemize}
			\begin{figure}[h!]
				\centering
				\includegraphics[width=0.85\linewidth]{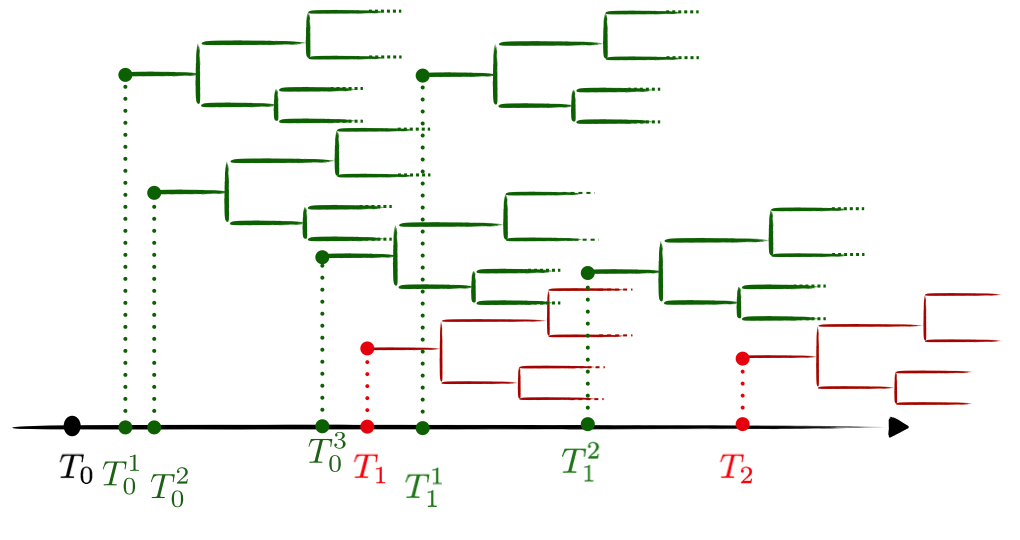}
				\caption{\textbf{Jumping times and cell lineages.} Each cell lineage represents schematically the random process $C^{k,j}$, arising either from the cell event $\mathcal{R}_1$ (green trees) or $ \mathcal{R}_2 $ (red trees), for the linear version of model $\mathcal{M}_{FC} $ (submodel $ (\mathcal{R}_1,\mathcal{R}_3,\mathcal{R}_4)$). For all $ k \in \llbracket 0, f_0 \rrbracket$, the random times $T_k$ are defined by equation \eqref{eq-jump-time-FL} and, for all $j \in \llbracket 1, N_k(t)  \rrbracket $ where $N_k(t)$ is given by equation \eqref{eq-N_k-formula}, the random times $T^j_k $ are defined by equation \eqref{eq-T-jk}. The times of the subsequent symmetric division events following times $T^j_k$ and $T_k$ are represented at arbitrary time points. }
				\label{fig:branchingillustration}
			\end{figure}
			
			{We use the sum defined in Eq. \eqref{eq-decomp-Yule} to obtain a necessary and sufficient condition to ensure that the mean of $C^L_{\taulinear}$ is finite, and to obtain analytical formulas for the submodels $(\mathcal{R}_1,\mathcal{R}_3)$ ($\beta=\gamma=0$) and $(\mathcal{R}_1,\mathcal{R}_4)$ ($\beta=\alpha_2=0$).}
			\begin{proposition}[First moment of $C^L_{\taulinear}$]\label{prop-Ce-linear}
			{Under Hypothesis \ref{hyp-rate1} we have that
			\begin{equation*}
					\mathbb{E}\left[ C^{L}_{\taulinear}  \right] <\infty
			\end{equation*}
			if, and only if, Hypothesis \ref{hyp-rate2} holds.
			Moreover,}
				\begin{enumerate}
					\item  {For the submodel $(\mathcal{R}_1,\mathcal{R}_3)$ ($\gamma=0$), we have}
					\begin{equation*}
					\mathbb{E}\left[ C^{L}_{\taulinear}  \right] = f_0 (1 + \frac{\alpha_2}{\alpha_1}).
					\end{equation*}
					\item {For the submodel $(\mathcal{R}_1,\mathcal{R}_4)$ ($\alpha_2=0$), we have 
					\begin{equation*}
					   \mathbb{E}\left[ C^{L}_{\taulinear} \right] = 1+\sum_{k=1}^{f_0-1} \frac{(f_0-k)!}{\left((f_0-k)-\frac{\gamma}{\alpha_1}\right)!}
					\end{equation*}
					 where we used notation $(m-x)!=\prod_{i=1}^{m}(i-x)$ for $m\in \mathbb{N}^*$ and $x\in[0,1)$.}
				\end{enumerate}
			\end{proposition}

{The proof of Proposition \ref{prop-Ce-linear} is classical and is given in Appendix \ref{ssec:proof_linear} for the reader convenience.}
			\begin{remarque}
				A simple analytical formula cannot be obtained for the first moment of $C^L_{\taulinear}$ for submodel ($\mathcal{R}_1,\mathcal{R}_3,\mathcal{R}_4$) since it is tricky to deal with expectation in the second term of relation \eqref{eq-decomp-Yule}.
			\end{remarque}

		\subsection{Lower and upper bounds of the nonlinear model \eqref{Model_FC}}\label{subsec:upperbound}
				In the general case, we cannot obtain analytical expressions for $\taup$, and we will rather use numerical simulations. To control the numerical error, we need tractable bounds of the stochastic model introduced in Eq. \eqref{eq-SDE}, which are obtained in this subsection. {We first note that all moments of $\taup$ are unconditionally finite, as $F_t$ decreases by one at rate at least $\alpha_1 F_t$, so that $\taup$ is stochastically dominated by $\taulinear$ given in Proposition \ref{prop-loi-fLT}. Our main result is a necessary and sufficient condition to obtain finite moments for $\cextin$.}
{\begin{theorem}\label{thm_borne_cns}
	For any $p\geq 1$, we have, under Hypothesis \ref{hyp-rate1}, that
		\begin{equation*}
		\mathbb{E}\left[\left(\cextin\right)^p \right]<\infty\,,
		\end{equation*}
	if, and only if
		\begin{equation}\label{eq:hyp_p_gamma_cns}
		p\gamma<\alpha_1+\beta\,.
		\end{equation}
		In particular, for $p=1$, $\mathbb{E}\left[\left(\cextin\right)\right]<\infty$ if and only if Hypothesis \ref{hyp-rate2} holds.
\end{theorem}}

		{The main step in the proof of Theorem \ref{thm_borne_cns} is Proposition \ref{prop-couplage}, which provides us with a concrete upper bound usable to control the numerical error of our algorithm. We first need the following definitions to set up the upper-bound process:
		\begin{definition}[truncated extinction time]\label{def-trunc-ext}
	Let $X_t := (F_t,C_t)$, for all $t\geq 0$, the solution of \eqref{eq-SDE}. For any $n\geq 1$, we define
	\[\tau_n=\min\left( \inf\{t\geq 0\,, C_t\geq n\}\,, \taup \right)\,,\]
\end{definition}
		\begin{definition}[upper-bound process]\label{def-gene-sup}
 For any $n\geq 1$, we define the following infinitesimal generators
\begin{align*}
	\mathcal{L}^{n}_F \phi(f) = \left(\alpha_1 +\beta\frac{n}{1+n}\right) f \left[ \phi(f-1) -   \phi(f) \right], \\		\mathcal{L}_C \phi(c) = \left[ (\alpha_1 + \beta + \alpha_2)f_0  + \gamma c \right] \left[ \phi(c+1) -   \phi(c) \right]\,,
\end{align*}
for any $\phi$ bounded on $\mathbb N$, and any $f,c\in \mathbb N$.
	\end{definition}
}
{
\begin{proposition}[Upper-bound]\label{prop-couplage}
Let $X_t := (F_t,C_t)$, for all $t\geq 0$, the solution of \eqref{eq-SDE}. For any $n\geq 1$, there exists a couple $(F^{n},C^{n})$ such that, for any $t<\tau_n$
\begin{equation}\label{eq:cond_init_upperbound}
F_t^{n}=f_0\,,\quad C_t^{n}=n\,
\end{equation}
and $(F_{t+\tau_n}^{n},C_{t+\tau_n}^{n})_{t\geq 0}$ is a continuous time Markov chain of generators \\
$(\mathcal{L}^{n}_F,\mathcal{L}_C)$ satisfying, 
\begin{equation}\label{eq:coupling_bound}
\cextin \leq C_{\tau_n+\overline{\tau}_n}^{n}\,.
\end{equation}
where
\begin{equation}\label{eq:tausup}
   \overline{\tau}_n:=\inf \{  t>0; \quad F_{t+\tau_n}^{n} = 0 | f_0 \}\,.
\end{equation}
\end{proposition}
}

{The proof of Proposition \ref{prop-couplage} proceeds by a coupling argument between $(F,C)$ and $(F^{n},C^{n})$. The random variable $C_{\tau_n+\overline{\tau}_n}^{n}$ has a finite $p$-moment under \eqref{eq:hyp_p_gamma_cns}, and this moment is analytically tractable, thanks to Proposition \ref{prop-esp-tilde} in Appendix \ref{ssec:proof_linear}.
\begin{proof}[of Proposition \ref{prop-couplage}]
Let $F_t,C_t$ given by equation \eqref{eq-SDE}. 
Let $n\geq 1$ and $(F^{n},C^{n})$ defined on $(0,\taup_n)$ by \eqref{eq:cond_init_upperbound}. 
Clearly, for $t\in (0,\taup_n)$, 
\begin{equation*}
   F_{t}\leq F_{t}^{n}\,,\quad  C_{t} \leq C_{t}^{n}\,.
\end{equation*}
Then, we may choose $(F_{t+\taup_n}^{n}, C_{t+\taup_n}^{n})_{t\in \mathbb R^+}$ such that its infinitesimal generator is given by $(\mathcal{L}^{n}_F,\mathcal{L}_C)$ in Definition \ref{def-gene-sup}, and such that its trajectories satisfy, for any $t\geq 0$,
\[F_{t+\taup_n}\leq F_{t+\taup_n}^{n}\,,\quad C_{t+\taup_n}\leq C_{t+\taup_n}^{n}\,, \quad \textit{a.s.}\]
This is clearly possible as
\begin{item}
\item[$\bullet$] For $t=0$, the initial condition satisfies
\begin{equation*}
   F_{\taup_n}\leq F_{\taup_n}^{n}\,,\quad  C_{\taup_n} \leq C_{\taup_n}^{n}\,.
\end{equation*}
\item[$\bullet$] For any $t> 0$, we can that ensure $F_{t+\taup_n}$ stays below $ F_{t+\taup_n}^{n}$ because
\begin{equation*}
   \left(\alpha_1+\beta\frac{C_{t+\taup_n}}{F_{t+\taup_n} + C_{t+\taup_n}}\right)F_{t+\taup_n}\geq \left(\alpha_1+\beta\frac{n}{1 + n}\right)F_{t+\taup_n}\,,
\end{equation*}
as $C_t$ is non-decreasing.
\item[$\bullet$] For any $t> 0$, we can ensure that $C_{t+\taup_n}$ stays below $ C_{t+\taup_n}^{n}$ because
\begin{equation*}
   (\alpha_1+\alpha_2)F_{t+\taup_n} + \beta\frac{C_{t+\taup_n}}{F_{t+\taup_n} + C_{t+\taup_n}}F_{t+\taup_n}\leq (\alpha_1+\alpha_2+\beta)f_0\,,
\end{equation*}
as $F_t$ is non-increasing.
\end{item}\\
On the event $(\taup_n<\taup)$, we have $C_{t+\taup_n}\leq C_{t+\taup_n}^{n}$ for all times $t\geq 0$, and $0<\taup-\taup_n \leq \overline{\tau}_n$. 
We note that $\overline{\tau}_n$ follows a generalized Erlang law of parameter $\alpha_1+\beta\frac{n}{1+n}$, by straightforward adaptation of Proposition  \ref{prop-loi-fLT}. Hence 
\begin{equation*}
C_{\taup}=C_{\taup-\taup_n+\taup_n}\leq C_{\taup-\taup_n+\taup_n}^{n} \leq C_{\overline{\tau}_n+\taup_n}^{n}\,,
\end{equation*}
as $C^{n}(t)$ is non-decreasing. 
\\
On the event $(\taup_n=\taup)$, we clearly have $C_\taup \leq n \leq C_{\overline{\tau}_n+\taup_n}^{n}$.\\
Combining both cases proves equation \eqref{eq:coupling_bound}.\\
\myspecialendproof
\end{proof}
}
{We now proceed to the proof of Theorem \ref{thm_borne_cns}.
	\begin{proof}[of Theorem \ref{thm_borne_cns}]
Let $p\geq 1$. Condition~\eqref{eq:hyp_p_gamma_cns} implies that there exists (a sufficiently large) $n\geq 1$ such that 
	\begin{equation*}
p\gamma<\alpha_1+\beta\frac{n}{1+n}\,
\end{equation*}
holds true. In particular, for such a $n$, we have
		\[\mathbb{E}\left[\left(C_{\overline{\tau}_n+\tau_n}^{n}\right)^p \right]<\infty\,\]
		by consequence of Proposition \ref{prop-esp-tilde} in Appendix \ref{ssec:proof_linear} (which also provides explicit bounds for $p=1,2$).
We now prove that condition \eqref{eq:hyp_p_gamma_cns} is necessary, excluding the trivial situation where $f_0=1$ and $\alpha_2=0$ (in which case $C_{\taup}=1$). First, note that $\taup$ is stochastically lower-bounded by an exponential random variable of rate $(\alpha_1+\beta)$, because the maximal activation rate of $F$ is $\alpha_1+\beta$. The Yule process in Eq.~\eqref{eq:Yule} provides a lower-bound for $C_{t}$ for times $t$ greater than the first event (given by an exponential random variable of rate $\alpha_1+\alpha_2$). Thus, the Yule process stopped at an exponential time of parameter $\alpha_1+\beta$ provides a lower bound for $C_{\taup}$. We conclude again by Proposition \ref{prop-esp-tilde} in Appendix \ref{ssec:proof_linear} (with $n=1$). 
\myspecialendproof
\end{proof}
}
			
		\subsection{Numerical scheme for the mean extinction time and mean number of proliferative cells at the extinction time}\label{part-general-case}
		
		    {We now have all the ingredients to study} {numerically} {the impact of the model parameters on the mean activation duration of an ovarian follicle (mean extinction time of precursor cells) and the mean number of proliferative cells produced during this phase.} \\
		    {From the martingale problem \eqref{MartingaleProblem_eq}, it is a standard result to} compute the moment of $ \taup$ and $C_\taup$. 	Let the domain $\mathcal{D}$ be defined as 
			\begin{equation*}
				\mathcal{D} := \llbracket 1, f_0 \rrbracket \times \mathbb{N}.
			\end{equation*}
{We look for} the value $g(f_0,0) $ where $g$ is solution of
			\begin{equation}\label{g_function}
				\forall (f,c) \in \mathcal{D}, \, \mathcal{L}g(f,c) = \alpha \text{ and }  g(0,c) = g_0(c), \, \forall c \in \mathbb{N} 
			\end{equation}
			where function  $g_0 $ and scalar $\alpha$ are to be chosen according to whether we want to obtain  $\mathbb{E}\left[ \taup \right] $ or $\mathbb{E}\left[ \cextin \right] $. 
				\begin{enumerate}
					\item For $\mathbb{E}\left[ \taup \right] $, we take, for all $c \in \mathbb{N} $, $g_0(c) = 0$ and $\alpha = -1 $.
					\item For $\mathbb{E}\left[ \cextin \right] $, we take, for all $c \in \mathbb{N} $, $g_0(c) = c$ and $\alpha = 0 $. 
				\end{enumerate}
		
			We can notice that system \eqref{g_function}, which is similar to the Kolmogorov backward equation, is unclosed, and there exists no analytical solution. We now obtain a numerical estimate for the scalar $g(f_0,0)$ using a domain truncation method, as proposed in \cite{munsky_finite_2006,kuntz_deterministic_2017}. 
			
			\paragraph{Domain truncation method}
			
			{For $r \in \mathbb{N}^*$, let $\mathcal{D}^r$ be the following truncated domain
			\begin{equation*}
			    \mathcal{D}^r=\llbracket 1, f_0 \rrbracket \times \llbracket 0, r-1 \rrbracket\,.
			\end{equation*}
			Note that the truncated extinction time\footnote{{Although the cut-off $r$ plays a similar role as the index $n$ from section \ref{subsec:upperbound}, we will need two distinct values for the numerical scheme, so that we stick with two different notations, to avoid possible confusion.}} $\tau_r$ defined in Definition \ref{def-trunc-ext} is the first exit time from $\mathcal{D}^r$,
			\begin{equation*}
					\tau_r = \inf \left( t \text{ such that } X_t \notin \mathcal{D}^r \right) .
				\end{equation*}
				As $\mathcal{D}^r \subset \mathcal{D}$, we clearly have $\tau_r \leq \tau$ and consequently  $C_{\tau_r}\leq C_{\taup}$. Also, as $\mathcal{D}^r$ is a strictly increasing sequence of sets such that $\cup_r \mathcal{D}^r = \mathcal{D}$, the upper-bound obtained in Proposition \ref{prop-couplage} will allow us to prove that 
				\[
				\lim_{r\to\infty}\tau_r= \tau\,,
				\]
				and 
				\[
            \lim_{r\to\infty}C_{\tau_r}= C_{\taup}\,,
				\]
				and to control the speed of convergence.
				\begin{proposition}[Domain truncation relative error]\label{prop-erreur}
				Let $p \in \mathbb{N}^*$, such that $\mathbb{E}[(C_{\taup })^p] < \infty  $. Then, we have
				\begin{equation*}
				0 \leq \mathbb{E}\left[ \taup \right] -	\mathbb{E}\left[\tau_r \right] \leq \mathbb{E}\left[\overline{\tau}_r\right]\frac{\mathbb{E}[(C_{\taup })^p]}{r^p} 
				\end{equation*}
				 and 
				\begin{equation*}
				0 \leq  \mathbb{E}\left[C_{\tau}\right] -	\mathbb{E}\left[C_{\tau_{r}} \right] \leq \mathbb{E}\left[C_{\tau_r+\overline{\tau}_r}^{r}-r\right]\frac{\mathbb{E}[(C_{\taup })^p]}{r^p} \,,
				\end{equation*}
				 where $\overline{\tau}_r$ and $C^{r}$ are defined in Proposition \ref{prop-couplage} (see Eqs.~\eqref{eq:cond_init_upperbound}-\eqref{eq:tausup}).
			\end{proposition}
			\begin{proof}
			At $\tau_r$, we either have $F_{\tau_r}=0$, in which case $\tau_r=\tau$, or $F_{\tau_r}\in [1,f_0]$, in which case $\tau_r<\tau$ and $C_{\tau}\geq r$. We use this dichotomy to compute the difference $\mathbb{E}\left[ \taup-\tau_r \right]$:
			\begin{multline*}
			    \mathbb{E}\left[ \taup-\tau_r \right]=\mathbb{E}\left[ \taup-\tau_r | F_{\tau_r}=0\right]\mathbb{P}\left[F_{\tau_r}= 0 \right]\\
			    +\mathbb{E}\left[ \taup-\tau_r | F_{\tau_r}\geq 1\right]\mathbb{P}\left[F_{\tau_r}\geq 1 \right]\\
			    =\mathbb{E}\left[ \taup-\tau_r | F_{\tau_r}\geq 1\right]\mathbb{P}\left[F_{\tau_r}\geq 1 \right]\,.
			\end{multline*}
            Given that $F_{\tau_r}\geq 1$, we have $C_t\geq r$ for all $t\in [\tau_r,\tau]$. Hence, by the same coupling procedure as in the proof of Proposition \ref{prop-couplage}, $\taup-\tau_r \leq \overline{\tau}_r$ where $\overline{\tau}_r$ is independent of $F_{\tau_r}$, and follows a generalized Erlang law of parameter $\alpha_1+\beta\frac{r}{1+r}$. Moreover, we clearly have
            \[\mathbb{P}\left[F_{\tau_r}\geq 1 \right] \leq \mathbb{P}\left[C_{\tau}\geq r \right]\,. \]
            We then conclude by Chebychev inequality that
            \begin{equation*}
			    \mathbb{E}\left[ \taup-\tau_r \right]\leq \mathbb{E}\left[\overline{\tau}_r\right]\frac{\mathbb{E}[(C_{\taup })^p]}{r^p}\, .
			\end{equation*}
			Using the same reasoning we obtain
			\begin{equation*}
			    \mathbb{E}\left[ C_{\tau}-C_{\tau_{r}} \right]=\mathbb{E}\left[ C_{\tau}-C_{\tau_{r}} | F_{\tau_r}\geq 1\right]\mathbb{P}\left[F_{\tau_r}\geq 1 \right]\,,
			\end{equation*}
			and, given that $F_{\tau_r}\geq 1$, 
			\begin{equation*}
		 C_{\tau}-C_{\tau_{r}}=C_{\tau-\tau_{r}+\tau_{r}}-C_{\tau_{r}}
			  \leq C_{\overline{\tau}_r+\tau_{r}} - r \leq C_{\overline{\tau}_r+\tau_{r}}^{r}- r
			\end{equation*} where again $C^{r}$ is defined in Proposition \ref{prop-couplage} and is independent of $F_{\tau_r}$.
				\myspecialendproof 
			\end{proof}
			The sequence of random variables $\overline{\tau}_r$ is uniformly bounded in mean,
			\[\sup_r \mathbb{E}\left[\overline{\tau}_r\right] \leq  \mathbb{E}\left[\tau_L\right]=\frac{1}{\alpha_1} \sum_{k = 1} \frac{1}{k}<\infty\,.\]
			Proposition \ref{prop-erreur} can then be used effectively with any $p\geq 1$ to approximate $\mathbb{E}\left[ \taup\right]$ with $\mathbb{E}\left[\tau_r \right]$. Using Proposition	\ref{prop-couplage}, we deduce that for any $\gamma <\alpha_1+\beta$ (e.g. under Hypothesis \ref{hyp-rate2}), $\mathbb{E}\left[ \taup-\tau_r \right]$ is decreasing as $\frac{1}{r}$, with a computable pre-factor given in Proposition \ref{prop-esp-tilde} in Appendix \ref{ssec:proof_linear}.\\
			Similarly,  under Hypothesis \ref{hyp-rate2} and from Proposition \ref{prop-esp-tilde}, we obtain the explicit bound for $r$ sufficiently large (such that $\gamma<\alpha_1+\beta\frac{r}{1+r}$),
		\[\mathbb{E}\left[C_{\overline{\tau}_r+\tau_{r}}^{r}- r\right]\leq \left(r+ \frac{(\alpha_1 + \beta + \alpha_2)f_0 }{\gamma}\right) \left(\frac{f_0!}{\left(f_0-\frac{\gamma}{\alpha_1+\beta\frac{r}{1+r}}\right)!}-1\right)<\infty\,.\]
		The latter expression increasing at most linearly in $r$, Proposition \ref{prop-erreur} can be used with any $p\geq2$. Under the assumption that $2\gamma < \alpha_1+\beta$, we deduce from Proposition	\ref{prop-couplage} that $\mathbb{E}\left[C_{\tau}\right] -	\mathbb{E}\left[C_{\tau^{r}} \right]$  is also decreasing as $\frac{1}{r}$, with a computable pre-factor given in Proposition \ref{prop-esp-tilde} in Appendix \ref{ssec:proof_linear}.
			}
			
			{To be complete, we detail in Appendix \ref{appendix:numericalscheme} the pseudo-code (Algorithm \ref{algo-gr}) to compute $\mathbb{E}\left[\tau_r \right]$ and $\mathbb{E}\left[C_{\tau_{r}} \right]$.}
	
			We simulate Algorithm \ref{algo-gr} to explore the influence of {the amplification rate $\beta$} on both the mean of $\taup$, $\mathbb{E}\left[\taup\right]$, and the mean of $\cextin$, $\mathbb{E}\left[\cextin\right]$ for the nonlinear model. \\
				{First, we can prove that both $\mathbb{E}\left[\taup\right]$ and $\mathbb{E}\left[\cextin\right]$ are monotonously decreasing with increasing rate $\beta$, and that the following limits hold, with fixed $(f_0,\alpha_1,\alpha_2,\gamma)$ and under Hypothesis \ref{hyp-rate2}:
				\begin{equation}\label{eq:limit_beta1_tau_ctau}
				\begin{array}{ll}
				    \lim_{\beta\to 0} \mathbb{E}\left[\taup\right]=\mathbb{E}\left[\taup_L\right]\,, & \lim_{\beta\to \infty} \mathbb{E}\left[\taup\right]=\frac{1}{(\alpha_1+\alpha_2)f_0}\\
				    \lim_{\beta\to 0} \mathbb{E}\left[\cextin\right]=\mathbb{E}\left[C^L_{\taulinear}\right]\,, & \lim_{\beta\to \infty} \mathbb{E}\left[\cextin\right]=f_0+\frac{\alpha_2}{\alpha_1+\alpha_2}\,.
				\end{array}
				\end{equation}
				The limit $\beta\to 0$ is the consequence of the continuity of the Markov chain with respect to its parameters. For the limit $\beta\to\infty$, note that the very first event to occur is either $\mathcal{R}_1$ or $\mathcal{R}_3$, and the remaining ones are $\mathcal{R}_2$ almost surely in the limit $\beta\to\infty$. These limits are illustrated on Figure \ref{fig:odetimef}. Moreover, the different parameter configuration used in Figure \ref{fig:odetimef} leads to the guess\footnote{we are not able to prove it, as no analytical formula is available for the full model} that both $\mathbb{E}\left[\taup\right]$ (left panel) and $\mathbb{E}\left[\cextin\right]$ (right panel) have a high sensitivity to $\beta$ in the range  $\beta\approx \alpha_1$.
				}
			 The numerical simulations {indicate furthermore} that{, in the presence of the auto-amplified event $\mathcal{R}_2$}, the division rates $\alpha_2$ and $\gamma$ have {very little influence on $\mathbb{E}\left[\taup\right]$ while they affect dramatically $\mathbb{E}\left[\cextin\right]$}.  It is {also} clear from the analytical solutions of the linear model, that the initial number of precursor cells $f_0$ and the spontaneous transition rate $\alpha_1$ have a major impact on {both} $\mathbb{E}\left[\taup\right]$ and {$\mathbb{E}\left[\cextin\right]$}.
			 {A high sensitivity of model outputs to parameters is interesting to suggest possible key biological measurements (if feasible) in order to improve parameter identifiability (see paragraph \ref{ssec:timereconstruc}-\ref{ssec:moment_data}).}
			 
			\begin{figure}
					\centering
					\includegraphics[width=\linewidth]{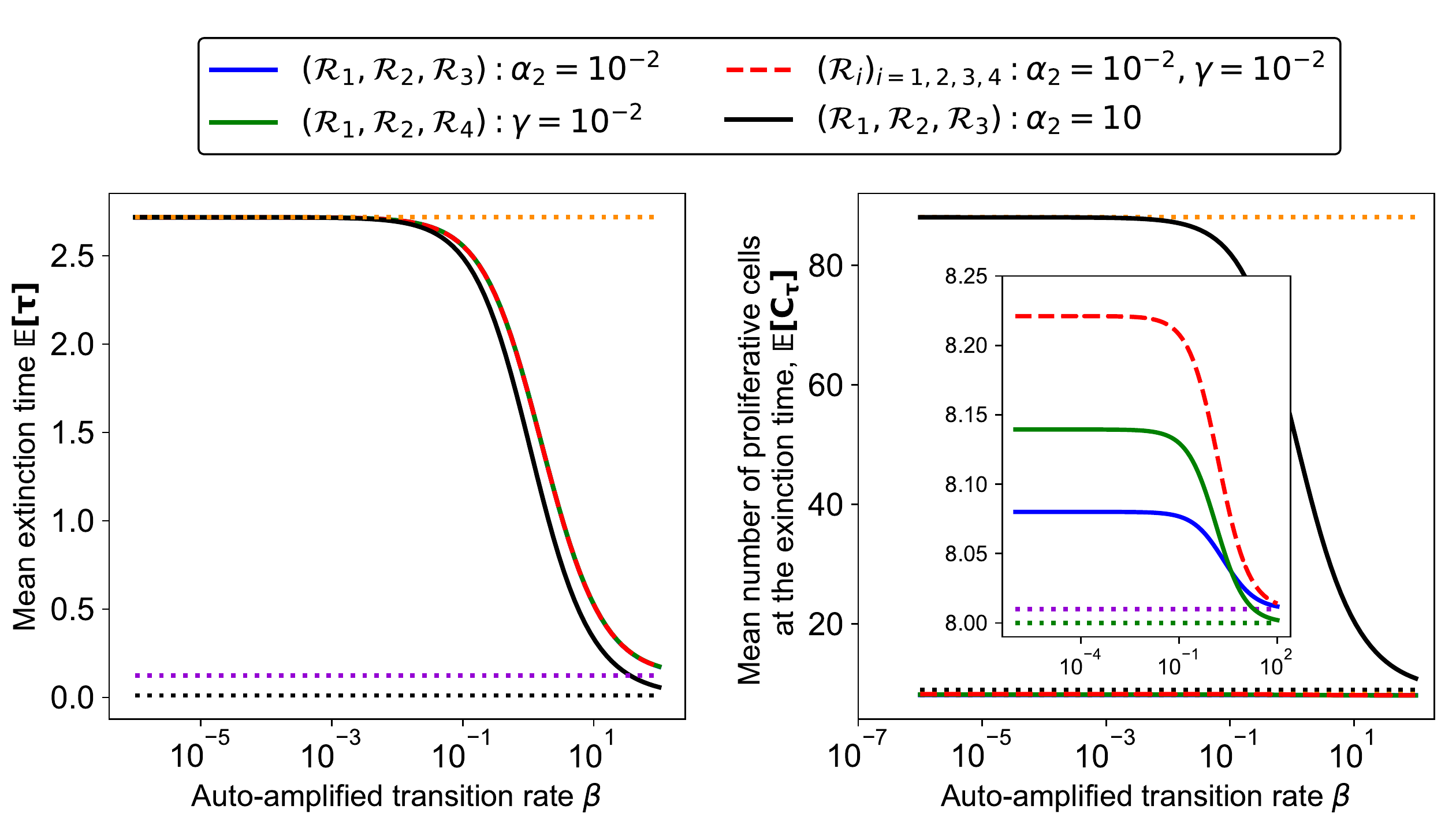}
					\caption{Mean extinction time $\mathbb{E}\left[\taup\right]$ and mean number of proliferative cells at the extinction time $\mathbb{E}\left[\cextin\right]$, as a function of the auto-amplified transition rate $\beta$. Using Algorithm \ref{algo-gr} with $\epsilon = 10^{-2}$, we compute $\mathbb{E}\left[\taup\right]$ and $\mathbb{E}\left[\cextin\right]$. {The solid and dashed lines correspond to the change in $\mathbb{E}\left[\taup\right]$ (left panel) and $\mathbb{E}\left[\cextin\right]$ (right panel) a function of $\beta$.}  In each panel, we use four different parameter configurations as follows. In all cases, $f_0 = 8$ and $\alpha_1=1$. Black solid line: submodel $(\mathcal{R}_1,\mathcal{R}_2,\mathcal{R}_3)$ with $\alpha_2 = 10$. Blue solid line: submodel $(\mathcal{R}_1,\mathcal{R}_2,\mathcal{R}_3)$ with $\alpha_2 = 0.01$. Green solid line: submodel $(\mathcal{R}_1,\mathcal{R}_2,\mathcal{R}_4)$ with $\gamma = 0.01$.  Red dashed line: complete model \ref{Model_FC} with $\alpha_2=\gamma = 0.01$.  The orange dotted horizontal lines represent $\mathbb{E}\left[\taup\right]$ and $\mathbb{E}\left[\cextin\right]$ when $\beta = 0$ (applying {Eq.~\eqref{eq:limit_beta1_tau_ctau} and} formulas in Proposition \ref{prop-Ce-linear} or, for submodel $(\mathcal{R}_1, \mathcal{R}_3,\mathcal{R}_4)$, simulating the stochastic process). The dotted  horizontal lines correspond to $\mathbb{E}\left[\taup\right]$ and $\mathbb{E}\left[\cextin\right]$ when $\beta\to\infty$ {(applying Eq.~\eqref{eq:limit_beta1_tau_ctau})}. {For $\mathbb{E}\left[\taup\right]$:  black dotted line: model $(\mathcal{R}_1,\mathcal{R}_2,\mathcal{R}_3)$ with $\alpha_2 = 10$;  purple dotted line: the three remaining models (superimposed). For $\mathbb{E}\left[\cextin\right]$: black dotted line: model $(\mathcal{R}_1,\mathcal{R}_2,\mathcal{R}_3)$ with $\alpha_2 = 10$; green dotted line: model $(\mathcal{R}_1,\mathcal{R}_2,\mathcal{R}_4)$; purple dotted line: the two remaining models (superimposed).}}
					\label{fig:odetimef}
				\end{figure}
			   
\section{Parameter calibration}\label{sec-parameter_calibration}
	In this section, we calibrate the model parameters using a likelihood approach. We first describe {in subsection \ref{subsect-dataset_description}} the available experimental dataset, as well as \textit{in silico} datasets that we use as a benchmark for our methodology. {In subsection \ref{subsec-likelihood-method}} we derive a likelihood function based on the embedded Markov chain from the underlying continuous-time Markov process. We explain how this likelihood is specifically adapted to the data, which are time-free measurements of cell numbers.
	{We present the estimation results in subsection \ref{ssec:fittingresuts} for model \eqref{Model_FC} and  each of the five different submodels: $(\mathcal{R}_1,\mathcal{R}_3)$, $(\mathcal{R}_1,\mathcal{R}_4)$, $(\mathcal{R}_1,\mathcal{R}_2,\mathcal{R}_3)$, $(\mathcal{R}_1,\mathcal{R}_2,\mathcal{R}_4)$ and $(\mathcal{R}_1,\mathcal{R}_3,\mathcal{R}_4)$. We recall that the different submodels are named by the reactions which have corresponding positive reaction rates. All the submodels considered are thus nested models, or reduced model compared to the full model $\eqref{Model_FC}$.  We carry out a comprehensive comparison between the different models using model selection criteria.}  
	{Thanks to a practical parameter identifiability analysis, we obtain model predictions in subsection \ref{ssec:prediction}, where we manage to retrieve hidden kinetic information and assess transit times and number of division events during the activation phase, with given confidence intervals. Finally, we discuss the biological interpretation of the calibration results in subsection \ref{sec:interpretation}.}

\subsection{Dataset description}\label{subsect-dataset_description}
	
	\paragraph{\textbf{Experimental dataset}}
	
	Follicles undergoing the activation process have been classified according to three types \cite{braw-tal_studies_1997,gougeon_morphometric_1987,lundy_populations_1999,meredith_classification_2000}. Primordial follicles (Type I or B) have either not yet or just initiated activation; they are composed of a single layer of flattened cells surrounding the oocyte. Primary follicles (Type II or C) have completed initiation; they only contain cuboidal (transitioned) somatic cells organized in less than two layers (this means that some follicles are strictly mono-layered, while in others an extra partially fulled layer is being built-up).
	In between Types I and II lies a class of transitory follicles (Type IA or B/C), with a mixture of flattened and cuboidal cells coexisting within a single layer. The progression from Type I to Type II is accompanied with a more or less pronounced increase in the total cell number (flattened plus cuboidal cells) and enlargement in the oocyte (and follicle) diameter (see bottom-right panel of Figure \ref{fig:plancf}).
	
	\begin{figure}[h!]
		\centering
		\includegraphics[width=5.8cm]{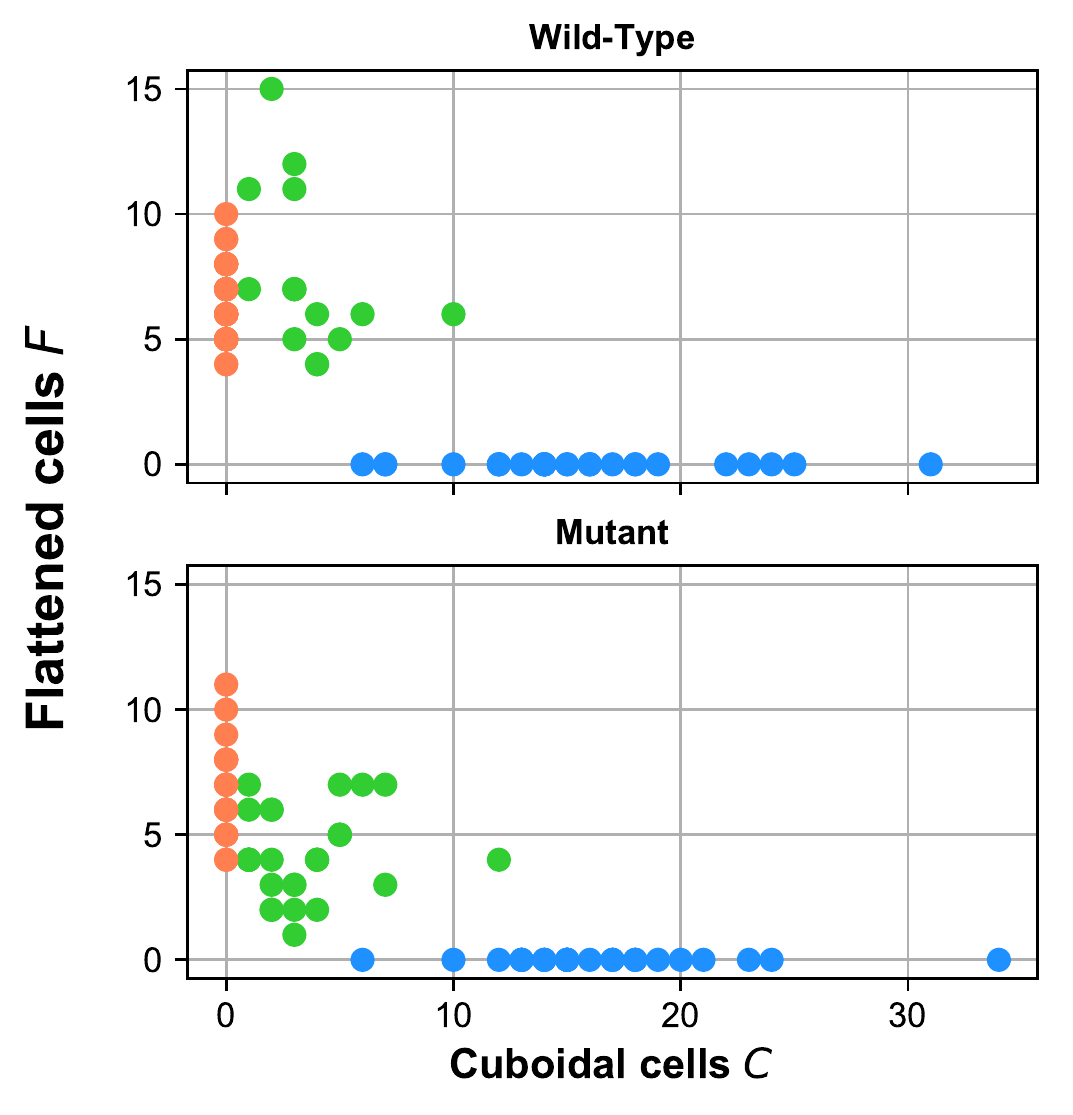}
		\includegraphics[width=5.8cm]{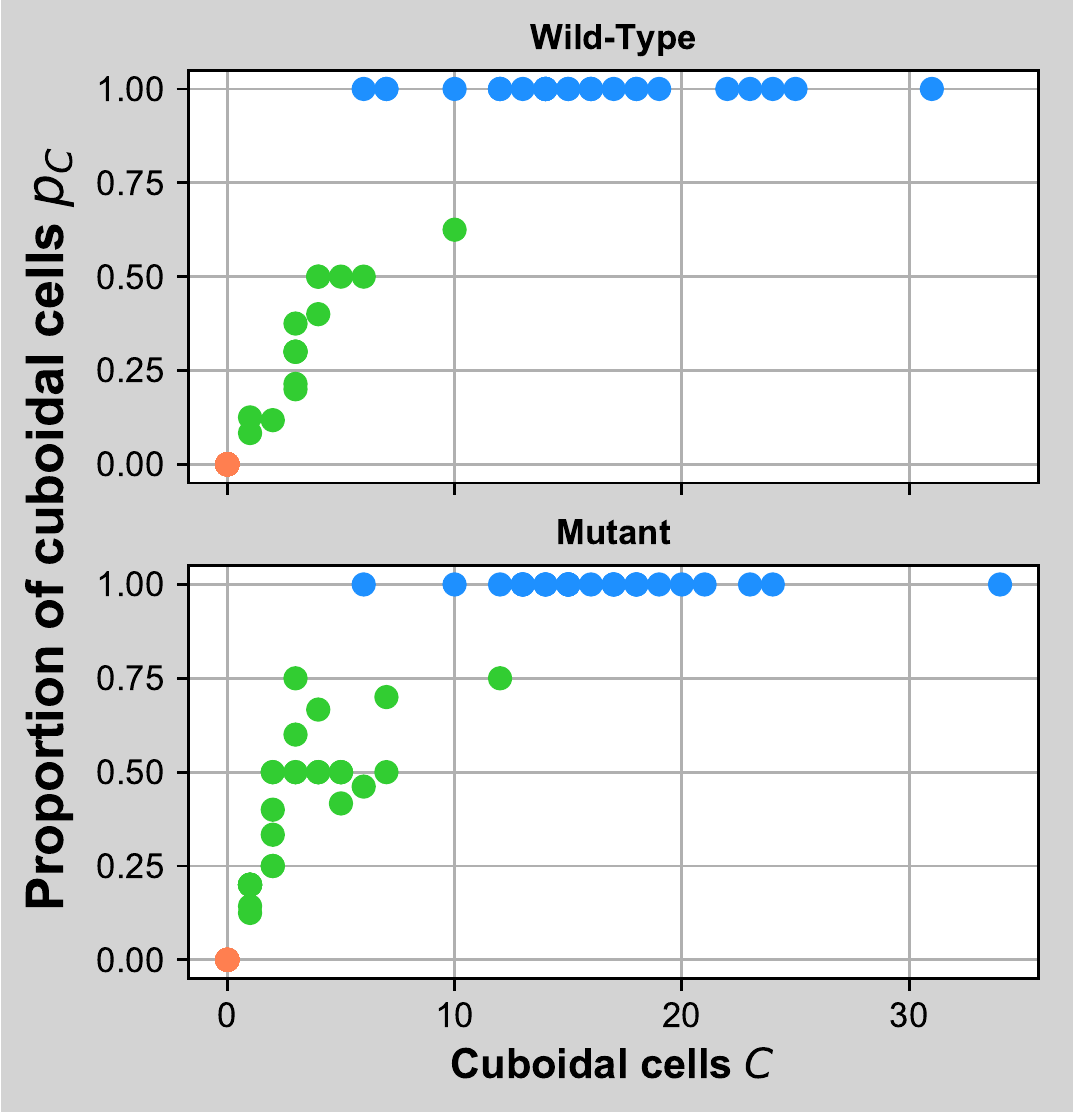}
		\includegraphics[width=5.8cm]{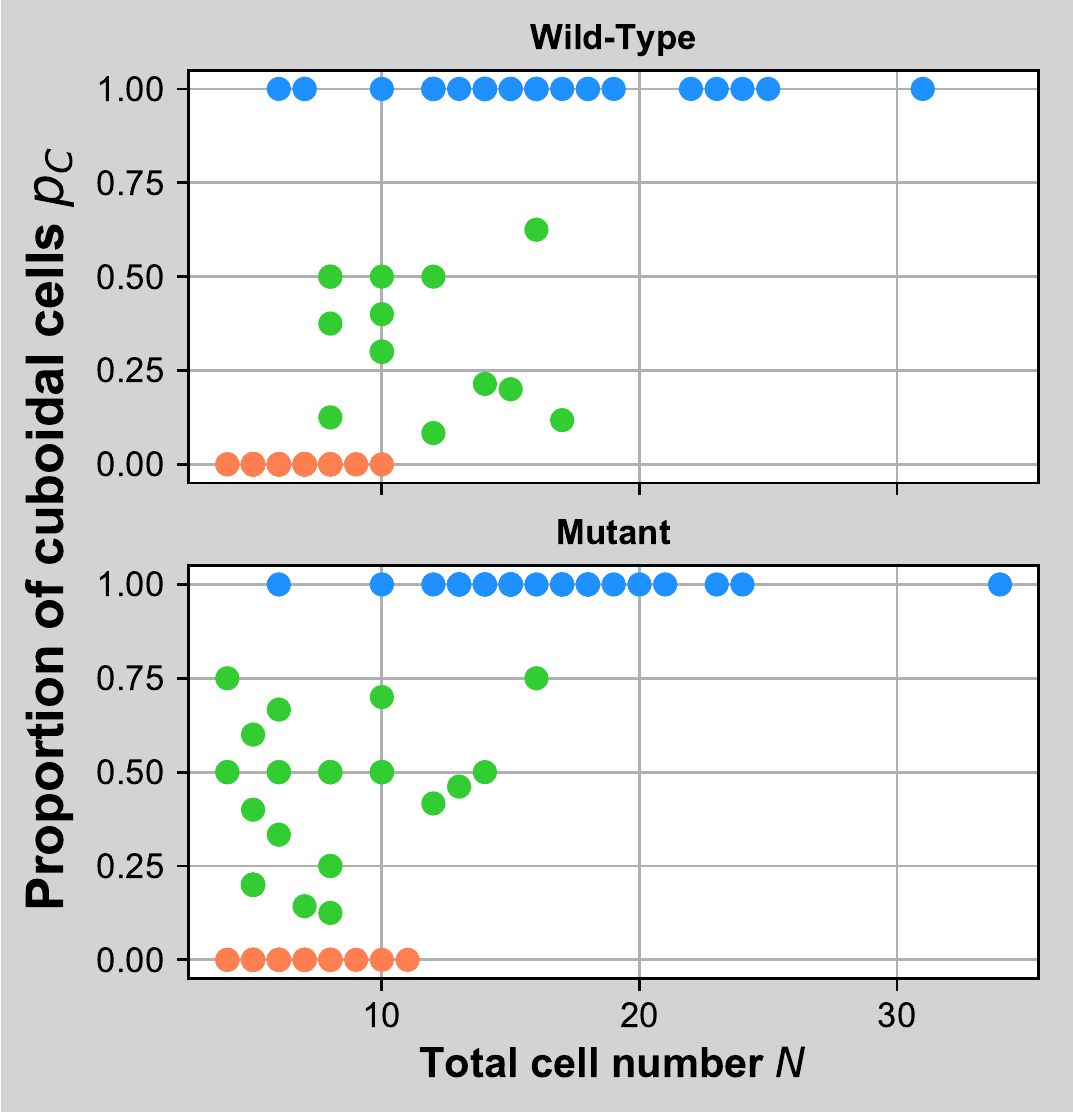}
		\includegraphics[width=5.8cm]{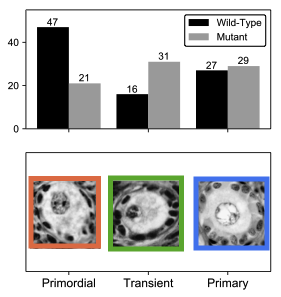}
		\caption{\textbf{Description of the experimental dataset}. {We recall that $F$ stands for the flattened cells, $C$ the cuboidal cells, $p_C$ the proportion of cuboidal cells, $p_C=C/(F+C)$, and $N=F+C$ the total number of cells.} Top-left, top-right and bottom-left panels: experimental data points projected onto three different phase planes, respectively: $(F,C)$, $(C,p_c)$ and $(N,p_C) $, for both the Wild-Type and Mutant subsets.  Red points: primary follicles, green points: transitory follicles, blue points: primary follicles. Bottom-right panel:  histological slices illustrating the different steps of activation (from left to right: primordial, transitory and primary follicles). Experimental dataset: courtesy of Ken McNatty; histological images: courtesy of Danielle Monniaux.}
		\label{fig:plancf}
	\end{figure}

    We have made use of a dataset acquired in sheep fetuses \cite{lundy_populations_1999,wilson_01} (courtesy of Ken McNatty), which provides us with flattened (precursor) and cuboidal (proliferative) cell numbers in a sample of follicles distributed into the three activation steps. The dataset is subdivided into two subsets corresponding to two different sheep strains : the ``wild-type'' Romney strain and the ``mutant'' Booroola strain. The latter is characterized by a natural mutation affecting the receptor to growth factor BMP15 and resulting in the alteration of follicle development (see the Introduction section).
	
		We have $90$ data points for the Wild-Type dataset, and $81$ for the Mutant dataset.
		More specifically, the measures consist of the cell numbers counted on the largest 2D cross-section of histologically fixed follicles of type I, IB or II. This 2D number can be correlated  with the total 3D cell number from standard stereological considerations \cite{lundy_populations_1999}. 
		
		{On a time horizon of several weeks, as it is the case in the experimental study we based on \cite{lundy_populations_1999}, a primordial follicle can undergo three different fates: it can either get activated, die, or remain quiescent. When considering the long term evolution (during the whole reproductive lifespan), all healthy primordial follicles will eventually get activated , the vast majority of non healthy primordial follicles will die before leaving the pool, and the remaining ones will never get activated {\cite{reddy_10}} (in women there are approximately one thousand follicles left in the ovaries after menopause {-- on the order of $0.2\%$ of the initial pool} {\cite{broekmans_09}}).} 

{Except in the case of morphological abnormalities, it is not obvious to classify primordial follicles as healthy, in the sense that they will get activated at one time or another. To use appropriate data for the model fitting step, which should concern only, or at least mostly, ``activable'' follicles, we thus needed to gather complementary {\it a priori} biological information.}

{The viability of the oocyte enclosed in a primordial follicle is the main determinant of the follicle health. Whether an oocyte is viable or not at this developmental stage results from the process of follicle formation, that we describe here very briefly (we refer the interested reader to \cite{juengel_02,monniaux_18b,tingen_09,wang_17} for a complete overview). During embryonic development, primordial germ cells colonize the territory corresponding to the future gonads. In females, these cells undergo several rounds of mitotic divisions while they interact locally with somatic cells. They gather in syncytium structures, the germ cell cysts, that ultimately fragment into primordial follicles. Only a small proportion of the germ cells survive to this step (25\% in sheep \cite{smith_93}), which also coincides with the entry into meiosis. Most oocytes will die and transfer a part of their cytoplasmic material to surviving oocytes. Each surviving oocyte recruits a variable number of somatic cells to build a primordial follicle \cite{sawyer_02}}.
{The recruitment of a sufficient number of somatic cells is crucial to ensure future oocyte survival. Somatic cells secrete trophic factors, such as neurotrophin, whose intrafollicular levels should be high enough to guarantee the oocyte survival \cite{spears_03}. A further survival requirement is the existence of tight intercellular communications between the oocyte and its surrounding somatic cells, and between the somatic cells.} 

{From these information, the following ``viability criteria'' can be derived to select proper type~I follicles.  First, there is a threshold oocyte diameter compatible with oocyte viability, and, second, for a given oocyte diameter, there is a threshold number of somatic cells surrounding the oocyte, that should be organized as a connected cell network paving the oocyte surface. The critical oocyte diameter can be determined in a rather straightforward way by comparison with the minimal oocyte diameter observed in transitory follicles. In the sheep species, following the study of  \cite{cahill_81} (see also \cite{picton_01} for other species), we set the oocyte diameter threshold to 24 $\mu$m (corresponding to a follicle diameter threshold of  34 $\mu$m given the nominal depth of the somatic cell layer, and an absolute number of 15 cells in the largest  2D cross-section).}

{To assess the critical cell number relative to oocyte diameter, we computed a paving index, $e_O$, that represents the average contact length between a somatic cell and the oocyte: $e_O = \frac{\pi d_O}{N^s_g}$, where $d_O$ is the oocyte diameter and $N^s_g$ is the number of cells counted on the largest section. With a 24 $\mu$m diameter and 15 cells, we get a higher
bound of $5 \mu$m for $e_O$, that we applied as a filter to the rough data.}

		{On the other side of the activation process, we have only retained the strictly mono-layered type II follicles. Indeed, we intend to}
		deal with a final cell number as close as possible to the number reached at the first time when all flattened cells have transitioned to cuboidal cells (hence to the extinction time in the model),  Yet, due to the oocyte enlargement and the resulting increased capacity of the first layer, one cannot preclude that a significant amount of cuboidal cells have been generated after the end of the transition period.
		
		{Combining these criteria, we get the  dataset described on Figure \ref{fig:plancf}, which} illustrates the repartition of the data points according to the follicle type and sheep strain in each phase plane  ($C$, $F$) ($C$, $p_C$), ($N$, $p_C$).
		
	   \paragraph{\textit{In silico} datasets} 
	    {In addition to the experimental dataset, we have constructed \textit{in silico} datasets from simulations of SDE \eqref{eq-SDE} in a way that mimics the experimental protocol} {(see details in Appendix~\ref{appendix-insilico}).} 

	   In the sequel these datasets will be used as benchmark tools for the parameter identifiability study and the statistical comparison between the submodels and full model. In any case, the set of estimated parameters will match the set of cell events included in the model used to generate the \textit{in silico} dataset. For instance, we will estimate the values of parameters $\alpha_2$ and $\gamma$ on the two datasets generated from submodel  $(\mathcal{R}_1,\mathcal{R}_3,\mathcal{R}_4)$ {($\alpha_1$ will be fixed to $1$ in the sequel)}.
	   
\subsection{Likelihood method}\label{subsec-likelihood-method}

    Since the experimental dataset is made of time-free observations, we are going to confront the model to the data using only the information on some state space values taken by the process, without their corresponding time information. This notion is intrinsically related to the embedded Markov chain which we detail below. We will use this Markov chain to compute a likelihood function.
	Note that the proliferative cell population increases by one cell at each event ($\mathcal{R}_1$, $\mathcal{R}_2$, $\mathcal{R}_3$ or $\mathcal{R}_4$), while the precursor cell population can either remain constant ($\mathcal{R}_3$ or $\mathcal{R}_4$) or decrease by one ($\mathcal{R}_1$ or $\mathcal{R}_2$). The proliferative cell population $C$ can thus be used as an event counter. 
	Indeed, as a continuous-time Markov process, $X$ (defined in Eq.~\ref{eq-SDE}) can be decomposed into an embedded Markov chain  $(F_n,C_n)_{n\in\mathbb{N}}$ and a sequence of random jump times $(s_{n})_{n \in \mathbb{N}}$ with
	\begin{equation*}
	    s_{n+1} = s_n + \mathcal{E}\left((\alpha_1 + \alpha_2) F_n + \beta \frac{F_n C_n}{F_n + C_n} + \gamma C_n\right), \quad s_0 = 0.
	\end{equation*}
	Note that the sequence of jump times $(s_{n})_{n \in \mathbb{N}}$ corresponds exactly to the sequence of jump times associated with process C, and
	\begin{equation*}
	    C(t)=\sum_{n\in \mathbb{N}}  \mathds{1}_{s_n\leq t}\,, \quad C_n=n\,.
	\end{equation*}\\
    Given that $C_n=n$ is deterministic, it is clear that the precursor cell population $F_n$ (alone) is also a (non-homogeneous) Markov chain. To clarify the link with the data, we will index the embedded chain $F_n$ by the number of proliferative cells $C=c$, rather than by the number of events that occurred: 
    let $F_c$ be the random variable corresponding to the number of precursor cells given that there are $c\in \mathbb{N}$ proliferative cells. 
	According to the dichotomy between the two division events ($\mathcal{R}_3$, $\mathcal{R}_4$) and the two transition events ($\mathcal{R}_1$, $\mathcal{R}_2$), we deduce the law of $F_c$ at the ``pseudo-time'' $C = c$ from the law of $F_{c-1}$ at the ``pseudo-time'' $C = c-1 $ as follows: for all $(f,c) \in \mathcal{S}$,
	\begin{equation}\label{eq-like}
	\mathbb{P}\left[ F_c = f \right] =  \underbracket{q_{f+1,f}(c-1) \mathbb{P}\left[ F_{c-1} = f + 1 \right]}_{\text{transition }} 
	+ 	\underbracket{q_{f,f}(c-1)  \mathbb{P}\left[ F_{c - 1}= f \right] }_{\text{asymmetric/symmetric division}},
	\end{equation}
	where
	\begin{multline}\label{eq-q}
		q_{f+1,f}(c) =  \frac{\alpha_1 (f+1) + \beta \frac{(f+1) c}{f+ 1+c}}{(\alpha_2 + \alpha_1) (f+1) + \gamma c  + \beta \frac{(f+1) c }{f+ 1+c}}, \\
		q_{f,f}(c) = \frac{ \alpha_2 f + \gamma c}{(\alpha_2 + \alpha_1)f + \gamma c+ \beta\frac{fc}{f + c} }.
	\end{multline}
	Hence $(F_c)_{c \in \mathbb{N}}$ is a non-homogeneous discrete time Markov chain. Notice that the law of  $C_{\tau}$, the number of proliferative cells at the extinction time of the precursor cells, corresponds to the law of the first ``pseudo-time'' $c$ such that $F_c=0$, e.g. $C_{\tau}=\inf \{ c \in \mathbb{N}^*,\, F_c = 0 \}$.
	
	In addition to Eq.~\eqref{eq-like}, to compute the law of  $(F_c)$, we need to specify an initial condition $F_0$. {As detailed in Section \ref{subsect-dataset_description}, the initial pool of flattened cells is highly variable.} {To limit the number of parameters,} we assume that the initial number of precursor cells follows a truncated Poisson law on $\mathbb{N}^*$ (with a single parameter $\mu \in \mathbb{R}_+$) given by, for all $f \in \mathbb{N}^*$,
	\begin{equation}\label{eq-poisson-trun-law}
		\mathbb{P}\left[ F_0 = f \right] =  \frac{\mu^f}{(e^\mu - 1)f!}.
	\end{equation}
	Then, we can use Eq.~\eqref{eq-like} to compute $\mathbb{P}[F_c = f]$ by recurrence from the initial probability vector $(\mathbb{P}[F_0 = i])_{i \in \llbracket 0, c + f \rrbracket } $.  Hence, we have built a discrete time Markov chain $(F_c)_{c \in \mathbb{N}}$ from model \eqref{Model_FC} adapted to our time-free observations. 
	
	As can be seen from Eq.~\eqref{eq-q}, the timescale cannot be inferred, so that we fix arbitrarily $\alpha_1 = 1$,  whatever the dataset, to obtain dimensionless parameters. The time unit of the remaining parameters is thus relative to the timescale of one spontaneous transition event, and their estimated  values may  depend on the specific dataset (experimental or {\it in silico}).
	
	Finally, we suppose that all data points are independent of one another, and that the observations are free of measurement errors. 
    {Therefore, accordingly to Eq.~\eqref{eq-like}-\eqref{eq-q}-\eqref{eq-poisson-trun-law}, our statistical model assumes that the observed variability is due to a random initial number of precursor cells, and to the occurrences of random cell events among cell transitions and cell divisions}. 
	We obtain the following likelihood function 
	\begin{equation}\label{eq:likelihood_total}
		\mathcal{L}((f_i,c_i)_{i=1..N};\theta) := \mathbb{P}\left[(f_i,c_i)_{i=1..N} | \theta \right] = \prod_{i = 1}^{N} \mathbb{P}\left[ F_{c_i} = f_i| \theta \right]\,,
	\end{equation}
{for $N$ data points $(f_i,c_i)_{i=1..N}$ and the parameter vector $\theta$ depending on each submodel.}

{We infer the parameter values using the maximum likelihood estimator (MLE), and apply the practical approach based on profile likelihood estimate (PLE) to analyze the parameter identifiability and assess confidence intervals \cite{raue_structural_2009}. We also perform model selection using classical AIC and BIC criteria to discriminate between the full model and different submodels. The whole procedure is described in Appendix~\ref{appendix:fittingproc}}.

{Note that the initial condition parameter $\mu$ can be either estimated together with the other parameters from a given dataset, using the likelihood given in Eq.~\eqref{eq:likelihood_total}, or, alternatively, from the cell number of the primordial follicles only.
        In the latter case, with the law of $F_0$ given by \eqref{eq-poisson-trun-law}, we obtain the likelihood function 
        \begin{equation}\label{eq:likeli_init}
            \mathcal{L}_{ini}((f_i,0)_{i=1..N'};\mu) := \prod_{i \in \llbracket 1, N' \rrbracket } \frac{\mu^{f_i}}{(e^\mu - 1)f_i !}\,,
        \end{equation}
        for $N'$ data points $(f_i,0)_{i=1..N'}$. 
        From the likelihood defined in Eq. \eqref{eq:likeli_init}, we deduce MLE and PLE to infer the value of $\mu$ solely from the primordial follicle data.}

\subsection{Fitting results}\label{ssec:fittingresuts}

	In this subsection, we {present our fitting results using the procedure described in subsection \ref{subsec-likelihood-method}} for several submodels derived from model \eqref{Model_FC}:
	\begin{itemize}
	    \item two-event submodels, including the spontaneous transition event together with either the asymmetric $(\mathcal{R}_1,\mathcal{R}_3)$ or symmetric division $(\mathcal{R}_1,\mathcal{R}_4)$;
	    \item three-event submodels, {the linear submodel $(\mathcal{R}_1,\mathcal{R}_3,\mathcal{R}_4)$ and the two-nonlinear one,} including auto-ampli\-fied transition events, together with either the asymmetric $(\mathcal{R}_1,\mathcal{R}_2,\mathcal{R}_3)$ or symmetric $(\mathcal{R}_1,\mathcal{R}_2,\mathcal{R}_4)$ division event;
	    \item the full model \eqref{Model_FC}$=(\mathcal{R}_1,\mathcal{R}_2,\mathcal{R}_3,\mathcal{R}_4)$
	\end{itemize}
	
	{The fitting results {obtained with the total likelihood (Eqs. \eqref{eq-like}-\ref{eq:likelihood_total})} on the experimental datasets are shown in Figure
			\ref{fig:bestfit} for submodels $(\mathcal{R}_1,\mathcal{R}_3)$,  $(\mathcal{R}_1,\mathcal{R}_4)$ and the full model . {The corresponding} fitting results for the {\it in silico} datasets for the same submodels are provided in Figure~\ref{fig:datafitt_2parms_insilico} (Appendix~\ref{appex:detail_cal}). The fitting results for the three-event submodels for both the experimental and \textit{in silico} datasets are provided in Appendix~\ref{appex:detail_cal}. For both the Wild-Type and Mutant datasets, a visual inspection shows that submodel $(\mathcal{R}_1, \mathcal{R}_4)$ leads to a ``direct'' transition, followed by prolonged cell proliferation after precursor cell extinction, while, with submodel $(\mathcal{R}_1, \mathcal{R}_3)$, there is a higher probability that the total number of cells increases before precursor cell extinction.} {The model selection criteria, summarized in Table \ref{table-comp-wt-vs-mut}, shows that all submodels without cell event $\mathcal{R}_4$ can be safely rejected. } {The visual inspection of  Figure~\ref{fig:bestfit} leads to the following explanation. If event $\mathcal{R}_4$ is present, as in submodel $(\mathcal{R}_1, \mathcal{R}_4)$, the proliferative cells can keep dividing after the extinction of the precursor cells (line $F=0$). Once the precursor cell number reaches zero for a given $c$, all remaining points $(0,c')$ for $c'\geq c$ are reached with probability one, which results in a high contribution of all $(0,c)$ data points to the maximum likelihood. In contrast, if event $\mathcal{R}_4$ is not present, as in submodel $(\mathcal{R}_1, \mathcal{R}_3)$, the process stops as soon as the precursor cell population $F$ gets extinct, which prevents the likelihood of all $(0,c')$ points from being close to one (they rather take all intermediate values).} {This observation is consistent with the fitting results of the \textit{in silico} datasets (Figure~\ref{fig:datafitt_ana_insilico} in Appendix~\ref{appex:detail_cal})}.\\
			{The model selection criteria further suggest that the best models associated with the experimental datasets are the full model and the three-event linear submodel $(\mathcal{R}_1,\mathcal{R}_3,\mathcal{R}_4)$. The two-event submodel $(\mathcal{R}_1,\mathcal{R}_4)$ appears to be a possible alternative but still less relevant than the two others.} {In Figure~\ref{fig:bestfit}, we observe that the trajectories associated with an intermediate level of cell proliferation before precursor cell extinction are more likely in the full model than in the two-event submodel $(\mathcal{R}_1,\mathcal{R}_4)$, with a more pronounced effect for the mutant subset than the wild-type subset. We will come back to this consideration in section \ref{ssec:prediction}.}\\
{\noindent The parameter identifiability study (detailed in Appendix~\ref{appex:detail_cal}) leads to the following results:}
	{\begin{itemize}
	        \item {The initial condition parameter $\mu$ is always practically identifiable (Figures \ref{fig:datafitt} and \ref{fig:datafitt_3parm}) and the MLE yields similar values from one submodel to another (see Tables \ref{table-wt} and \ref{table-mut}). Moreover, its fitted value is close to the true one for the \textit{in silico} datasets, with some small bias in some cases  (Figures~\ref{fig:datafitt_2parms_insilico} and \ref{fig:datafitt_3parm}).}
	        \item The asymmetric division parameter $\alpha_2$ is tightly identifiable only in submodel $(\mathcal{R}_1,\mathcal{R}_3)$,  (Figure~\ref{fig:datafitt}). {However, significant bias occurs, as revealed by the \textit{in silico} datasets yet (Figures~\ref{fig:datafitt_2parms_insilico} and \ref{fig:datafitt_3parm}).}
	        \item The symmetric division rate $\gamma$ is rarely practically identifiable, yet an upper-bound can always be obtained (Figures \ref{fig:datafitt}, \ref{fig:datafitt_2parms_insilico} and \ref{fig:datafitt_3parm}).
	        \item The self-amplified transition rate $\beta$ is not identifiable in any case (Figures \ref{fig:datafitt}, \ref{fig:datafitt_2parms_insilico} and \ref{fig:datafitt_3parm}).
	    \end{itemize}
	}
	
	{This one-dimensional parameter identifiability analysis hides however more subtle parameter constraints. The self-amplification transition rate is actually constrained to be greater than the symmetric division rate $\gamma$, as shown in the two-dimensional profile likelihood analysis in Figure \ref{fig:ple_fit_gamma_beta1} in Appendix~\ref{appex:detail_cal}. This {result confirms} the tendency {observed with} the best fit trajectories in Figure~\ref{fig:bestfit}, that favor transition over proliferation.}
	
	{The fitting results obtained on the initial condition parameter $\mu$ from primordial follicle data (using the likelihood given by Eq.~\eqref{eq:likeli_init}) is shown in Figure~\ref{fig:primary_primordial}. We have obtained identifiable parameter values with each submodel, yet associated with  broader confidence intervals than with the global fitting approach given by Eqs.\eqref{eq-like}-\eqref{eq:likelihood_total}. As expected, using more information {reduces} the uncertainty, hence the confidence intervals are smaller when the whole datasets are used (for all models and subsets considered).}
	
	\begin{figure}[htb!]
	\centering
	\includegraphics[width=0.9\linewidth]{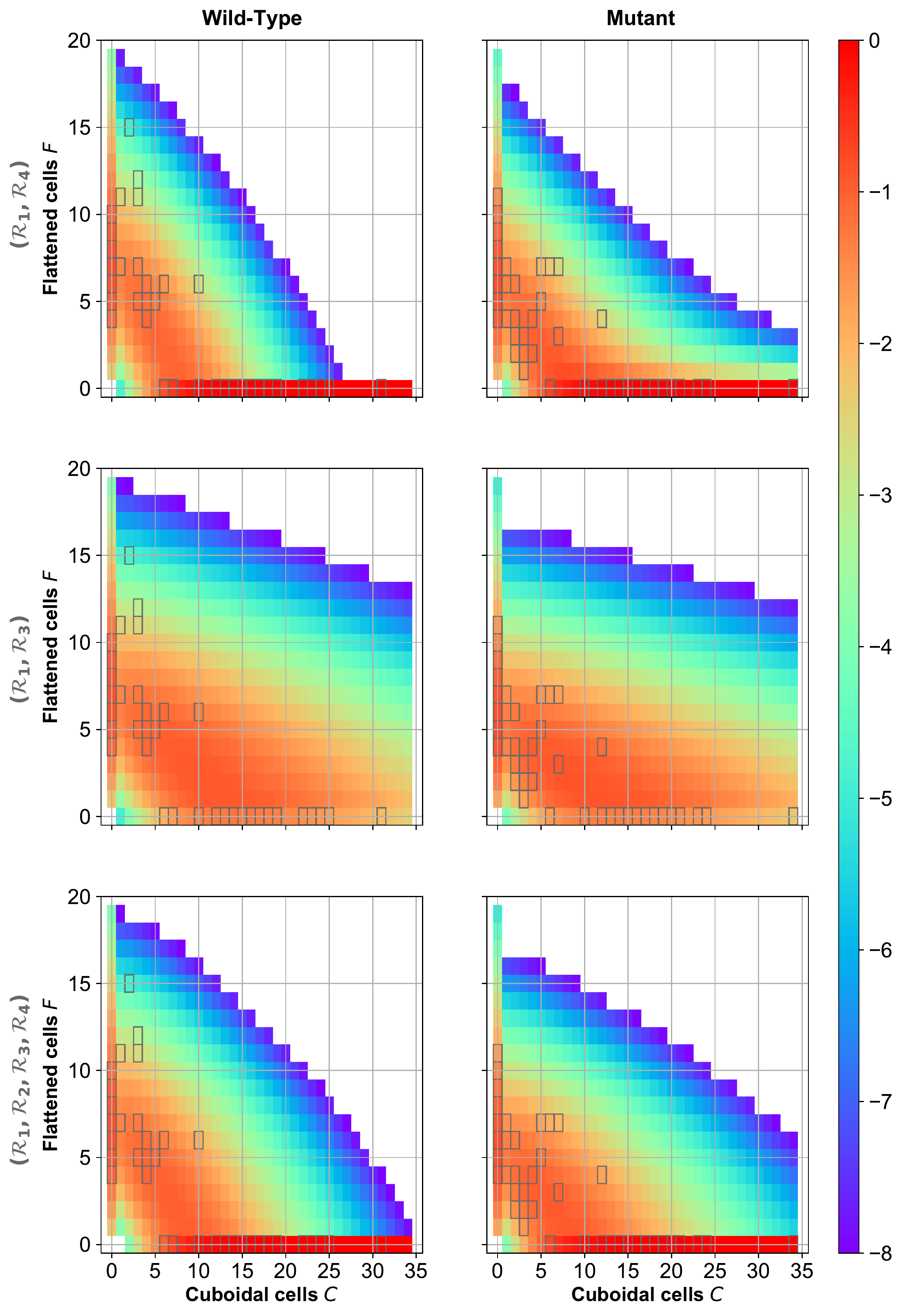}
	\caption{{\textbf{Two-event submodels and full model: best fit trajectories.} Using Eqs.\eqref{eq-like}-{\eqref{eq-poisson-trun-law}}, we compute each probability $\mathbb{P}\left[F_c = f\right] $ for submodels $(\mathcal{R}_1, \mathcal{R}_4)$ (top-panels), $(\mathcal{R}_1, \mathcal{R}_3)$ (middle panles) and the full model $(\mathcal{R}_1, \mathcal{R}_2, \mathcal{R}_3, \mathcal{R}_4)$ (bottom panels) with their respective MLE parameter set for Wild-Type dataset (left column) and Mutant dataset (right column). Each empty gray square corresponds to a data point. The colormap corresponds to the probability values $\mathbb{P}\left[F_c = f \right] $ in log10 scale.}}
		\label{fig:bestfit}
	\end{figure}

	\begin{figure}
				\centering
			\includegraphics[width=\linewidth]{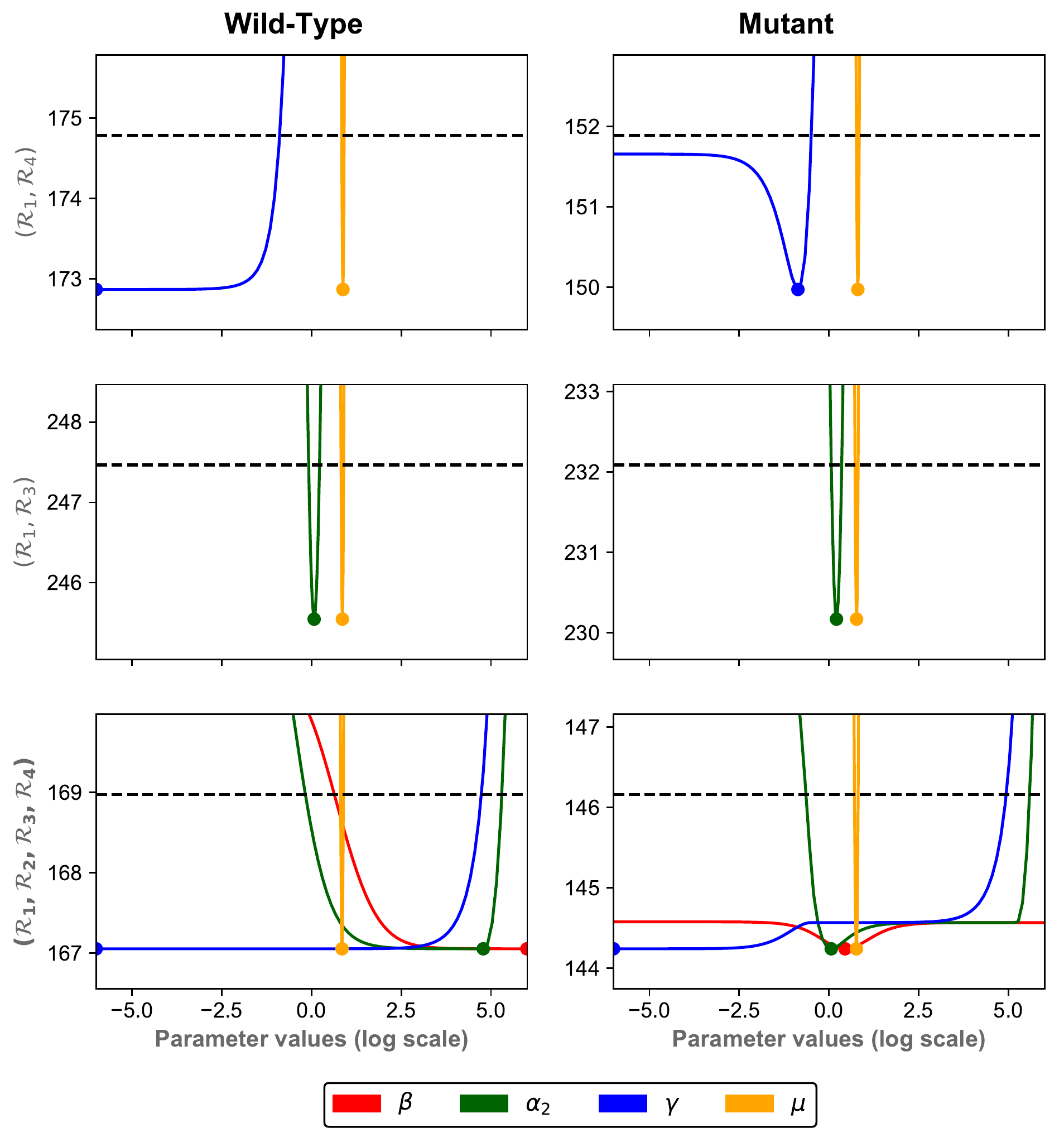}
				\caption{{\textbf{Two-event submodels and full model: PLE}. Each panel represents the PLE, in log10 scale, obtained from the experimental datasets, and either submodel $(\mathcal{R}_1,\mathcal{R}_4)$ (top panels), $(\mathcal{R}_1,\mathcal{R}_3)$ (center panels), or $(\mathcal{R}_1,\mathcal{R}_2, \mathcal{R}_3,\mathcal{R}_4)$ (bottom panels).
				The dashed black line represents the 95\%-statistical threshold. 
				Orange solid lines: PLE values for the initial condition parameter $\mu$; blue solid lines: PLE values for the symmetric cell proliferation rate $\gamma$; green solid lines: PLE values for the asymmetric cell division rate  $\alpha_2$; red solid lines: PLE values for the self-amplification transition rate $\beta$.
				 The colored points represent the associated MLE. }}
				\label{fig:datafitt}
			\end{figure}

	\begin{table}[htb!]
			\begin{tabular}{|c|| c|c|c||c|c|c| }
			    \hline 
                 &  \multicolumn{3}{|c||}{ \textbf{\sc{Wild-Type}}} & \multicolumn{3}{c|}{\textbf{\sc{Mutant}}}\\
				\hline 
				\vspace{0.02cm}
				Model & \tiny{-$\log \mathcal{L}(\theta; \mathbf{x})$} & AIC & BIC &  \tiny{-$\log \mathcal{L}(\theta; \mathbf{x})$}  & AIC & BIC\\ 
				\hline 
				\hline \\[-2.5ex]
				\footnotesize{$(\mathcal{R}_1, \mathcal{R}_4) $} & 172.87 &  \shortstack{ 349.74 \\ \tiny{$w = 0.02$} \\ \tiny{$\Delta = 7.6$} } & \shortstack{\textcolor{blue}{354.74} \\ \tiny{\textcolor{blue}{$w =0.15$}} \\ \tiny{\textcolor{blue}{$\Delta  = 3.0$}} }  & 149.97 &  \shortstack{303.94 \\ \tiny{$w = 0.08$} \\ \tiny{$\Delta = 8.8$} } & \shortstack{ \textcolor{blue}{308.73} \\ \tiny{\textcolor{blue}{$w =0.03$}} \\ \tiny{\textcolor{blue}{$\Delta = 6.4$}} } \\ 
				\hline \\[-2.5ex]
				\footnotesize{$(\mathcal{R}_1, \mathcal{R}_3) $} & 245.54 &   \shortstack{495.09 \\ \tiny{$w< 10^{-10}$} \\ \tiny{$\Delta >> 10$}}& \shortstack{500.09  \\ \tiny{$w < 10^{-10}$} \\ \tiny{$\Delta >> 10 $}} & 230.17 &  \shortstack{464.34 \\ \tiny{$ w < 10^{-10} $} \\ \tiny{$\Delta >> 10 $ }}  & \shortstack{ 469.13 \\ \tiny{$w < 10^{-10} $} \\ \tiny{$\Delta >> 10 $}}\\ 
				\hline \\[-2.5ex]
				\footnotesize{$(\mathcal{R}_1, \mathcal{R}_2, \mathcal{R}_4) $} & 172.77 & \shortstack{351.54 \\ \tiny{$w =0.008$} \\ \tiny{$\Delta = 9.44$}} & \shortstack{359.04 \\ \tiny{$w < 10^{-10}$}  \\ \tiny{$\Delta = 7.4$}}& 148.14 & \shortstack{302.27 \\ \tiny{$w = 0.02$} \\ \tiny{$\Delta = 7.1$}} & \shortstack{309.46  \\ \tiny{$w=0.02$} \\ \tiny{$\Delta = 7.1$}}\\ 
				\hline \\[-2.5ex]
				\footnotesize{$(\mathcal{R}_1, \mathcal{R}_2, \mathcal{R}_3) $} &  242.51 & \shortstack{491.02  \\ \tiny{$w < 10^{-10}$}  \\ \tiny{$\Delta >> 10$}} & \shortstack{498.52 \\ \tiny{$w < 10^{-10}$} \\ \tiny{$\Delta >> 10 $} } &  229.44  & \shortstack{464.89 \\ \tiny{$w < 10^{-10} $} \\ \tiny{$\Delta >> 10 $}}& \shortstack{472.07 \\ \tiny{$w< 10^{-10}$} \\ \tiny{$\Delta >> 10 $}}\\ 
				\hline \\[-2.5ex]
				\footnotesize{$(\mathcal{R}_1, \mathcal{R}_3, \mathcal{R}_4) $} & 170.58  & \shortstack{\textcolor{blue}{347.16} \\ \tiny{\textcolor{blue}{$w =0.07$}}  \\ \tiny{\textcolor{blue}{$\Delta = 5.0$}}} & \shortstack{\textcolor{blue}{354.66} \\ \tiny{\textcolor{blue}{$w =0.15$}} \\ \tiny{\textcolor{blue}{$\Delta = 3.0$}} }& 144.58  & \shortstack{ \textcolor{red}{295.15} \\ \tiny{\textcolor{red}{$w =0.64$}}}  & \shortstack{\textcolor{red}{302.34}\\ \tiny{\textcolor{red}{$w =0.81$}}} \\ 
				\hline \\[-2.5ex]
			    \footnotesize{$(\mathcal{R}_i)_{i \in \llbracket 1, 4 \rrbracket} $} & 167.05  & \shortstack{\textcolor{red}{342.10} \\ \tiny{\textcolor{red}{$w =0.90$}}} & \shortstack{\textcolor{red}{351.68} \\ \tiny{\textcolor{red}{$w =0.68$} }}& 144.24 & \shortstack{\textcolor{blue}{296.48} \\ \tiny{\textcolor{blue}{$w =0.33$}} \\ \tiny{\textcolor{blue}{$\Delta = 1.3$}}} & \shortstack{\textcolor{blue}{306.06} \\ \tiny{\textcolor{blue}{$w =0.12$}} \\ \tiny{\textcolor{blue}{$\Delta = 3.7$}} }\\ 
				\hline 
				\hline
			\end{tabular} 
			\caption{\textbf{Model comparison analysis.} For each experimental subset and each submodel, we compute both the  Akaike information criterion (AIC) and Bayesian information criterion (BIC), {the AIC and BIC differences $\Delta^{AIC}_i := AIC_i - AIC_{\min}$ and $\Delta^{BIC}_i = BIC_i - BIC_{\min}$, and the corresponding Akaike and Bayesian weights $w^{AIC}_i = \frac{\exp(-0.5 \Delta^{AIC}_i)}{\sum_{k = 1}^6\exp(-0.5 \Delta^{AIC}_k)} $ and $w^{BIC}_i = \frac{\exp(-0.5 \Delta^{BIC}_i)}{\sum_{k = 1}^6\exp(-0.5 \Delta^{BIC}_k)} $ following \cite[(Chapter 2 and 3)]{burnham_model_2003}.   
			{The best models are highlighted in red and the remaining selected models in blue (details are provided in Appendix \ref{appendix:fittingproc}).}
			}}\label{table-comp-wt-vs-mut}
		\end{table}

\subsection{Model prediction}\label{ssec:prediction}
    
    In this subsection, we use the {MLE together with their confidence interval obtained with the PLE}
    of the best models (the two linear submodels $(\mathcal{R}_1,\mathcal{R}_4)$ and $(\mathcal{R}_1,\mathcal{R}_3,\mathcal{R}_4)$ and the full model) 
    to infer information on the experimental subsets.
    
    \subsubsection{Distribution of the initial condition}
        In Figure \ref{fig:primary_primordial}, we compare for both the Wild-Type and Mutant subset the distributions derived from model $(\mathcal{R}_1,\mathcal{R}_4)$, $(\mathcal{R}_1,\mathcal{R}_3, \mathcal{R}_4)$ and $(\mathcal{R}_1,\mathcal{R}_2,\mathcal{R}_3, \mathcal{R}_4)$, {using the {whole} data, together with the distribution inferred from the primordial follicle data only}. From the top panels of Figure \ref{fig:primary_primordial}, we observe that in all cases, there is an overestimation {of the head and tail of the distribution of $F_0$, which suggests that a more peaked distribution than the {truncated} Poisson distribution {would be more suitable for the initial condition.} }
        {The distribution inferred from the primordial follicle data only is slightly closer to the datapoint than the distribution with $\mu$ inferred using the complete follicle data (as expected), as assessed by the evaluation of the likelihood \eqref{eq:likeli_init} at each MLE, shown in the lower panels of Figure \ref{fig:primary_primordial}.}\\
        {A detailed inspection of the lower panels of Figure \ref{fig:primary_primordial} shows furthermore that the likelihood \eqref{eq:likeli_init} based on the primordial follicle data cannot discriminate between the Wild-Type and Mutant subset. However, using the likelihood \eqref{eq-like}-\eqref{eq:likelihood_total} with the whole data induces} a shift of approximately one cell in average, in opposite directions for the Wild-Type and Mutant subset: for the Wild-Type subset, the mean cell number is found to be greater when the whole data are used, while for the Mutant subset, the mean cell number is found to be smaller (for all three models considered). 
        {Hence, considering the subsequent follicle trajectories, shaped by transition and proliferation, modifies the most likely value of $\mu$ and can discriminate the Wild-Type subset from the  Mutant subset. The precise value of $\mu$ is biologically important, since it can be considered as the equivalent of the number of founder cells in lineage studies. Indeed, until ovulation (where the total cell number is on the order of several millions in sheep), there will not be any recruitment of somatic cells, and all cells with derive from the initial flattened cells.}

    \begin{figure}[htb!]
		\centering
		\includegraphics[width=11.9cm]{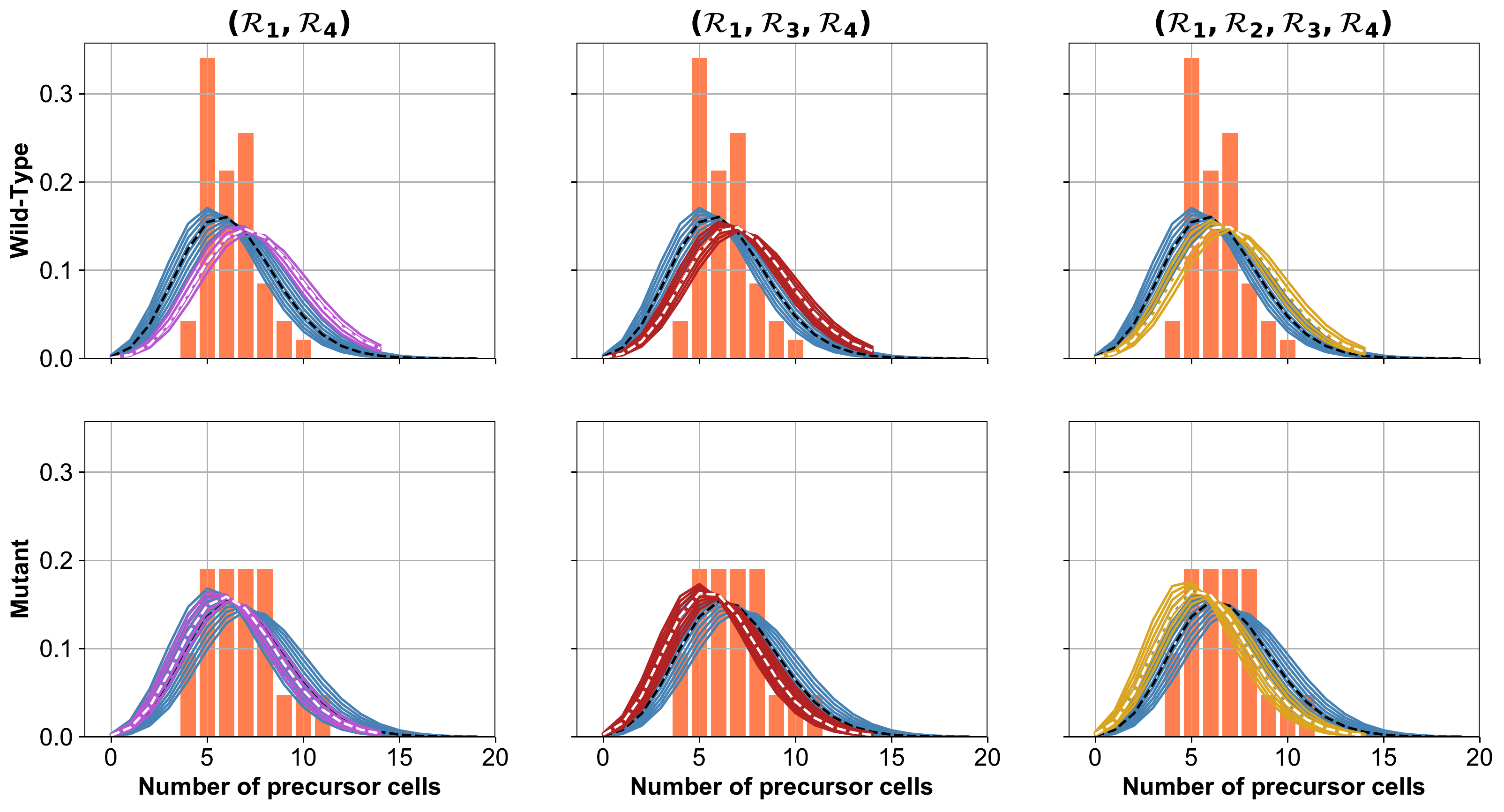}
		\includegraphics[width=11.9cm]{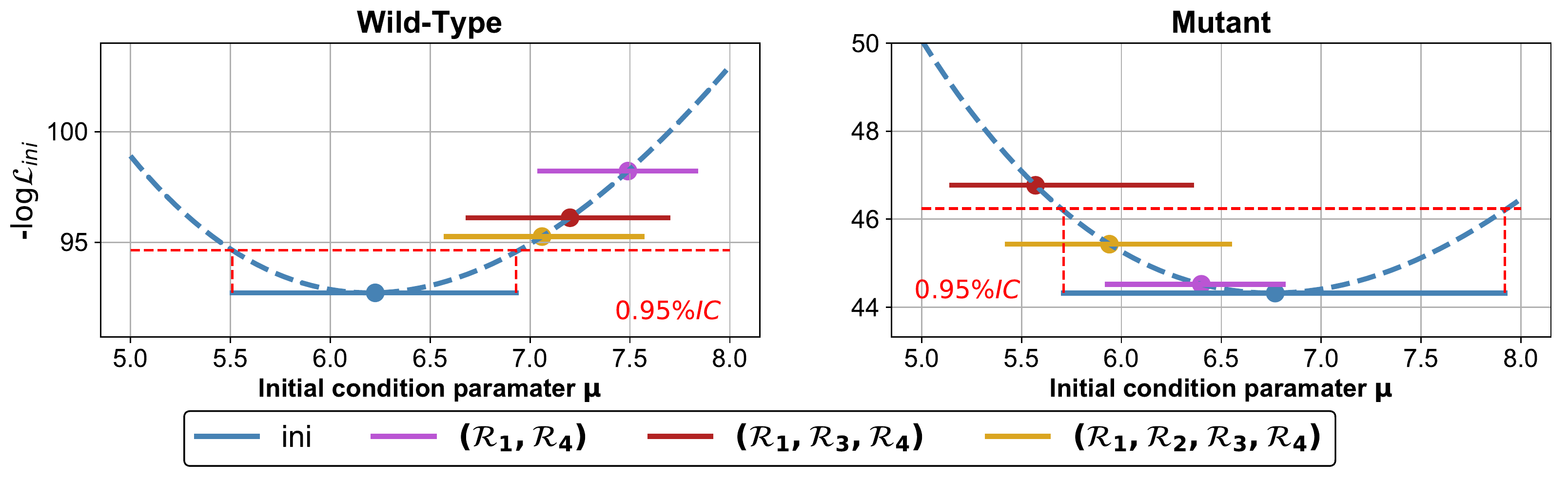}
		\caption{\textbf{Estimates on the initial condition parameter $\mu$ and initial distribution}. Top and middle panels: experimental data histograms of the number of precursor cells in primordial follicles with inferred Poisson distributions. Histograms with coral-colored bars: initial precursor cell number in primordial follicles for Wild-Type (top panels) and Mutant (middle panels) subsets. For submodels $(\mathcal{R}_1,\mathcal{R}_4)$ (left panels), $(\mathcal{R}_1,\mathcal{R}_3, \mathcal{R}_4)$ (center panels) and $(\mathcal{R}_1,\mathcal{R}_2,\mathcal{R}_3, \mathcal{R}_4)$ (right panels), we plot: in white dashed lines, the truncated Poisson distribution \eqref{eq-poisson-trun-law} with MLE using Eqs.\eqref{eq-like}-\eqref{eq:likelihood_total} ($\mu$ is estimated together with the remaining parameters) and, in colored solid lines, the truncated Poisson distribution with $\mu$ in the associated confidence interval of the MLE; in black dashed lines: the truncated Poisson distribution with MLE using Eq.~\eqref{eq:likeli_init}  ($\mu$ is estimated only with primordial dataset) and, in gray solid lines, the truncated Poisson distribution with $\mu$ in the associated confidence interval of the MLE.
		Bottom panels: Wild-Type (left panel), Mutant (right panel). Cyan dashed lines: log-likelihood function  $\mathcal{L}_{ini}$ given by Eq.~\eqref{eq:likeli_init} (primordial data set only); red dashed lines: $ 95\%$ confidence interval;  colored solid lines (resp. filled circles): confidence intervals of $\mu$ (resp. MLE) for each submodel with evalution of the log-likelihood function  $\mathcal{L}_{ini}$ at the MLE. }
		\label{fig:primary_primordial}
	\end{figure}
		
	\subsubsection{Proliferative cell proportion: reconstruction of time}\label{ssec:timereconstruc}
	In Figure \ref{fig:modelpred_twoparms}, we represent the predicted changes in the proliferative cell proportion with respect to time. 
	{For sake of readability, these predictions are derived from the deterministic formulation of the full model (Eq.~\eqref{eq-pc-time}). We expect that a similar trend would be observed with the stochastic CTMC formulation.} For each model, we superimpose the time trajectories corresponding to the parameter combinations for which the PLE is below the $95\%$ threshold. In both the Wild-Type and Mutant cases, despite the uncertainty affecting the model parameters for the two linear submodels (left and right upper panels), the dynamics just exhibit small uncertainties: the proportion of proliferative cells reaches $ 50\%$-$70\%$  in one time unit, which corresponds to the time unit of a single spontaneous transition event. This might due partly to the fact that  parameter $\gamma$ is partially identifiable and is estimated to relatively low values. In contrast, the lack of parameter identifiability of the full model results in a huge uncertainty on the dynamics, that can be up to 5 order of magnitude faster than a single spontaneous transition event: the proportion of proliferative cells reaches $ 50\%$ between $10^{-6}$ and $1$ time unit. Indeed, cell event $\mathcal{R}_2$ (controlled  by parameter $\beta$) can speed up the transition dynamics, and cell event $\mathcal{R}_3$ (controlled by parameter $\alpha_2$) can trigger the first transition, leading to a possible fast activation which avoids the bottleneck of the spontaneous transition timescale ($\alpha_1=1$). {It is difficult to instantiate these relative durations in physical time units. The only kinetic information available on the activation process is given by studies that have monitored the sequential apparition of different follicle types during fetal development. In wild-type animals, the first primordial follicles appear around 75 days of gestation, while the first primary follicles are observed around 100 days \cite{mcnatty_development_1995}. A 25 day-duration can thus be considered as close to the minimal duration. }
	No clear timescale separation between the Wild-type and Mutant dynamics can be revealed, although some parameter combinations are compatible with a faster transition in the Wild-Type case than in the Mutant case. {This is again compatible with monitoring studies, which observed that the times of apparition of both the first primordial and primary follicles are shifted compared with wild-type animals (they appear a little later), yet the delay in between does not appear to be significantly different.}
			
	\begin{figure}[htb!]
		\centering
		\includegraphics[width=\linewidth]{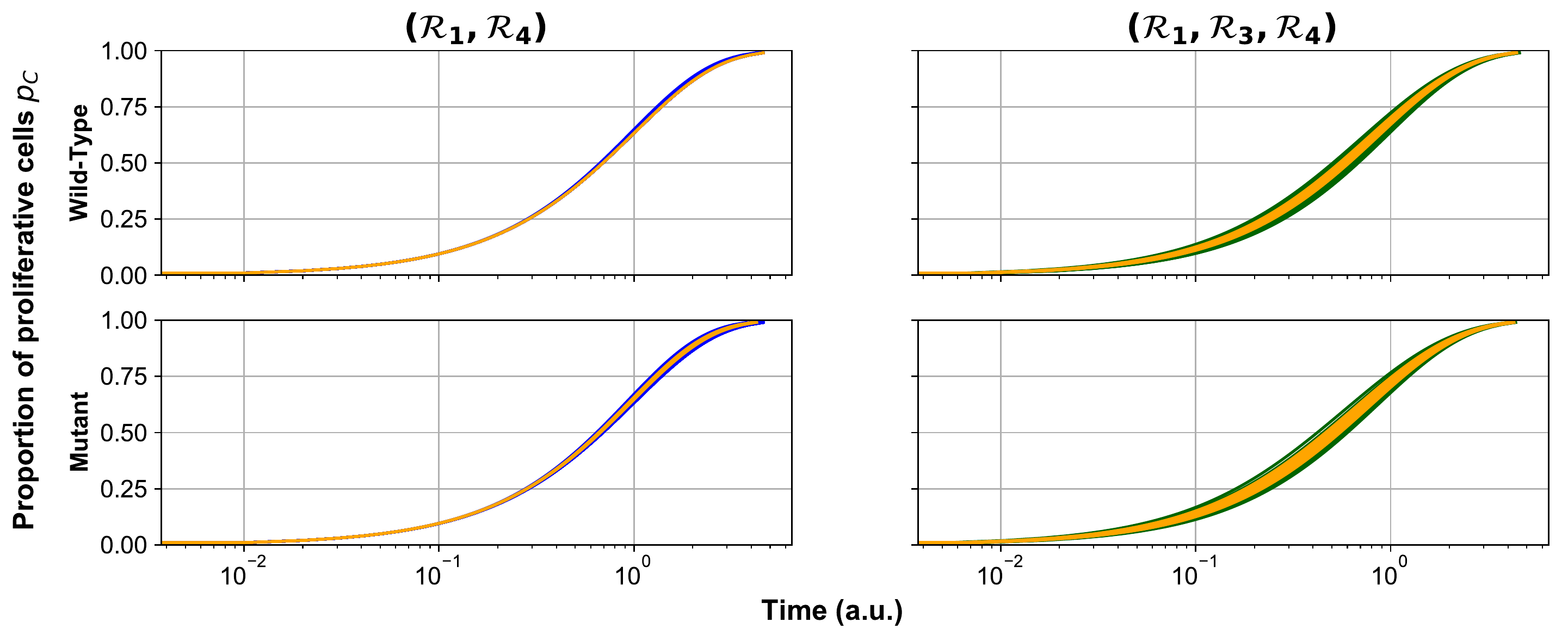}
		\includegraphics[width=\linewidth]{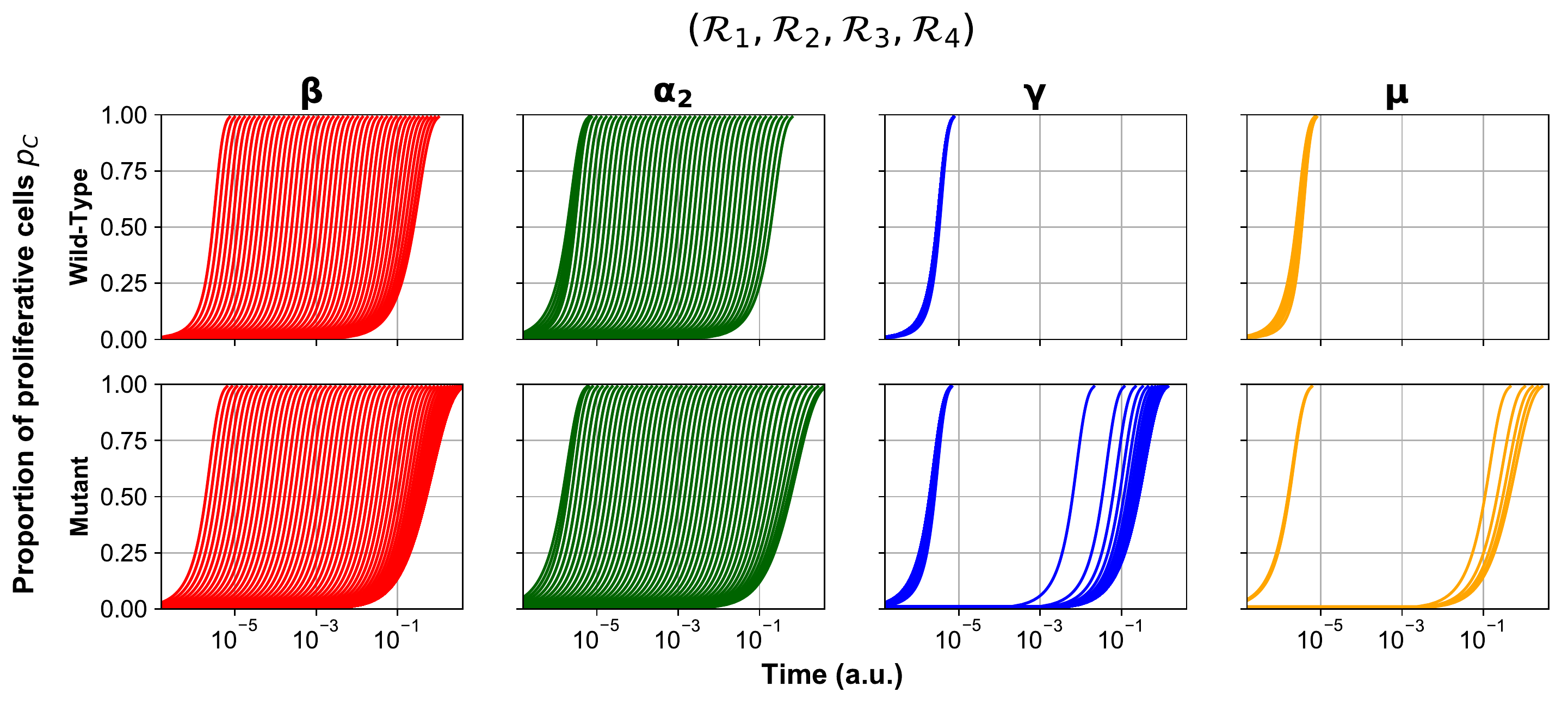}
		\caption{{\textbf{Dynamics of the proportion of proliferative cells $p_C(t)$}. For submodel $(\mathcal{R}_1, \mathcal{R}_4)$ (top left panels), $(\mathcal{R}_1,\mathcal{R}_3, \mathcal{R}_4)$ (top right panels) and whole model $(\mathcal{R}_1,\mathcal{R}_2,\mathcal{R}_3, \mathcal{R}_4)$ (bottom panels), we plot the deterministic proportion of proliferative cells $p_C(t)$ computed from Eq.~\eqref{eq-pc-time} with the fitted parameters lying in the MLE confidence interval associated with each PLE (see subsection \ref{subsec-likelihood-method} for details). Red lines: $p_C(t)$ with parameters in the PLE of the auto-amplified transition rate $\beta$; green lines: $p_C(t)$ with parameters in the PLE of the asymmetric division rate $\alpha_2$; blue lines: $p_C(t)$ with parameters in the PLE of symmetric division rate $\gamma$; yellow lines: $p_C(t)$ with parameters in the PLE of the initial condition parameter $\mu$.}}
		\label{fig:modelpred_twoparms}
	\end{figure}

    \subsubsection{Mean extinction time, mean number of cells at the extinction time and mean number of division events during activation}\label{ssec:moment_data}
        In Figure \ref{fig:modelpred_model_four}, we represent the mean number of proliferative cells, $\mathbb{E}\left[ C_{\tau}\right]$, and the mean number of division events during activation, $\mathbb{E}\left[ C_{\tau}-F_0\right]$, as a function of the mean extinction time $\mathbb{E}\left[ \tau\right]$, as predicted from the selected (sub)models $(\mathcal{R}_1,\mathcal{R}_4)$, $(\mathcal{R}_1,\mathcal{R}_3,\mathcal{R}_4)$ and $(\mathcal{R}_1,\mathcal{R}_2,\mathcal{R}_3,\mathcal{R}_4)$. These predictions are obtained from a direct stochastic simulation of the  trajectories of each model (with Gillespie algorithm, or SSA)\footnote{We use here the direct simulation rather than Algorithm \ref{algo-gr}, because the parameter range explored by the symmetric division rate $\gamma$ gets close to the theoretical necessary and sufficient condition $\gamma<\alpha_1+\beta$, while the Algorithm \ref{algo-gr} requires $2\gamma<\alpha_1+\beta$.}, using the parameter values obtained from the identifiability analysis, for which the PLE is below the $95\%$ threshold. For each subset (Wild-Type or Mutant), the predicted value for $\mathbb{E}\left[ C_{\tau}\right]$ is similar in each submodels and lies between 8 and 10 cells. Interestingly, the predicted value for $\mathbb{E}\left[ C_{\tau}\right]$ is approximately 6-8 cells lower than the empirical mean number of proliferative cells obtained directly from the primary follicle data (data points $(0,C)$
         with $F=0$) (Figure \ref{fig:modelpred_model_four}, top panels). This observation is consistent with the trajectory analysis performed from Figure \ref{fig:bestfit}, from which we have concluded that the activation process follows with high probability a trajectory reaching state $F=0$ with a low cell number, and characterized by direct transition and very little concomitant cell proliferation.
         {Similarly, {$\mathbb{E}\left[ C_{\tau}-F_0\right]$} is approximately 5-7 cells lower than the increase in the mean empirical number of cells between the primordial follicle datasets and primary follicle datasets (Figure \ref{fig:modelpred_model_four}, bottom panels).}
       {$\mathbb{E}\left[ \tau\right]$ in} the two linear submodels $(\mathcal{R}_1,\mathcal{R}_4)$ and $(\mathcal{R}_1,\mathcal{R}_3,\mathcal{R}_4)$ depends only on the initial condition and is estimated to a value around $2.5$ a.u. with a small uncertainty, similarly as in Figure \ref{fig:modelpred_twoparms}. In contrast, the full model yields a larger uncertainty on {$\mathbb{E}\left[ \tau\right]$}, with a confidence interval between $10^{-6}$ and $0.5$ a.u. for the Wild-Type subset, and between $10^{-6}$ and $2.5$ a.u. for the Mutant subset, consistently with the prediction on the dynamics of the proliferative cell proportion (Figure \ref{fig:modelpred_twoparms}). {From our theoretical results on parameter {sensitivity} in section \ref{part-general-case} (see Figure \ref{fig:odetimef}), we have found that $\beta$ has a profound impact on $\taup$. Any additional knowledge on the follicle activation duration would thus be valuable to further constraint the parameter uncertainty.}   

	\begin{figure}[htb!]
		\centering
		\includegraphics[width=\linewidth]{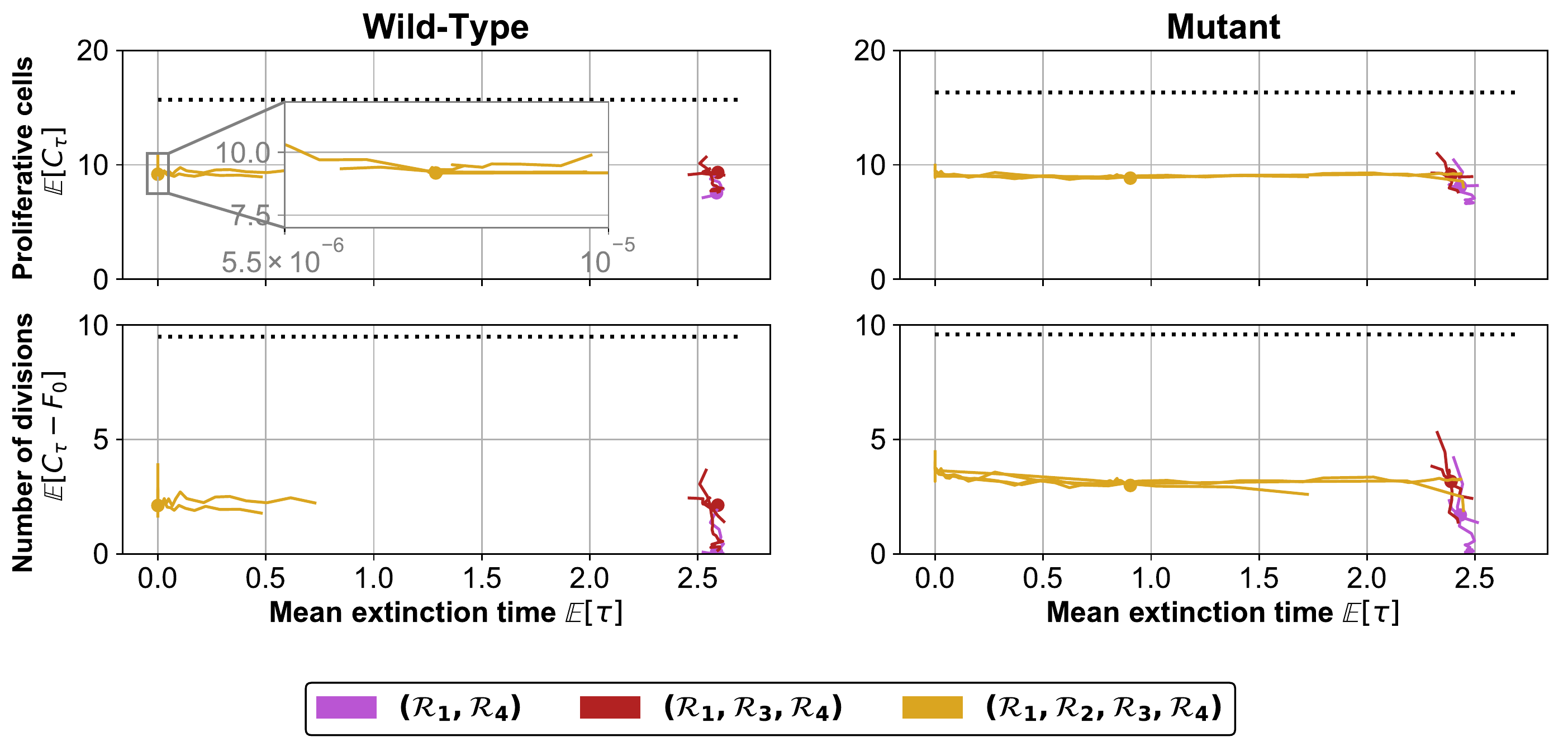}
		\caption{\textbf{Prediction of the  mean number of proliferative cells at the precursor cell extinction time, and mean number of division events during follicle activation.} We plot the mean number of proliferative cells at the extinction time $\mathbb{E}\left[\cextin\right]$ (top panels), and the mean number of division events before extinction $\mathbb{E}\left[\cextin-F_0\right]$ (bottom panels) as a function of the mean extinction time $\mathbb{E}\left[\taup\right]$ (left panels: Wild-Type; right panels: Mutant). For each selected submodel and for each parameter sets lying within a MLE confidence interval (see subsection \ref{subsec-likelihood-method} for details), we simulate 10,000 trajectories with the Gillespie algorithm, up to the extinction event $\{F=0\}$, and compute $\mathbb{E}\left[\taup\right] $, $\mathbb{E}\left[\cextin\right]$ and $\mathbb{E}\left[\cextin - F_0\right]$ from standard empirical mean estimates. Colored solid lines: $\mathbb{E}\left[\cextin\right]$, $\mathbb{E}\left[\cextin-F_0\right]$ as a function of $\mathbb{E}\left[\taup\right]$ for parameters lying in a MLE confidence interval; The filled circles represent the optimal MLE value. Dotted black lines: standard empirical mean estimate of proliferative cell numbers (top panels) and division events (bottom panels) before extinction using the primary follicles data set (all data points without flattened cells). }
		\label{fig:modelpred_model_four}
	\end{figure}

{Predictions on the mean number of divisions events could be in theory amenable to validation by experimental cell kinetics study. Such studies, enabling for instance to infer the possibly time-varying doubling times in cell populations have been performed for later developmental stages (\cite{turnbull_77} in sheep or \cite{pedersen_70} in mice). They cannot be conducted as such for the earliest stages because of the excessive slowness of cell events. Promising {\it ex-vivo}/{\it in vitro} devices \cite{morohaku_16} reproducing all steps of follicle development could be appropriate to settle elaborate cell lineage tracing informing on cell division events. Yet this is a long term perspective, since such devices are rather at the proof-of-concept level for the time being. In addition, they are up to now restricted to the mouse, since no feasible culture system for primordial and primary follicles is yet available for other species (see the overview picture including sheep in \cite{morohaku_19}).}

\subsection{Biological interpretation}\label{sec:interpretation}

    {From the primordial follicle data, we have found that the mean initial number of precursor cells for the Wild-Type subset is about the same as for the Mutant. Moreover, the prediction on the total number of proliferative cells at the end of the activation phase, $\mathbb{E}\left[C_{\tau}\right]$, is also very similar in the Wild-Type and Mutant cases. The observed shift in opposite directions for the mean initial cell number inferred from the MLE of the dynamical models (see bottom panel of Figure \ref{fig:primary_primordial}) 
    	 is thus compensated for by the differences in cell dynamics. The number of divisions during the transition is smaller in the Wild-Type than in the Mutant subset ($\mathbb{E}[C_\tau-F_0]\approx 2$ in Wild-Type, $\mathbb{E}[C_\tau-F_0]\approx 4$ in Mutant), as a result of a global difference between the MLE parameters: the order of magnitude of the division rates are closer to that of the transition rates in the Mutant compared to the Wild-Type subset. In overall, we conclude from our extensive datafitting analysis that the Wild-Type subset exhibits a clearer separation of dynamics during follicle activation (first cell transition, then cell proliferation), while in the Mutant cell proliferation could occur at a substantial rate before precursor cell extinction. We note that this conclusion has to be tempered by the sparse character of our experimental dataset. In particular, a detailed examination of the experimental data reveals that the four data points available for transitory follicles in the Wild-Type subset correspond to a clearly higher number of precursor cells than any of the primordial follicles, which certainly impacts our results. In contrast, the Mutant subset contains transitory follicles with significantly fewer precursor cells than the primary follicles.}
    
    {Even if there is a clear trend in the data to substantiate the existence of an auto-amplification of the transition from flattened to cuboidal cells, complementary data would be useful to decide the question. Indeed, getting more datapoints with a proportion of cuboidal cells in the range of 50 to 100 \% would constrain a step further the {follicle activation} trajectories, hence the parameter values and differences between nested models. The very fact that fewer follicles are counted in this range in the Wild-Type subset pleads for a possible acceleration of the transition.  Statistically, including more animals and more gestation times in the study would increase the number of data, including data missing in our current dataset, yet it would require enrolling many experimental animals. Ideally, monitoring the cell dynamics of ovarian follicles {\it in vivo}, in a non invasive manner, would provide all needed data. Yet, it is far from being a reachable target at the moment. Even the morphological monitoring of follicles (individual changes in follicle diameters) can only be performed for much later developmental stages due to size resolution (no reliable data can be obtained below 2mm diameter). 
    {An alternative} would be to record the location of the cuboidal cells with respect to the flattened ones, in consistency with the spatial interpretation of auto-amplification. The auto-amplification rate is motivated by two possible (and non exclusive) underlying mechanisms. First, the very first cell transitions could awake the oocyte and settle a positive feedback loop between the somatic cells and the oocyte \cite{adhikari_09,monniaux_18b} that would in turn secrete stimulatory factors reaching the surrounding somatic cells by diffusion (global amplification). Second, communications between adjacent somatic cells could help propagate activation step by step, from one (or a few) originally activated cell (local amplification). Local amplification might be detected in the data by recording the location of cuboidal cells and checking whether cluster of spatially related cuboidal cells can be detected. Global detection is expected to have a more homogenous effect, hence to be hardly detectable from static histological data.} 

    Finally, we highlight that the $\beta$-free linear submodel $(\mathcal{R}_1, \mathcal{R}_3, \mathcal{R}_4) $ performs as well as, and even better than the complete  model \eqref{Model_FC} $(\mathcal{R}_i)_{i \in \llbracket 1, 4 \rrbracket}$ in Mutant compared to Wild-Type ewes, which is compatible with the functional hypotheses applicable to the BMP15R mutation {\cite{reader_12}}. Indeed, one could speculate that the diminished BMP15 signaling would hamper the molecular dialog between the oocyte and somatic cells after follicle activation triggering, so that the auto-amplified cell event would barely occur in the Mutant group.
	
\section{Conclusion}

    In this work, we have introduced a stochastic nonlinear cell population model to study the sequence of events occurring just after the initiation of follicle growth. We have characterized the dynamics of precursor and proliferative cell populations according to the parameter values, for both the stochastic model and its deterministic mean-field counterpart. We have studied in details the extinction time of the precursor cell population, and designed an algorithm to compute numerically both the mean extinction time and mean number of proliferative cells at the extinction time. The algorithm is based on a domain truncation similar to the Finite State Projection (FSP) method proposed in \cite{munsky_finite_2006,kuntz_deterministic_2017}. The FSP approach aims to approximate the law of the process at a given time by solving a truncated version of the Kolmogorov forward system. We have adapted the FSP algorithm to {close} the infinite recurrence relation satisfied by the extinction time moments.
    We have found a consistent spatial boundary to solve the closure problem, thanks to a coupling technique and tractable upper-bound process. The numerical cost of the algorithm is deeply related to the proper choice of the upper-bound processes {and gets worse than direct simulation as $2\gamma$ gets close to the required bound $\alpha_1+\beta$ of Algorithm~\ref{algo-gr}}. 
    
    This algorithm has nevertheless allowed us to investigate the parameter influence on the precursor cell extinction time and number of proliferative cells at the end of the follicle activation phase. The auto-amplified transition rate  $\beta$ exerts a critical control on the mean extinction time, with a sharp timescale reduction when $\beta$ exceeds the spontaneous cell transition $\alpha_1$, while the division rates ($\alpha_2$, $\gamma$) have relatively less effect. The effect of the auto-amplification process is probably dependent on the specific parameterization of the cell event rates chosen in this work, yet our findings bring interesting insight into the  mechanisms underlying follicle activation; nonlinear feedbacks mediated through cell-to-cell communication certainly play a role, and our estimation results have shown that any impairment of this feedback would change drastically the kinetics of follicle activation.
    
    Moreover, our results can be useful to understand the variability in the cell numbers among ovarian follicles at the end of the activation phase, which can be used as initial conditions for models describing the following stages of follicle development \cite{clement_coupled_2013,CRY2019}. 
     {Going even further, the sequence of events occurring just after the initiation of follicle growth is determinant for the remaining of the entire follicle development process. The whole cell population in mature (ovulatory) follicles (up to tens of millions in large mammal species as humans) emanates from the few cells a primordial follicle is endowed with. The timings of the cells' first divisions will determine the distribution of cytological cell ages in the population, which will ultimately influence the distribution of the times of cell cycle exit in fully differentiated cells. Collectively, the exit time distribution controls the switch from  proliferation to differentiation, a key event in the selection of ovulatory follicles \cite{clement_13b}. Also, the proliferative vitality of the cuboidal (transitioned) cells will control the clonal composition of the follicles and participate in the spatial and functional heterogeneity within follicle cell populations, persisting very late in development.}
    
    We have performed the parameter calibration in a special context of time-free data. It turns out that the proliferative cell {number} can be seen as a clock for the whole process, and that the embedded Markov chain is better adapted to time-free data than the continuous-time model. We have used the embedded Markov chain to define a proper likelihood function and a statistically rigorous framework. The likelihood function has allowed us to perform an extensive data fitting analysis, using the very useful concept of profile likelihood estimate. This analysis sheds light onto several aspects of the activation of ovarian follicles. First, the transition scenario, where cell proliferation is mostly posterior to cell transition, and the cell number increase is moderate, seems to be predominant versus a more proliferative scenario. While the question is still open, it seems likely that cell transition is favored in the Wild-Type strain compared to the Booroola mutant strain. With the available experimental dataset, we have yet not managed to make a clear distinction between, on one side, a progressive transition with a steady net flux from flattened to cuboidal cells, and, on the other side, an auto-catalytic transition with an ever increasing flux all along the activation phase.
    
    Beyond our application in female reproductive biology, we believe that the modeling approach presented here can have a more generic interest in cell kinetics related issues, especially when a small number of cells is involved. Also, from the mathematical biology viewpoint, the analysis performed on the extinction time, combining theoretical (coupling) and numerical (finite state projection) tools may have an interest for first passage time studies in stochastic processes.

\section{Appendix}

\subsection{Justification of the choice of the rate of $\mathcal{R}_2$}\label{sec:appendix-R2}
{As detailed in Section \ref{sec:interpretation} the auto-amplification can result from two non-exclusive mechanisms, a nonlocal (global) one and a local one.}\\
{{Global amplification}: consider that each proliferative cell sends a fixed amount of {growth }signals to the oocyte. The oocyte thus receives a signal proportional to the number of proliferative cells $C$. We consider that the oocyte secrete in turn (instantaneously) a stimulatory signal, {at a level} proportional to the amount of {growth} signals received {from somatic cells}. By homogeneous diffusion, {the oocyte signal} is shared equally to all somatics cells, so that each precursor cell receive a signal proportional to $C/(F+C)$.\\
{Local amplification}: for a given precursor cell, assuming a random repartition of the cell {types} around the oocyte {(hence neglecting local cell-to-cell effects)}, the probability that a neighbor cell is a proliferative cells is $C/(F+C-1)$, which is also consistent with our choice. 
}

\subsection{Mean-field formulation}
		{To get some insight into the model behavior, we describe the mean-field version of model \ref{Model_FC}, given by the following set of ODE:
		\begin{equation}\label{ODE_MC}
			\left\{ 
			\begin{array}{c}
				\frac{d}{dt}f(t) = - \alpha_1 f(t) - \beta f(t) \frac{c(t)}{f(t) + c(t)}, \\
				\frac{d}{dt}c(t) = (\alpha_1 + \alpha_2)f(t) + \beta f(t) \frac{c(t)}{f(t) + c(t)} +\gamma c(t),
			\end{array}
		\right.
		\end{equation}
		with the initial condition $(f(0), c(0)) = (f_0, 0) $, with $f_0 \in \mathbb{R}_+ $.
		We start by solving analytically the deterministic formulation, and then investigate the effect of each parameter on the model outputs.
		}

		{From the ODE sytem \eqref{ODE_MC}, we deduce the change in the proliferative cell proportion $p_C(t) := \frac{c(t)}{f(t) + c(t)}$:
			\begin{multline}\label{ODE_pC}
				\frac{d}{dt}p_C(t) =  \alpha_1 + \alpha_2 - (\alpha_1 + 2 \alpha_2 - \beta - \gamma)p_C(t) + (\alpha_2 - \beta - \gamma)p_C(t)^2 \\
				= (\alpha_2 - \beta - \gamma) (p_C(t) - 1) (p_C(t) - \frac{ \alpha_1 + \alpha_2}{\alpha_2 - \beta - \gamma}).
			\end{multline}
			From ODEs \eqref{ODE_MC} and \eqref{ODE_pC}, using the classical method of separation of variables, we can compute the analytical expressions for the proliferative cell proportion $p_C(t)$, proliferative cell number $c(t)$ and precursor cell number $f(t)$:  
			\begin{proposition}\label{prop-ODE_system_sol}
				The solution of the ODE system \eqref{ODE_MC} is, for all $t\geq 0$, 
				\begin{multline*}
				f(t) = f_0\exp \left(- \alpha_1 t - \beta \int_{0}^{t}p_C(s)ds \right),  \\ 
				c(t) = f_0 \left( \exp \left(\alpha_2 t +  (\gamma - \alpha_2)  \int_{0}^{t} p_C(s)ds\right) -  \exp \left(- \alpha_1 t - \beta \int_{0}^{t}p_C(s)ds \right)\right).
				\end{multline*}
				In addition, the solution of ODE \eqref{ODE_pC} is 
			\begin{equation}\label{eq-pc-time}
					p_C(t) =  \frac{1 - \exp\left( -(\alpha_1 + \beta + \gamma)t\right) }{1 -  \frac{\alpha_2 - \beta - \gamma}{\alpha_1 + \alpha_2}\exp\left( -(\alpha_1 + \beta + \gamma)t\right)} .
				\end{equation}
				and 
				the total cell number verifies
				\begin{equation*}
				n(t) := f(t) + c(t)= f_0 \exp \left(\alpha_2 t +  (\gamma - \alpha_2)  \int_{0}^{t} p_C(s)ds\right).
				\end{equation*}
			\end{proposition}
	}
		{From Proposition \ref{prop-ODE_system_sol}, it is clear that the proliferative cell proportion $p_C$ converges to $1$. If $\gamma>0$, the proliferative cell number $c$ grows asymptotically exponentially at a rate $\gamma$ when $t\to\infty$. If $\gamma=0$, $c(t)$ is bounded because $t\mapsto 1-p_C(t)$ is converging exponentially fast to $0$, hence is integrable on $(0,\infty)$.}
		{Moreover, the proliferative cell proportion $p_C$ has an inflexion point if and only if \[\beta+\gamma>\alpha_1+2\alpha_2\,.\] 
		An inflexion point denotes the presence of at least two distinct phases, with a first progressive acceleration phase followed by a saturating phase.}			
			
			Finally, note that according to the observed variables, the submodels cannot be distinguished from one another, or, alternatively, different parameter values (within a same submodel) may lead to identical outputs. Indeed, the changes in the precursor cell population are independent of parameters $\alpha_2,\gamma$, and, more strikingly,  parameters $ \beta$ and $\gamma$ cannot be separated in the analytical solution \eqref{eq-pc-time}, leading to the same kinetic patterns for $p_C$ as long as the combination $\gamma + \beta$ remains unchanged.

	\subsection{Analytical expressions in the linear case}\label{ssec:proof_linear}
	
	\begin{proof}[Proof of Proposition \ref{prop-loi-fLT}]
							Let $t \geq 0$ and $f \in \llbracket 0, f_0 \rrbracket $. Since $F_t$ is autonomous and is a pure death process, we can directly write the following forward Kolmogorov equation: for all  $f \in \llbracket 0, f_0 \rrbracket$,
						\begin{multline}\label{proof-FLt-law}
						\frac{d}{dt}\mathbb{P}\left[ F^L_t = f |F_0 = f_0\right] = \\
						 \alpha_1(f + 1) \mathbb{P}\left[ F^L_t = f+1 |F_0 = f_0\right]  -  \alpha_1 f  \mathbb{P}\left[ F^L_t = f |F_0 = f_0\right] .
						\end{multline}
						Solving by recurrence \eqref{proof-FLt-law}, we deduce that, for all  $f \in \llbracket 0, f_0 \rrbracket$,
						\begin{equation*}
							\mathbb{P}\left[ F^L_t = f | F_0 = f_0\right]  = \dbinom{f_0}{f} (e^{ - \alpha_1 t})^f(1 - e^{ - \alpha_1  t})^{f_0 -f}.
						\end{equation*}
						Note that $\mathbb{P}\left[ F_t^L = 0 | F_0 = f_0 \right] = (1 - e^{-\alpha_1 t})^{f_0}$ which converges to $1$ when $t\mapsto 1-p_C(t)$. Hence, process $F^L$ extincts almost surely (a.s.) when $t$ goes to infinity, hence $\taulinear < \infty $. Before computing the law of $\taulinear$, we can directly obtain its mean using the recursive expression \eqref{eq-jump-time-FL}:
						\begin{equation*}
							\mathbb{E}\left[ \taulinear \right] = \sum_{k = 0}^{f_0 - 1} 	\mathbb{E}\left[ T_{k + 1} - T_k \right]  = \sum_{k = 0}^{f_0 - 1} 	\mathbb{E}\left[ \mathcal{E}\left(\alpha_1 (f_0 - k)\right) \right] = \frac{1}{\alpha_1}\sum_{k = 1}^{f_0}\frac{1}{k}.
						\end{equation*}
						Using again  Eq.~\eqref{eq-jump-time-FL}, we deduce that $\taulinear (= T_{f_0})$ follows a generalized Erlang law whose density function is:
						\begin{equation}\label{taup-e1}
						f_{\taulinear}(t) = \mathds{1}_{t \geq 0} \sum_{i = 0}^{f_0 - 1} \prod_{j \neq i, j = 0}^{f_0 - 1}\frac{f_0 -j}{i - j} \alpha_1 (f_0 - i)e^{-\alpha_1(f_0 - i)t}.
						\end{equation}
						Due to the specific form of the exponential rate, we can simplify Eq.~\eqref{taup-e1} further.
						As $ \displaystyle \prod_{j \neq i, j = 0}^{f_0 - 1}(f_0 -j)  = \frac{f_0 !}{f_0 - i}$ and
						\begin{align*}
						\displaystyle \prod_{j \neq i, j = 0}^{f_0 - 1} (i - j)  = & \displaystyle \prod_{j= 0}^{i-1} (i-j) \times \prod_{j= i+1}^{f_0 - 1} (i - j) \\
						& = i! (-1)^{f_0 - 1 - i} \prod_{j= 1}^{f_0 - 1 - i} j = (-1)^{f_0 -1 - i}i !(f_0 - 1 - i) !,
						\end{align*}
						we deduce
						\begin{align*}
						f_{\taulinear}(t) = & \alpha_1 \mathds{1}_{t \geq 0} \sum_{i = 0}^{f_0 - 1} \frac{f_0 !}{i !(f_0 - 1 - i) !} (-1)^{f_0 - 1 - i} e^{-\alpha_1(f_0 - i)t} \\
						=  & \alpha_1 f_0 e^{-\alpha_1t}\mathds{1}_{t \geq 0} \sum_{i = 0}^{f_0 - 1}  \dbinom{f_0 - 1}{i} (-e^{-\alpha_1t})^{f_0 - i - 1}  \\
						= & \alpha_1 f_0 e^{-\alpha_1t} (1 - e^{-\alpha_1t})^{f_0 - 1}\mathds{1}_{t \geq 0}.
						\end{align*}
			\myspecialendproof \end{proof}
			
	\begin{proof}[Proof of Proposition \ref{prop-Ce-linear}]
		According to Proposition \ref{prop-loi-fLT}, $\taulinear $ is a.s. finite. To take the expectation of $C^L_t $ at time $t =  \taulinear $, we check that  $\mathbb{E}\left[C^{k,j}_{\taulinear - T_k^j}  \right]< \infty $, for all $k$ and $j$. For all $t \geq 0$, $C_t^{k,j}$ is $L_1-$integrable (as a Yule process) with $\mathbb{E}\left[C^{k,j}_{t}  \right] = e^{\gamma t} $. Conditionning on the law of $\taulinear$, we get {(with the change of variables $x=1-e^{-\alpha_1 t}$)
				\begin{multline*}
					I:=\mathbb{E}\left[C^{k,j}_{\taulinear}  \right]= \int_{0}^{+ \infty} e^{\gamma t} f_{\taulinear}(t)dt   = f_0 \int_{0}^{+ \infty}e^{\gamma t}(1 - e^{-\alpha_1 t})^{f_0 - 1}\alpha_1e^{- \alpha_1t} dt \\
				= f_0 \int_0^1 (1-x)^{-\frac{\gamma}{\alpha_1}}x^{f_0-1}dx=f_0B\left(f_0,1-\frac{\gamma}{\alpha_1}\right)
				\end{multline*}
				where $B$ is the standard Beta function. {Hence} $I < \infty$ {if and only if} Hypothesis \ref{hyp-rate2} holds. Note that using the properties of the Beta function, we have
				\begin{equation}\label{eq:I}
			 I=\frac{f_0!}{\left(f_0-\frac{\gamma}{\alpha_1}\right)!}\,,
				\end{equation}
				where we use the notation $\left(m-x\right)!=\prod_{k=1}^m (k-x)$.
				}
				 {Thus, if Hypothesis \ref{hyp-rate2} holds true}, and given that $C^{k,j}$ is a positive increasing process, we deduce:
				\begin{equation*}
			\mathbb{E}\left[C^{k,j}_{\taulinear - T_k^j}  \right] \leq \mathbb{E}\left[C^{k,j}_{\taulinear}  \right]  < \infty.
				\end{equation*}
				Then, taking the expectation of \eqref{eq-decomp-Yule} at time $t =  \taulinear $, we obtain:
			\begin{equation}\label{eq-ctau-lin}
			\mathbb{E}\left[C^{L}_{\taulinear} \right] = \sum_{k = 1}^{f_0}  \mathbb{E} \left[  C^{k,0}_{\taulinear - T_k^0} \right] + \sum_{k = 0}^{f_0 - 1} \mathbb{E} \left[\sum_{j = 1}^{N_k(\taulinear)} C^{k,j}_{\taulinear - T_k^j} \right].
			\end{equation} 
			{Moreover, we have that each counting process $N_k(t)$ can be dominated by
			\begin{equation*}
			N_k(t) \leq \mathcal{Y}_3 \left( \alpha_2 f_0 t\right)\,,
			\end{equation*}
			so that
			\begin{equation*}
			\sum_{j = 1}^{N_k(\taulinear)} C^{k,j}_{\taulinear - T_k^j} \leq \sum_{j = 1}^{\mathcal{Y}_3(\taulinear)} C^{k,j}_{\taulinear}\,.
			\end{equation*}
			Finally, conditionally on $\taulinear$, $\mathcal{Y}_3(\taulinear)$ is independent of each $C^{k,j}_{\taulinear}$, and the latter are independent and identically distributed random variables. Using that
			\begin{equation*}
			\mathbb{E}\left[ \sum_{j = 1}^{\mathcal{Y}_3(\taulinear)} C^{k,j}_{\taulinear} \right] = \mathbb{E}\left[\mathbb{E}\left[ \sum_{j = 1}^{\mathcal{Y}_3(\taulinear)} C^{k,j}_{\taulinear} \mid \taulinear \right]\right]\,,
			\end{equation*}
			and the Wald equation \cite[Chap. XII]{Feller} , we obtain
			\begin{equation*}
			     \mathbb{E} \left[\sum_{j = 1}^{N_k(\taulinear)} C^{k,j}_{\taulinear - T_k^j} \right] \leq  \alpha_2 f_0 \int_{0}^{+ \infty} t e^{\gamma t}   f_{\taulinear}(t)dt\,, 
			\end{equation*}
			which is finite under Hypothesis \ref{hyp-rate2}. Finally, if Hypothesis \ref{hyp-rate2} does not hold, we have, as long as $f_0\geq 2$:
			\begin{equation*}
			 \mathbb{E} \left[  C^{1,0}_{\taulinear - T_1^0} \right]  \geq \mathbb{E} \left[  C^{1,0}_{T_2^0 - T_1^0} \right]=\infty\,.  
			\end{equation*}
			}
			In some {special} cases, Formula {\eqref{eq-ctau-lin}} can be used to obtain the first moment of $C^L_{\taulinear}$.
			
				When $ \gamma $ is zero, then for all $t \geq 0$, for all $k \in \llbracket 1, f_0  \rrbracket $ and for all $j \in \llbracket  1,N_k(\taulinear) \rrbracket $, $ C^{k,j}_t = 1$. We deduce directly from Eq.~\eqref{eq-ctau-lin} that
				\begin{equation}\label{eq-ctaul-proof}					\mathbb{E}\left[C^{L}_{\taulinear} \right] =  f_0 + \sum_{k = 0}^{f_0 - 1} \mathbb{E} \left[ N_k(\taulinear) \right].
				\end{equation}
				From Eq.~\eqref{eq-N_k-formula}, we have 
				\begin{multline*}
					\mathbb{E}\left[ N_k(\taulinear) \right] = \mathbb{E}\left[ \mathcal{Y}_3 \left( \alpha_2 \int_{0}^{T_{k+1}}  F^L_s ds \right) - \mathcal{Y}_3 \left( \alpha_2 \int_{0}^{T_{k}}  F^L_s ds \right) \right] \\
					= \mathbb{E}\left[ \mathcal{Y}_3 \left( \alpha_2 \int_{T_k}^{T_{k+1}}  F^L_s ds \right) \right] = \mathbb{E}\left[ \alpha_2 \int_{T_k}^{T_{k+1}}  F^L_s ds  \right],
				\end{multline*}
				by Poisson process property. Since for all $ t \in [T_k, T_{k+1})$, $F^L_t = f_0 - k$, we deduce that $\mathbb{E}\left[ N_k(\taulinear) \right] = \mathbb{E}\left[ \alpha_2 (f_0 - k)(T_{k+1} - T_k) \right]$.
				Using \eqref{eq-jump-time-FL}, we deduce that 
				$\mathbb{E} \left[ N_k(\taulinear) \right]  = \frac{\alpha_2(f_0 - k) }{\alpha_1(f_0 - k) } =  \frac{\alpha_2}{\alpha_1 } $ and conclude with \eqref{eq-ctaul-proof}.\\
				
				When $ \alpha_2 $ is zero, $ N_k (t)$ is null for all $t \geq 0$, and we deduce directly from \eqref{eq-ctau-lin} that
				\begin{equation}\label{eq-proof-Ctau-e2}
					\mathbb{E}\left[C^{L}_{\taulinear} \right] = \sum_{k = 1}^{f_0}  \mathbb{E} \left[  C^{k,0}_{\taulinear - T_k} \right].
				\end{equation}
				Since $T_{f_0} = \taulinear$, we have $  C^{f_0,0}_{\taulinear - T_{f_0}}  = 1$. 
				Let  $k \in \llbracket 1, f_0 - 1 \rrbracket $. Since $\taulinear - T_k \overset{(law)}{=} \sum_{i = k + 1}^{f_0} \mathcal{E}\left( \alpha_1 (f_0 - i + 1)\right) \overset{(law)}{=}  \sum_{i = 1}^{f_0 - k } \mathcal{E}\left( \alpha_1 i \right)  $, using Proposition \ref{prop-loi-fLT}, we deduce that the density function of $\taulinear - T_k $ is
				\begin{equation*}
					f_{\taulinear - T_k}(t) = \alpha_1 (f_0 - k) e^{- \alpha_1 t} (1 - e^{-\alpha_1 t})^{f_0 - k - 1}\mathds{1}_{t \geq 0}.
				\end{equation*}
				Then, conditioning $C^{k,0}_{\taulinear - T_k} $ on the law of $\taulinear - T_k $, we first deduce that
				\begin{equation*}
				\displaystyle \mathbb{E} \left[  C^{k,0}_{\taulinear - T_k} \right] = \int_{0}^{+ \infty} \mathbb{E} \left[  C^{k,0}_{t} \right]	f_{\taulinear - T_k}(t) dt,
				\end{equation*}
				Then, since $\mathbb{E} \left[  C^{k,0}_{t} \right] = e^{\gamma t} $, we have, {similarly as in Eq.~\eqref{eq:I},
				\begin{equation*}
				    \mathbb{E} \left[  C^{k,0}_{\taulinear - T_k} \right] = \frac{(f_0-k)!}{\left((f_0-k)-\frac{\gamma}{\alpha_1}\right)!}\,,
				\end{equation*}}
				which ends the proof using \eqref{eq-proof-Ctau-e2}.
			\myspecialendproof
			
	\end{proof}
{	The following proposition is analogous to Proposition \ref{prop-Ce-linear}, yet with the decoupled processes $\tilde F$ and $\tilde C$, whose moments are easier to estimate. Note that parameters  $\tilde \alpha,\tilde \beta, \tilde \gamma$ below are generic ones.}
				\begin{proposition}\label{prop-esp-tilde}
				{Let $\tilde F,\tilde C$ be independent pure-jump stochastic processes on $\mathbb N$, of infinitesimal generators
				\begin{align*}
				\overset{\sim}{\mathcal{L}}_F \phi(f) = \tilde\alpha f \left[ \phi(f-1) -   \phi(f) \right], \\
				\overset{\sim}{\mathcal{L}}_C \phi(c) = \left[ \tilde\beta  + \tilde\gamma c \right] \left[ \phi(c+1) -   \phi(c) \right].
				\end{align*}
				with deterministic initial condition $\tilde F(0)=f_0$ and $\tilde C(0)=n\geq 1$, and where $\tilde \alpha,\tilde \beta, \tilde \gamma$ are non-negative rate parameters. Let
				\begin{equation*}
				    \tilde \tau = \inf \{  t>0; \quad \tilde F_{t} = 0 | f_0 \}
				\end{equation*}
				}
				{For any $p\geq 1$, 
					\begin{equation}\label{prop-Csup_positgamma_bounds}
					\mathbb{E}\left[(\tilde C_{\tilde \tau})^p \right]<\infty\,.
					\end{equation}
					if, and only if,
				\begin{equation}\label{eq:hyp_2bis}
				    p \tilde\gamma < \tilde\alpha\,,
				\end{equation}
}
				Moreover, we have:
				\begin{itemize}
					\item if $\tilde\gamma >0$:
			{for $p=1$,		\begin{equation*}
					\mathbb{E}\left[\tilde C_{\tilde \tau} \right] = n\frac{f_0!}{\left(f_0-\frac{\tilde\gamma}{\tilde\alpha}\right)!}+\frac{\tilde\beta}{\tilde\gamma} \left(\frac{f_0!}{\left(f_0-\frac{\tilde\gamma}{\tilde\alpha}\right)!}-1\right)\,, 
					\end{equation*}
					}
			{and for $p=2$,		
		\begin{multline*}
		\mathbb{E}\left[(\tilde C_{\tilde \tau})^2 \right] = \left(n+\frac{\tilde\beta}{\tilde\gamma}\right)\left(n+\frac{\tilde\beta}{\tilde\gamma}+1\right)\frac{f_0!}{\left(f_0-\frac{2\tilde\gamma}{\tilde\alpha}\right)!}\\-\left(n+\frac{\tilde\beta}{\tilde\gamma}\right)\left(1+2\frac{\tilde\beta}{\tilde\gamma}\right)\frac{f_0!}{\left(f_0-\frac{\tilde\gamma}{\tilde\alpha}\right)!}+\left(\frac{\tilde\beta}{\tilde\gamma}\right)^2
				\end{multline*} 
					}
					\item if $\tilde\gamma = 0$:
				{	\begin{equation*}
					\mathbb{E}\left[\tilde C_{\tilde \tau} \right]  = n+\frac{\tilde\beta}{\tilde \alpha} \sum_{i = 1}^{f_0 } \frac{ 1}{i}\,,
					\end{equation*}
				}
				{	\begin{equation*}
					\mathbb{E}\left[(\tilde C_{\tilde \tau})^2 \right]  = n+\frac{\tilde\beta}{\tilde \alpha} \sum_{i = 1}^{f_0 } \frac{ 1}{i}+\frac{\tilde\beta^2}{\tilde \alpha^2} \left(\sum_{i = 1}^{f_0 } \frac{ 1}{i^2}+\left(\sum_{i = 1}^{f_0 } \frac{ 1}{i}\right)^2\right)\end{equation*}
				}
				\end{itemize}
			\end{proposition}
			\begin{proof}
					Since $\tilde \tau $ and $\tilde C$ are independent, we deduce by conditioning on $\tilde \tau $ that 
				{\begin{equation}\label{proof-prop-esp-tilde-e1}
				\mathbb{E}\left[(\tilde C_{\tilde \tau})^p \right]  = \displaystyle \int_{0}^{+ \infty}\mathbb{E}\left[(\tilde C_{t})^p \right]f_{\tilde \tau}(t)dt, 
				\end{equation}}
				where $f_{\tilde \tau}$ is the density probability of $\tilde \tau $. Since $\tilde F$ is linear, we apply Proposition~\ref{prop-loi-fLT} and obtain 
			\begin{equation}\label{proof-prop-esp-tilde-e3}
				f_{\tilde \tau}(t) =\tilde \alpha f_0 e^{- \tilde\alpha t} (1 - e^{-\tilde\alpha t})^{f_0 - 1} \mathds{1}_{[0, + \infty)}(t).
				\end{equation}
				\\
				Now, we suppose that $\tilde\gamma >0 $. Then, $\tilde C $ {can be decomposed as the independent sum of $n$ Yule processes starting from $1$ (see Eq.~\eqref{eq:Yule}) and a birth process with immigration (starting from $0$)}. It is classical that {the Yule process follows a geometric law of parameter $e^{-\tilde \gamma t}$, and the birth process with immigration} follows a negative binomial law $\mathcal{BN}\left( \frac{\tilde\beta}{\tilde\gamma}, e^{- \tilde\gamma t} \right) $,
				{there exists $k,K>0$ (depending on model parameters, but independent of $t$) such that, for all $t \geq 0$,
				\begin{equation}\label{eq:control_moment_p}
				ke^{p\tilde\gamma t}\leq \mathbb{E}\left[(\tilde C_{t})^p \right] \leq Ke^{p\tilde\gamma t}\,.
				\end{equation}
				}
			{Combining Eq.~\eqref{eq:control_moment_p} with Eqs.~\eqref{proof-prop-esp-tilde-e1} and \eqref{proof-prop-esp-tilde-e3} yields \eqref{prop-Csup_positgamma_bounds}. To obtain the remaining analytical formulas, we note that	\begin{equation}\label{proof-prop-esp-tilde-e2}
				\mathbb{E}\left[\tilde C_{t} \right] = ne^{\tilde\gamma t}+\frac{\tilde\beta}{\tilde\gamma}(e^{\tilde\gamma t} - 1)=e^{\tilde\gamma t}\left(n+\frac{\tilde\beta}{\tilde\gamma}\right)-\frac{\tilde\beta}{\tilde\gamma}\,,
				\end{equation}
			and
			\begin{equation}\label{proof-prop-esp-tilde-e2_2}
				\mathbb{E}\left[(\tilde C_{t})^2 \right] = e^{2\tilde\gamma t}\left(n+\frac{\tilde\beta}{\tilde\gamma}\right)\left(n+\frac{\tilde\beta}{\tilde\gamma}+1\right)-e^{\tilde\gamma t}\left(n+\frac{\tilde\beta}{\tilde\gamma}\right)\left(1+2\frac{\tilde\beta}{\tilde\gamma}\right)+\left(\frac{\tilde\beta}{\tilde\gamma}\right)^2\,.
				\end{equation} 
				Also, for any $p$ such that \eqref{eq:hyp_2bis} holds true, we have (with the change of variables $x=1-e^{-\tilde\alpha t}$)
				\begin{equation*}
				    \int_0^\infty e^{p\tilde\gamma t}f_{\tilde \tau}(t)dt=f_0\int_0^1 (1-x)^{-\frac{p\tilde\gamma}{\tilde\alpha}}x^{f_0-1}dx=f_0B\left(f_0,1-\frac{p\tilde\gamma}{\tilde\alpha}\right)\,,
				\end{equation*}
				where $B$ is the standard Beta function. We deduce that
				\begin{equation*}
				    \int_0^\infty e^{p\tilde\gamma t}f_{\tilde \tau}(t)dt=\frac{f_0!}{\left(f_0-\frac{p\tilde\gamma}{\tilde\alpha}\right)!}\,,
				\end{equation*}
				where we use the notation $\left(m-x\right)!=\prod_{k=1}^m (k-x)$.
				Then, using Eqs.~\eqref{proof-prop-esp-tilde-e2}-\eqref{proof-prop-esp-tilde-e2_2} and \eqref{proof-prop-esp-tilde-e3}, we deduce from \eqref{proof-prop-esp-tilde-e1} that
				\begin{equation*}
		\mathbb{E}\left[\tilde C_{\tilde \tau} \right] = n\frac{f_0!}{\left(f_0-\frac{\tilde\gamma}{\tilde\alpha}\right)!}+\frac{\tilde\beta}{\tilde\gamma} \left(\frac{f_0!}{\left(f_0-\frac{\tilde\gamma}{\tilde\alpha}\right)!}-1\right)
				\end{equation*}
				and
				\begin{multline*}
		\mathbb{E}\left[(\tilde C_{\tilde \tau})^2 \right] = \left(n+\frac{\tilde\beta}{\tilde\gamma}\right)\left(n+\frac{\tilde\beta}{\tilde\gamma}+1\right)\frac{f_0!}{\left(f_0-\frac{2\tilde\gamma}{\tilde\alpha}\right)!}\\-\left(n+\frac{\tilde\beta}{\tilde\gamma}\right)\left(1+2\frac{\tilde\beta}{\tilde\gamma}\right)\frac{f_0!}{\left(f_0-\frac{\tilde\gamma}{\tilde\alpha}\right)!}+\left(\frac{\tilde\beta}{\tilde\gamma}\right)^2
				\end{multline*}
				}
			If $\tilde\gamma = 0$, then $\tilde C$ is a pure immigration process {starting from $n$}, and follows a {shifted} Poisson law $n+\mathcal{P}\left(\tilde\beta t \right) $ at time $t \geq 0$. Using the same approach, we obtain that  
				{\begin{equation*}
					\displaystyle \mathbb{E}\left[\tilde C_{\tilde \tau} \right]  =  \int_{0}^{+ \infty } (n+\tilde\beta t) f_{\tilde \tau}(t) dt =  n+\tilde\beta \mathbb{E}\left[ \tilde \tau\right]=
					n+\frac{\tilde\beta}{\tilde \alpha} \sum_{i = 1}^{f_0 } \frac{ 1}{i}\,,
				\end{equation*}
				and
				\begin{multline*}
					\displaystyle \mathbb{E}\left[(\tilde C_{\tilde \tau} )^2\right]  =  \int_{0}^{+ \infty } (n+\tilde\beta t(\tilde\beta t+1)) f_{\tilde \tau}(t) dt =  n+\tilde\beta \mathbb{E}\left[ \tilde \tau\right]+\tilde\beta^2 \mathbb{E}\left[ (\tilde \tau)^2\right]\\
					=n+\frac{\tilde\beta}{\tilde \alpha} \sum_{i = 1}^{f_0 } \frac{ 1}{i}+\frac{\tilde\beta^2}{ \tilde\alpha^2} \left(\sum_{i = 1}^{f_0 } \frac{ 1}{i^2}+\left(\sum_{i = 1}^{f_0 } \frac{ 1}{i}\right)^2\right)\,,
				\end{multline*}
				}.
			\myspecialendproof \end{proof}
			
	\subsection{Numerical scheme for $\mathbb{E}[\tau]$ and $\mathbb{E}[C_{\tau}] $}\label{appendix:numericalscheme}
	
	\paragraph{Pseudo-code}
			{
			We design algorithm~\ref{algo-gr} to compute a numerical estimate of $g(f_0,0)$, solution of Eq.~\eqref{g_function} that represents either $\mathbb{E}\left[\tau\right]$ or $  \mathbb{E}\left[C_{\tau}\right]$ according to the specific choice of boundary condition. This algorithm requires $ \gamma<\alpha_1+\beta$ to compute $\mathbb{E}\left[\tau\right]$, and $ 2\gamma<\alpha_1+\beta$ to compute $  \mathbb{E}\left[C_{\tau}\right]$, in agreement with Theorem \ref{thm_borne_cns}, Proposition \ref{prop-couplage} and Proposition \ref{prop-erreur}. The prefactor $A$ given below is obtained thanks to Proposition \ref{prop-esp-tilde}.
			}
			
			\begin{algorithm}[h]
				[1.] Fix $f_0$, $g_0$, $\alpha$, the parameter set $\theta = (\alpha_1,\alpha_2,\beta,\gamma)$ and the tolerance error $\epsilon$\;
				[2.]{ To compute $g(f_0,0)=	\mathbb{E}\left[\tau\right]$, choose  $n$ such that $\gamma<\alpha_1+\beta\frac{n}{1+n}$
				and fix 
				\begin{small}
				\begin{multline*}
				    A= \mathbb{E}\left[\tau_L^{f_0}\right]\mathbb{E}[C_{\tau_n+\overline{\tau}_n}^{n}]=\\
				  \frac{1}{\alpha_1} \left(\sum_{k = 1}^{f_0} \frac{1}{k}\right) \left( \left(n+\frac{(\alpha_1 + \beta + \alpha_2)f_0}{\gamma}\right)\frac{f_0!}{\left(f_0-\frac{\gamma}{\alpha_1+\beta\frac{n}{1+n}}\right)!}-\frac{(\alpha_1 + \beta + \alpha_2)f_0}{\gamma} \right)
				\end{multline*}
				\end{small}
				}
			[2bis.]{	To compute $g(f_0,0)=	\mathbb{E}\left[C_{\tau}\right]$, choose $n$ such that $2\gamma<\alpha_1+\beta\frac{n}{1+n}$ and
			 fix
		\begin{small}\begin{equation*}
				    A= \left(\left(1+ \frac{(\alpha_1 + \beta + \alpha_2)f_0 }{\gamma}\right) \left(\frac{f_0!}{\left(f_0-\frac{\gamma}{\alpha_1+\beta\frac{n}{1+n}}\right)!}-1\right)\right)\mathbb{E}\left[\left(C_{\tau_n+\overline{\tau}_n}^{n}\right)^2\right]
				\end{equation*}
				\end{small}
				where 
				\begin{small}
				\begin{multline*}
		\mathbb{E}\left[(C_{\tau_n+\overline{\tau}_n}^{n})^2 \right] =\left(\frac{(\alpha_1 + \beta + \alpha_2)f_0}{\gamma}\right)^2\\ +\left(n+\frac{(\alpha_1 + \beta + \alpha_2)f_0}{\gamma}\right)\left(n+\frac{(\alpha_1 + \beta + \alpha_2)f_0}{\gamma}+1\right)\frac{f_0!}{\left(f_0-\frac{2\gamma}{\alpha_1+\beta\frac{n}{1+n}}\right)!}\\-\left(n+\frac{(\alpha_1 + \beta + \alpha_2)f_0}{\gamma}\right)\left(1+2\frac{(\alpha_1 + \beta + \alpha_2)f_0}{\gamma}\right)\frac{f_0!}{\left(f_0-\frac{\gamma}{\alpha_1+\beta\frac{n}{1+n}}\right)!}
				\end{multline*} 
				\end{small}
				}
		[3.] Compute $r = \frac{A}{\epsilon}$\;
		[4.] Initialize $g_r(f,r) = {g_0(r)}$ for all $f \in \llbracket 0, f_0 \rrbracket$ \;
		[5.]		\For{c from $r - 1$ to $0$ }{
					$g_r(0,c) \leftarrow g_0(c)$ \;
					\For{$f$ from $1 $ to $f_0$}{
						$g_r(f,c) \leftarrow \frac{-\alpha + (\alpha_1 f + \beta \frac{fc}{f + c})g_r(f-1, c + 1) + (\gamma c + \alpha_2 f)  g_r(f, c + 1)}{\gamma c + (\alpha_1 f + \beta \frac{fc}{f + c}) + \alpha_2 f }$\;
				}}
			[6.]	Return $g_r(f_0,0) $\;
			
				\caption{Pseud-code for the numerical estimate of $\mathbb{E}\left[\tau\right]$ and $\mathbb{E}\left[C_{\tau}\right]$}\label{algo-gr}
			\end{algorithm}

\subsection{{\it In silico} dataset}\label{appendix-insilico}
        {We generate {\it in silico} datasets to further explore parameter identifiability.}
        {For each submodel, we {choose} two {different} parameter sets {with} contrasted values in the division rates $\alpha_2$ or $\gamma$ and/or transition rate $\beta$. The parameter values are summarized in Table~\ref{table-dataset-gen}. We obtain the corresponding 10 datasets by simulating  $1,000$ trajectories  from the SDE \eqref{eq-SDE}, with the Gillespie algorithm \cite{gillespie_general_1976}, starting from the initial condition $(F_0,0)$ at time $t=0$ up to the time when $C(t) = 31$ (the value $C(t) = 31$ corresponds to the maximal number of cuboidal cells observed in the experimental dataset). The initial  random variable $F_0$ follows a truncated Poisson law of parameter $\mu$ (see Eq.\eqref{eq-poisson-trun-law}). For each trajectory, we select uniformly randomly one point $(f,c)$ among the state space points reached by the trajectory, so that each \textit{in silico} datasets is composed of $N=1,000$ points.} {This way of sampling, letting to time-free and uncoupled datapoints, mimics the experimental protocol.}
	   
	   	\begin{table}[htb!]
	   			\centering
	   			\begin{tabular}{|c|c|c|c|c|c|}
	   				\hline 
	   				&  $\alpha_1$ &$\beta$ & $\alpha_2$ & $\gamma$ &  $\mu$ \\ 
	   				\hline 
	   				\hline \\[-2.ex]
	   				\multirow{2}{*}{$(\mathcal{R}_1,\mathcal{R}_3)$ } &  1 & 0 & 0.7 &0 & 5 \\[0.05cm]
	   							\vspace{0.05cm}																& 1 & 0  & 0.007  & 0 & 5 \\ 
	   				\hline 
	   				\hline \\[-2.ex]
	   				\multirow{2}{*}{$(\mathcal{R}_1,\mathcal{R}_4)$ } & 1 & 0  & 0  & 0.7  & 5 \\[0.05cm]  
	   								\vspace{0.05cm}		
	   								&  1 & 0  & 0  & 0.007  & 5 \\ 
	   				\hline 
	   				\hline \\[-2.ex]
	   				\multirow{2}{*}{$(\mathcal{R}_1, \mathcal{R}_2, \mathcal{R}_3)$ }	& 1 & 0.01  & 0.07  & 0 & 5 \\[0.05cm]  
	   									\vspace{0.05cm}		
	   									& 1 & 100  & 0.07  & 0 & 5 \\ 
	   				\hline 
	   				\hline \\[-2.ex]
	   				\multirow{2}{*}{$(\mathcal{R}_1, \mathcal{R}_3, \mathcal{R}_4)$ } 	& 1 & 0  & 0.007  & 0.7 & 5 \\[0.05cm]  
	   							\vspace{0.05cm}		
	   								& 1 & 0  & 0.007  & 0.07 & 5 \\ 
	   				\hline 
	   				\hline \\[-2.ex]
	   				\multirow{2}{*}{$(\mathcal{R}_1, \mathcal{R}_2, \mathcal{R}_4)$ } 	& 1 & 0.01  & 0  & 0.07 & 5 \\[0.05cm]  
	   				\vspace{0.05cm}		
	   					& 1 & 100  & 0  & 0.07 & 5 \\ 
	   				\hline 
	   			\end{tabular} 
   				\caption{\textbf{Parameter sets used to generate the \textit{in silico} datasets.} We use two distinct parameter sets for each submodel, shown in the two rows associated to each submodel.}\label{table-dataset-gen}
	   		\end{table}
	   		
	   	\subsection{Detailed fitting procedure}\label{appendix:fittingproc}
	   	
	   	\paragraph{Maximum  likelihood  estimator}
	   	    For each submodel and dataset,
	the optimal parameter values are given by the MLE $\hat{\theta} = \left(\widehat{\beta} , \widehat{\alpha_2} , \widehat{\gamma}, \widehat{\mu} \right) $, which we compute by minimizing the negative log-likelihood,  
	\begin{equation*}
		\hat{\theta} := \arg \min_{\theta \in \Theta} \left( - \log\left(\mathcal{L}(\mathbf{x};\theta)  \right)\right),
	\end{equation*}
    for a dataset $\mathbf{x}$ and where $\Theta$ is constructed by fixing all parameters related to the nonpresent events to the singleton $\{0\}$: for instance, in submodel $(\mathcal{R}_1,\mathcal{R}_4)$, we have $\Theta= \{0\} \times \{0\} \times \mathbb{R}_+ \times  [1, + \infty)$.
    
	   	    To compute the minimum, we use a derivative-free optimization algorithm: the Differential Evolution (DE) algorithm \cite{storn_differential_1997}. In the following, we describe the whole procedure for the complete model $ (\mathcal{R}_1,\mathcal{R}_2,\mathcal{R}_3,\mathcal{R}_4)$. The algorithm starts from an initial population in which each individual is represented by a set of real numbers $(\beta, \alpha_2, \gamma, \mu) $. Then, the population evolves along successive generations by mutation and recombination processes. At each generation, the likelihood function is used to assess the fitness of the individuals, and only the best individuals are kept in the population. 
	    We have set the intrinsic optimization parameters as follows: the initial population has a size of 20 individuals, and the probability of mutation and crossing-over equals to 0.8 and 0.7 respectively. The starting individual parameter sets are defined on a log scale, and drawn from a uniform distribution on $\Theta = [-6,6]^3 \times [0, 1.5] $.
	    The algorithm was run over 1,000 iterations. \\
	    
	   	\paragraph{Profile likelihood estimate}
	    	   	
	    \begin{table}[h]
	    \centering
			\begin{tabular}{|c|c ||c|c|c|}
				\hline 
				Model & Parameter & \shortstack{ Experimental \\ Wild-Type/Mutant \\ datasets}  & \shortstack{\textit{in silico}\\ Dataset $1$}  & \shortstack{\textit{in silico} \\ Dataset 2 } \\[0.05cm] 
				\hline 
				\hline
				\multirow{2}{*}{$(\mathcal{R}_1, \mathcal{R}_4) $}  & $\mu$ & 0.015 & 0.005 & 0.01 \\ 
				            & $\gamma$ & 0.12 & 0.01 & 0.06 \\ 
				\hline 
				\multirow{2}{*}{$(\mathcal{R}_1, \mathcal{R}_3) $}  & $\mu$ & 0.015 & 0.005 & 0.005 \\ 
				                                                    & $\alpha_2$ & 0.04 & 0.01 & 0.01 \\ 
				\hline 
				\multirow{2}{*}{$(\mathcal{R}_1, \mathcal{R}_2, \mathcal{R}_4) $}   & $\mu$ & 0.015 & 0.01 & 0.015 \\ 
				             & $\beta$ & 0.12 & 0.07 & 0.12 \\ 
				             & $\gamma$ & 0.12 & 0.07 & 0.12 \\ 
				\hline 
				\multirow{2}{*}{$(\mathcal{R}_1, \mathcal{R}_2, \mathcal{R}_3) $}   & $\mu$ & 0.015 & 0.015 & 0.015 \\ 
				         & $\beta$ & 0.12 & 0.12 & 0.12 \\ 
				         & $\alpha_2$ & 0.12 & 0.02 & 0.02 \\ 
				\hline 
				\multirow{2}{*}{$(\mathcal{R}_1, \mathcal{R}_3, \mathcal{R}_4) $ } & $\mu$ & 0.01 & 0.01 & 0.01   \\
				 & $\alpha_2$ & 0.08 & 0.01 & 0.01   \\
				  & $\gamma$  & 0.08 & 0.01 & 0.01   \\
				\hline 
				\multirow{2}{*}{$(\mathcal{R}_1, \mathcal{R}_2,\mathcal{R}_3, \mathcal{R}_4) $} & $\mu$ & 0.015 &  0.015 &  0.015   \\
			    	& $\beta$ & 0.12 & 0.12 & 0.12   \\
				    & $\alpha_2$ & 0.12 & 0.12 & 0.12   \\
				    & $\gamma$  & 0.12 & 0.12 & 0.12   \\
				\hline 
				\hline
			\end{tabular} 
			\caption{\textbf{Size-step} used for each parameter in the PLE estimate,  in log-scale, within each submodel and each datasets.}\label{table-stepsize-PLE}
		\end{table}
		
	   For each $i$th component of the MLE $\hat{\theta}_i$, $i \in \llbracket 1, 4 \rrbracket$,  we compute a vector $\hat{\theta}|[\theta_{i} = x]$ on a grid $G_i$ around the MLE $\hat{\theta}$, with $x \in G_i$:
	    \begin{equation*}
		\hat{\theta}| [\theta_{i} = x] := \arg \min_{\theta \in \Theta, \theta_{i} = x} \left( - \log\left(\mathcal{L}(\mathbf{x};\theta)  \right)\right),
	\end{equation*}
	and its associated PLE (vector)  $\mathcal{L}(\mathbf{x};\hat{\theta}|\theta_{i}) $. We design the grid $G_i$ around the MLE $\hat{\theta}_{i}$ with a fixed step size (see Table \ref{table-stepsize-PLE} for details), and re-optimize the remaining parameters using the DE algorithm with the same optimization parameters (mut=0.8, crossp=0.7, popsize=20, its = 1,000) and initial parameter sets defined on a log scale, and drawn from a uniform distribution on $[-6,6]^3$ for parameters $\beta$, $\alpha_2$ and $\gamma$, and on $[-1 + \log(\hat{\mu}), \log(\hat{\mu}) + 1]$  for parameter $\mu$.

\paragraph{Confidence intervals}

	Pointwise likelihood-based confidence intervals are constructed thanks to the likelihood ratio test, following \cite{raue_structural_2009} ; for each estimated parameter $\hat{\theta}_{i}$, we select all the parameters $\theta_{i} = x$ such that:
	\begin{equation*}
	   \mathcal{L}(\mathbf{x};\theta|[\theta_{i} =x]) - \mathcal{L}(x;\hat{\theta}) < 0.5 * \Delta_{\alpha} ,
	\end{equation*}
	where $\Delta_{0.95} = \chi^2(0.95, 1) = 3.84$ is the $0.95$-quantile of the $\chi^2$ law with $1$ degree of freedom. 
    
        \paragraph{Model selection.}
    {
            AIC and BIC analyses were performed to compare the submodels.
            The reader can refer to  \cite{burnham_model_2003} (Chapter 6) for a detailed presentation of the rule of thumb, classically used to analyze the $\Delta^{AIC}_i := AIC_i - AIC_{\min}$ and $\Delta^{BIC}_i = BIC_i - BIC_{\min}$ values, where $i$ is the index of the $i$th model:
            \begin{itemize}
            \item a $\Delta$ value lower than 2 indicates that  the considered model is almost as probable as the ``best'' model;
            \item a $\Delta$ value between 2 and 7 suggests that the considered model is a suitable alternative to the ``best'' model;
            \item a $\Delta$ value between 7 and 10 suggests that the considered model is less relevant than the ``best'' model;
            \item a $\Delta$ value upper than 10 suggests that the considered model can be safely ruled out.
            \end{itemize} 
            This $\Delta$ approach is completed by the AIC and BIC weight analyzes. For each dataset and  criterion (AIC or BIC), we order the AIC/BIC weights from the highest to the lowest values. We then compute the cumulative sum of these weights, starting from the highest one. The selected models are the first ones such that the cumulative sum reaches the threshold p-value $0.95$.
    }
    
	   	\subsection{Detailed calibration analysis} \label{appex:detail_cal}
	   	\paragraph{Two-event submodels}
	   	    {The fitting results obtained for submodels $(\mathcal{R}_1,\mathcal{R}_3)$ and $(\mathcal{R}_1,\mathcal{R}_4)$ from the experimental datasets are shown in Figure \ref{fig:bestfit} and discussed in the main text, Section \ref{ssec:fittingresuts}.}
			One fitting result for the \textit{in silico} datasets and for submodels $(\mathcal{R}_1,\mathcal{R}_3)$ and $(\mathcal{R}_1,\mathcal{R}_4)$ is shown in Figure \ref{fig:datafitt_ana_insilico}. We verify that the inferred trajectories are coherent with the selected datasets.

			In Figures \ref{fig:datafitt_2parms_insilico}, we show the PLE for each estimated parameter in each \textit{in-silico} dataset. Both the initial condition parameter $\mu$ (orange solid lines) and asymmetric division rate $\alpha_2 $ (green solid line) are practically identifiable  (in the sense given in \cite{raue_structural_2009}), while parameter $\gamma$ (blue solid line) is only partially practically identifiable in most cases. We observe that both parameters $\alpha_2$ ($\mathcal{R}_3$) and $\gamma$ ($\mathcal{R}_4$) are practically identifiable and close to their expected values (less than one $log10$ of difference) when the parameters are of the same order of magnitude than $\alpha_1$. In contrast, a small parameter value compared to $\alpha_1$ leads to a biased parameter estimate, with a huge shift between the estimated and true parameter values (up to two $log10$ difference).\\ {The estimator for the initial condition parameter $\mu$ may also be slightly biased with submodel $(\mathcal{R}_1,\mathcal{R}_3)$ (less than one $log10$ of difference) compared to submodel $(\mathcal{R}_1,\mathcal{R}_4)$ .}

	 \begin{figure}[htb!]
		\centering
	\includegraphics[width=\textwidth]{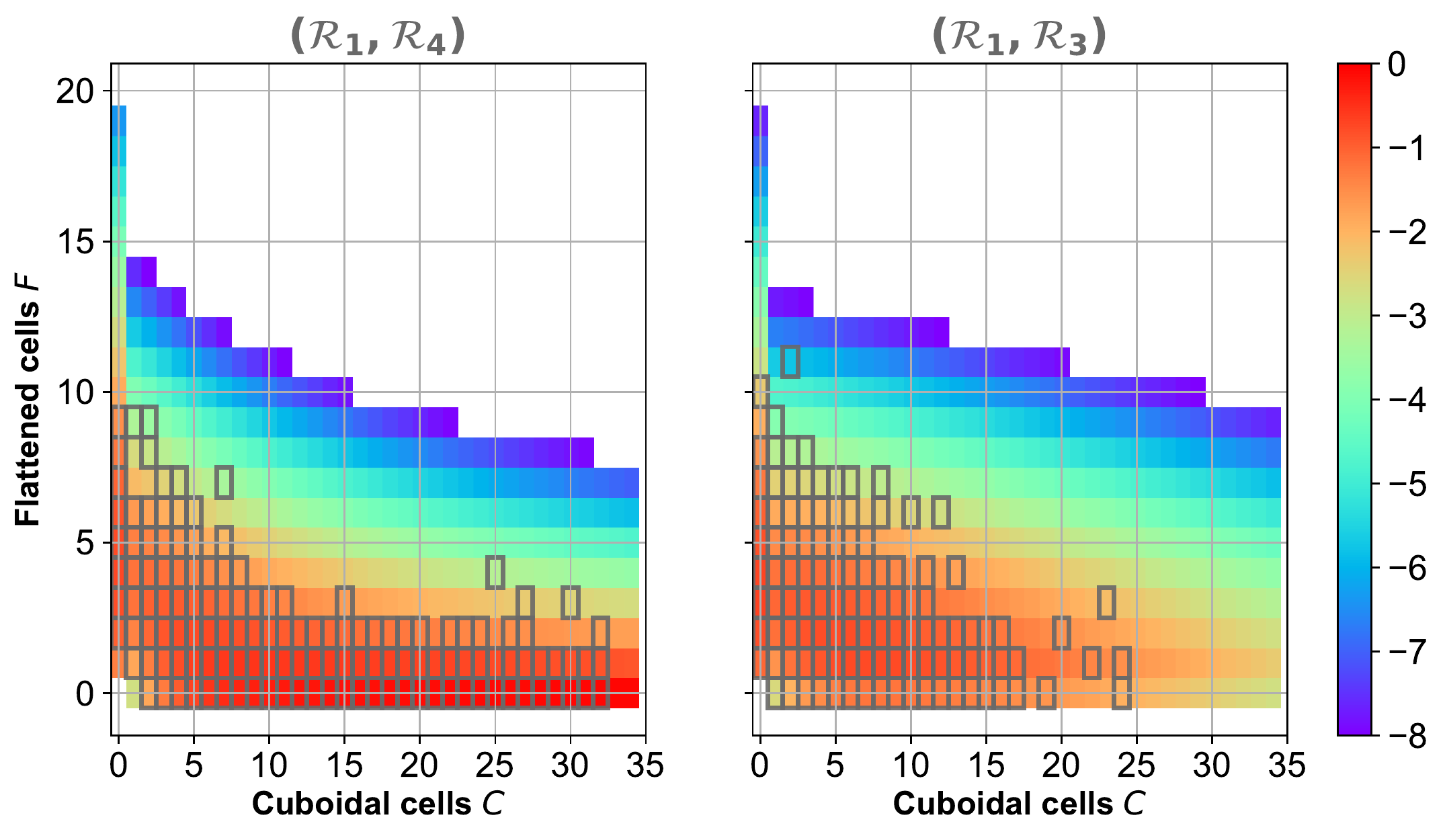}	\caption{\textbf{Two-event submodels: Best fit trajectories for \textit{in silico} datasets.} Using Eqs.\eqref{eq-like}-\eqref{eq-poisson-trun-law}, we compute each probability $\mathbb{P}\left[F_c = f\right] $ for submodel $(\mathcal{R}_1, \mathcal{R}_4)$ (left panel) and $(\mathcal{R}_1, \mathcal{R}_3)$ (right panel) with their respective MLE parameter set associated to the \textit{in silico dataset 1}. Each dark gray square corresponds to a data point. {The colormap corresponds to the probability values $\mathbb{P}\left[F_c = f\right] $ in log10 scale.} }
		\label{fig:datafitt_ana_insilico}
	\end{figure}
	
		 \begin{figure}[htb!]
	 	 \centering
	 	 \includegraphics[width=\linewidth]{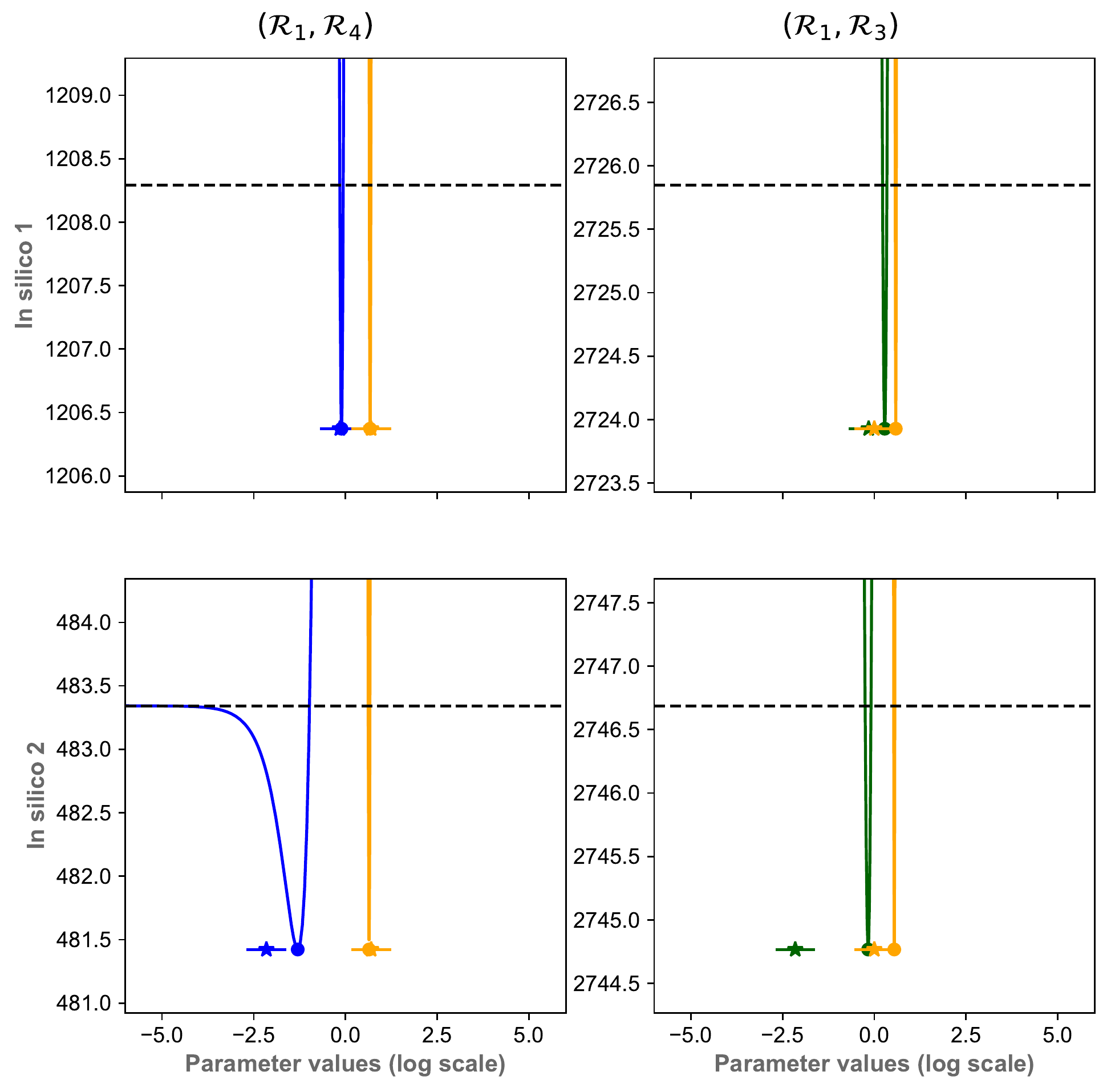}
	 	\caption{\textbf{Two-event submodels: PLE for \textit{in silico} datasets.}  Each panel represents the PLE, in log10 scale, obtained from the \textit{in silico} datasets, and either submodel $(\mathcal{R}_1,\mathcal{R}_4)$ (left panels) or $(\mathcal{R}_1,\mathcal{R}_3)$ (right panels). The dashed black line represents the 95\%-statistical threshold. Orange solid lines: PLE values for the initial condition parameter $\mu$; blue solid lines: PLE values for the symmetric cell proliferation rate $\gamma$; green solid lines: PLE values for the asymmetric cell division rate  $\alpha_2$.
				 The colored points represent the associated MLE, and the star symbols are the expected (true) parameter values (see Table \ref{table-dataset-gen}).}
	 	\label{fig:datafitt_2parms_insilico}
	 \end{figure}

	\paragraph{Three-event submodels and complete model}
	   We turn now to the analysis of three-event submodels $(\mathcal{R}_1,\mathcal{R}_2,\mathcal{R}_3)$, $(\mathcal{R}_1,\mathcal{R}_2,\mathcal{R}_4)$ and $(\mathcal{R}_1,\mathcal{R}_3,\mathcal{R}_4)$) and the complete model ($(\mathcal{R}_1,\mathcal{R}_2,\mathcal{R}_3,\mathcal{R}_4)$. 
	   Qualitatively, the fitting results for submodel $(\mathcal{R}_1,\mathcal{R}_2,\mathcal{R}_3)$ are similar to those for submodel $(\mathcal{R}_1,\mathcal{R}_3)$ (data not-shown); they are characterized by a high probability to produce ten or more proliferative cells before the precursor cell extinction. The fitting results for submodels $(\mathcal{R}_1,\mathcal{R}_2,\mathcal{R}_4)$ and $(\mathcal{R}_1,\mathcal{R}_3,\mathcal{R}_4)$ are rather similar to submodel\\ $(\mathcal{R}_1,\mathcal{R}_4)$; they are characterized by direct cell transition with very little concomitant cell proliferation, followed by prolonged cell proliferation after precursor cell extinction. The fitting results for the complete model are shown in the bottom panels of Figure \ref{fig:bestfit} for both the Wild-type and Mutant subsets {and discussed in the main text, Section \ref{ssec:fittingresuts}.}
	   
        \begin{figure}
				\centering
	       	    \includegraphics[width=\linewidth]{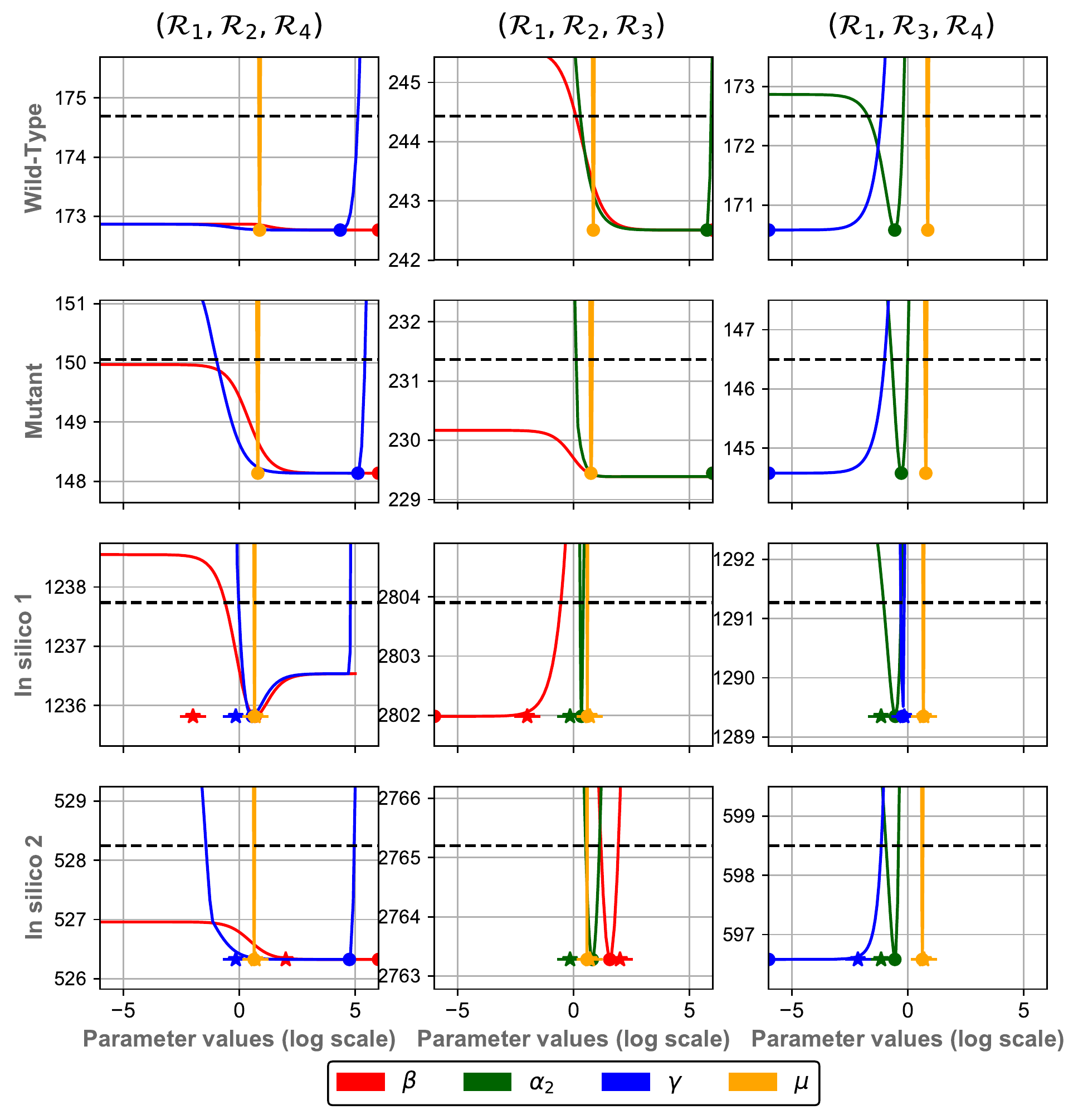}
				\caption{\textbf{Three-event submodels: PLE}. Each panel represents the PLE, in log10 scale, obtained from the experimental (top panels) and \textit{in silico} datasets (bottom panels), and either submodel $(\mathcal{R}_1,\mathcal{R}_2,\mathcal{R}_4)$ (left panels), $(\mathcal{R}_1,\mathcal{R}_2,\mathcal{R}_3)$ (center panels), or $(\mathcal{R}_1,\mathcal{R}_3),\mathcal{R}_4)$ (right panels).
				The dashed black line represents the 95\%-statistical threshold.  Orange solid lines: PLE values for the initial condition parameter $\mu$; blue solid lines: PLE values for the symmetric cell proliferation rate $\gamma$; green solid lines: PLE values for the asymmetric cell division rate  $\alpha_2$; red solid lines: PLE values for the self-amplification transition rate $\beta$.
				 The colored points represent the associated MLE, and (in the bottom panels) the star symbols are the expected (true) parameter values (see Table \ref{table-dataset-gen}).}
				\label{fig:datafitt_3parm}
			\end{figure}
			
	   The PLEs for each dataset and each parameter are presented in Figure \ref{fig:datafitt_3parm} for the three-event submodels. The corresponding parameter values and confidence intervals for the Wild-Type and Mutant subsets are given in Tables \ref{table-wt} and \ref{table-mut}. As observed for the two-event submodels, in each case, the initial condition parameter $\mu$ (orange solid lines) is always practically identifiable, and its fitted value is close to the true one for the \textit{in silico} datasets. In contrast, all other parameters have a lack of identifiability, {both with the experimental and \textit{in silico} datasets}. Specifically, the asymmetric division rate $\alpha_2$ is practically not identifiable for submodel  $(\mathcal{R}_1,\mathcal{R}_2,\mathcal{R}_3)$ with the experimental subsets. Interestingly, when the asymmetric division event is combined with the symmetric division event (submodel $(\mathcal{R}_1,\mathcal{R}_3,\mathcal{R}_4)$) rather than with the auto-amplified transition  (submodel $(\mathcal{R}_1,\mathcal{R}_2,\mathcal{R}_3)$), the asymmetric division rate {$\alpha_2$} becomes identifiable in the experimental subsets, which reveals complex parameter dependencies between the asymmetric division rate $\alpha_2$ and auto-amplified transition rate $\beta$. 
			
	\begin{table}[htb!]
	\centering
			\begin{tabular}{|c|c|c|c|c|}
				\hline 
				Model & $\beta$  & $\alpha_2 $  & $\gamma $  & $\mu$  \\ 
				\hline 
				\hline
				$(\mathcal{R}_1, \mathcal{R}_4) $  & / & / & \shortstack{$10^{-6}$ \\  $ \in (0; 0.12]$} & \shortstack{$7.49$ \\ $\in [7.05;7.83]$} \\ 
				\hline 
				$(\mathcal{R}_1, \mathcal{R}_3) $ &  / & \shortstack{$1.18$ \\ $ \in [0.67; 1.57]$}  & /  & \shortstack{$ 7.22 $ \\ $ \in [6.81;7.83]$ } \\ 
				\hline 
				$(\mathcal{R}_1, \mathcal{R}_2, \mathcal{R}_4) $  & $10^6 \in \mathbb{R}$ & / & \shortstack{$10^{4.35}$ \\ $\in (0; 10^{5.03}]$} & \shortstack{$7.45$ \\ $ \in [7.05;7.83] $} \\ 
				\hline 
				$(\mathcal{R}_1, \mathcal{R}_2, \mathcal{R}_3) $ & \shortstack{$10^6$ \\ $\in [1.52; + \infty)$} & \shortstack{$10^{5.75}$ \\ $\in [2.00; 10^{5.88}]$} & / & \shortstack{$7.07$ \\ $\in [5.15;6.35]$} \\ 
				\hline 
				$(\mathcal{R}_1, \mathcal{R}_3, \mathcal{R}_4) $ & / & \shortstack{$0.27$ \\ $\in [0.022;0.52]$} & \shortstack{$10^-6$ \\ $\in (0;0.068]$} & \shortstack{$7.20$ \\ $\in [6.69;7.69]$} \\ 
				\hline 
				$(\mathcal{R}_1, \mathcal{R}_2,\mathcal{R}_3, \mathcal{R}_4) $ & \shortstack{$10^6$ \\ $\in [4.64; + \infty)$ } & \shortstack{$10^{4.78}$ \\ $\in [0.87; 10^{5.27}]$ } & \shortstack{$10^{-6}$ \\ $\in (0; 10^{4.67}]$ } & \shortstack{$7.06$ \\ $\in [6.58;7.56]$ }   \\ 
				\hline 
				\shortstack{Primordial follicle \\ dataset} & / & / & / & \shortstack{$6.22$ \\ $\in [5.54;6.67]$} \\
				\hline
				\hline
			\end{tabular} 
			\caption{Wild-Type MLE parameter sets. MLE estimates and confidence intervals for each submodel using the likelihood given by Eqs.\eqref{eq-like}-\eqref{eq:likelihood_total}, and (last row) for the initial condition parameter using likelihood given by Eq.~\eqref{eq:likeli_init}.}\label{table-wt}
		\end{table}
		\vspace{-0.5cm}
		\begin{table}[htb!]
		\centering
			\begin{tabular}{|c|c|c|c|c|}
				\hline 
				Model & $\beta$  & $\alpha_2 $  & $\gamma $  & $\mu$  \\ 
				\hline 
				\hline 
				$(\mathcal{R}_1, \mathcal{R}_4) $  & / & / & \shortstack{0.14 \\ $\in (0;0.28]$ } & \shortstack{6.40 \\ $\in [5.93;6.81]$ }\\ 
				\hline 
				$(\mathcal{R}_1, \mathcal{R}_3) $ &  / & \shortstack{1.63 \\ $\in [1.26;2.20]$ } & /  & \shortstack{5.91 \\ $\in [5.34;6.35]$ } \\ 
				\hline 
				$(\mathcal{R}_1, \mathcal{R}_2, \mathcal{R}_4) $ & \shortstack{$10^6$ \\ $\in \mathbb{R}$ } & / & \shortstack{$10^{5.11}$ \\ $\in [0.12; 10^{5.39}]$ } & \shortstack{6.26 \\ $\in [5.72; 6.81]$ }\\ 
				\hline 
				$(\mathcal{R}_1, \mathcal{R}_2, \mathcal{R}_3) $& \shortstack{$10^6$ \\ $\in \mathbb{R}$ }& \shortstack{$10^{6}$ \\ $\in [1.52; + \infty)$} & / & \shortstack{ 5.57 \\ $\in [5.15; 6.35]$ }  \\ 
				\hline 
				$(\mathcal{R}_1, \mathcal{R}_3, \mathcal{R}_4) $& / & \shortstack{$0.52$ \\ $\in [0.21; 0.91]$} & \shortstack{$10^{-6}$ \\ $\in (0,0.98]$ } & \shortstack{ 5.94 \\ $\in [5.43; 6.54]$ }  \\ 
				\hline 
				$(\mathcal{R}_1, \mathcal{R}_2,\mathcal{R}_3, \mathcal{R}_4) $ & \shortstack{$2.81$ \\ $\in \mathbb{R}_+$ }  & \shortstack{$1.16$ \\ $\in [0.28;10^{5.51}]$ }  & \shortstack{$10^{-6}$ \\ $\in (0; 10^{4.9}]$ }  & \shortstack{$5.83$ \\ $\in [ 5.15; 6.35]$ }   \\ 
				\hline 
				\shortstack{Primordial follicle \\ dataset}  & / & / & / & \shortstack{$6.77$ \\ $\in [5.75;7.60]$} \\
				\hline
				\hline 
			\end{tabular} 
			\caption{Mutant parameter sets. MLE estimates and confidence intervals for each submodel using the likelihood given by Eqs.\eqref{eq-like}-\eqref{eq:likelihood_total}, and (last row) for the initial condition parameter using likelihood given by Eq.~\eqref{eq:likeli_init}.}\label{table-mut}
		\end{table}

	\begin{figure}[htb!]
		\centering
		\includegraphics[width=\linewidth]{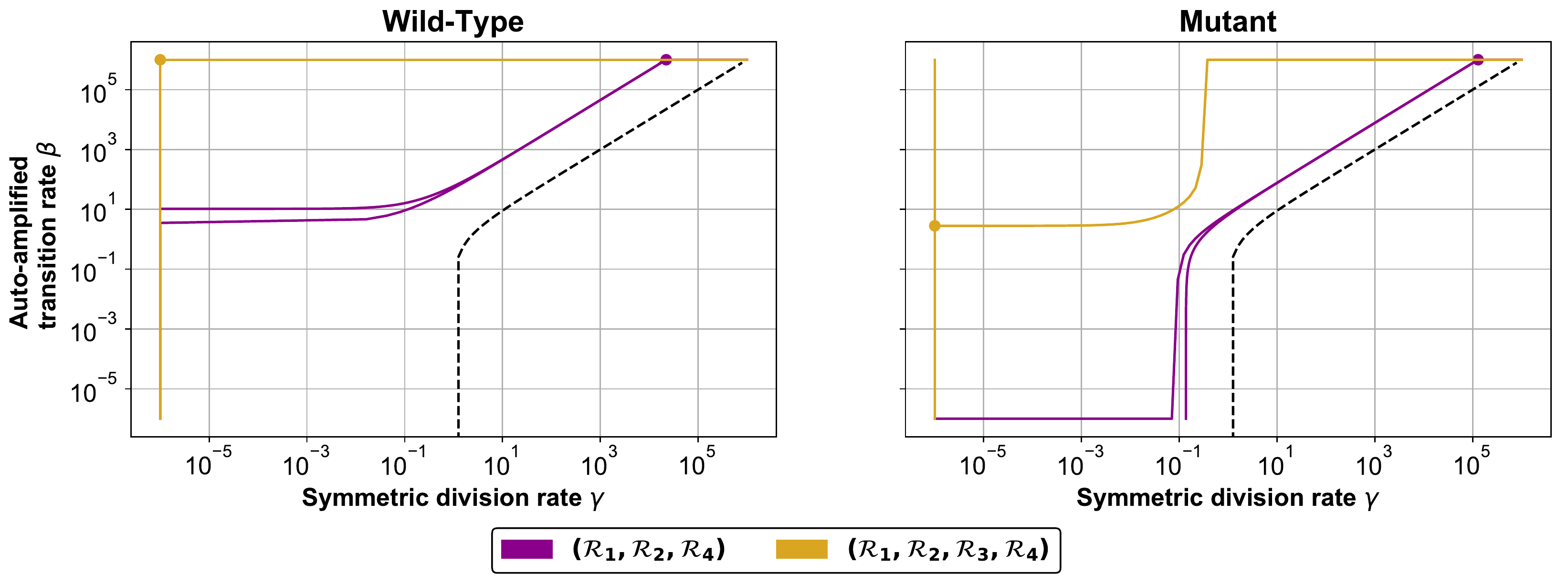}
		\caption{\textbf{Proliferation versus transition}. For the Wild-Type (left panel) and Mutant (right panel) datasets, and for submodel $(\mathcal{R}_1,\mathcal{R}_2,\mathcal{R}_4)$ and complete model $(\mathcal{R}_1,\mathcal{R}_2,\mathcal{R}_3,\mathcal{R}_4)$,  we represent in colored lines both the optimal value of self-amplification transition rate $\beta$ along the PLE of the symmetric cell proliferation rate $\gamma$, and the optimal value of the symmetric cell proliferation rate $\gamma$ along the PLE of self-amplification transition rate $\beta$. In black dashed line, we represent the straight line $\gamma=\beta+\alpha_1=\beta+1$.}
		\label{fig:ple_fit_gamma_beta1}
	\end{figure}

\clearpage

\begin{acknowledgements}
	The authors wish to thank Ken McNatty for providing the experimental dataset and Danielle Monniaux for helpful discussions.
\end{acknowledgements}

\bibliographystyle{spmpsci}      
\bibliography{Condition_Initiale}   

%
%

\end{document}